\documentclass[11pt]{amsart}
\usepackage[T1]{fontenc}
\usepackage[british]{babel}
\IfFileExists{lmodern.sty}{\usepackage{lmodern}}{}
\usepackage{microtype}
\IfFileExists{lmodern.sty}{}{\microtypesetup{expansion=false}}
\usepackage{mathtools,amssymb,mathrsfs}
\usepackage{booktabs}
\usepackage[numbers,sort&compress]{natbib}

\usepackage[margin=1in]{geometry}
\usepackage[hidelinks]{hyperref}
\usepackage{aliascnt}
\usepackage[capitalise,nameinlink,noabbrev]{cleveref}
\crefname{subsection}{Section}{Sections}
\Crefname{subsection}{Section}{Sections}

\newcommand{\R}{\mathbb R}
\newcommand{\C}{\mathbb C}
\newcommand{\N}{\mathbb N}
\newcommand{\Z}{\mathbb Z}
\newcommand{\cH}{\mathcal H}
\newcommand{\cA}{\mathcal A}
\newcommand{\cE}{\mathcal E}
\newcommand{\cL}{\mathcal L}
\newcommand{\Tr}{\operatorname{Tr}}
\newcommand{\EP}{\operatorname{EP}}
\newcommand{\Ent}{\operatorname{Ent}}
\newcommand{\Sp}{\operatorname{Sp}}

\newcommand{\fhat}{\widehat f_{\!\mathscr M}}
\newcommand{\dd}{\,\mathrm d}
\newcommand{\doi}[1]{\href{https://doi.org/#1}{doi:\nolinkurl{#1}}}

\theoremstyle{plain}
\newtheorem{theorem}{Theorem}[section]
\newaliascnt{proposition}{theorem}
\newtheorem{proposition}[proposition]{Proposition}
\aliascntresetthe{proposition}
\crefname{proposition}{Proposition}{Propositions}
\Crefname{proposition}{Proposition}{Propositions}
\newaliascnt{lemma}{theorem}
\newtheorem{lemma}[lemma]{Lemma}
\aliascntresetthe{lemma}
\crefname{lemma}{Lemma}{Lemmas}
\Crefname{lemma}{Lemma}{Lemmas}
\newaliascnt{corollary}{theorem}
\newtheorem{corollary}[corollary]{Corollary}
\aliascntresetthe{corollary}
\crefname{corollary}{Corollary}{Corollaries}
\Crefname{corollary}{Corollary}{Corollaries}
\theoremstyle{definition}
\newaliascnt{definition}{theorem}
\newtheorem{definition}[definition]{Definition}
\aliascntresetthe{definition}
\crefname{definition}{Definition}{Definitions}
\Crefname{definition}{Definition}{Definitions}
\newaliascnt{assumption}{theorem}

\aliascntresetthe{assumption}
\crefname{assumption}{Assumption}{Assumptions}
\Crefname{assumption}{Assumption}{Assumptions}
\theoremstyle{remark}
\newaliascnt{remark}{theorem}
\newtheorem{remark}[remark]{Remark}
\aliascntresetthe{remark}
\crefname{remark}{Remark}{Remarks}
\Crefname{remark}{Remark}{Remarks}
\numberwithin{equation}{section}

\title[Hypercontractive criteria for entropy decay]{Hypercontractivity and entropy decay for infinite-dimensional quantum Gibbs samplers}
\author{Simon Becker}
\address{Bocconi University, Milan, Italy}
\email{simon.becker@unibocconi.it}
\author{Cambyse Rouz\'e}
\address{Inria, T\'el\'ecom Paris--LTCI, Institut Polytechnique de Paris, 91120 Palaiseau, France}
\email{cambyse.rouze@inria.fr}
\author{Robert Salzmann}
\address{Department of Physics, RWTH Aachen, Otto-Blumenthal-Strasse 20, 52074 Aachen, Germany}
\email{robert.salzmann@rwth-aachen.de}


\begin{document}
\begin{abstract}
We derive rates of exponential relative-entropy decay for a broad class of infinite-dimensional quantum Gibbs samplers. A key technical ingredient is a perturbative stability theory for hypercontractive estimates and logarithmic Sobolev inequalities. For suitably chosen dissipative couplings, we obtain exponential entropy decay for full multimode Bose--Hubbard samplers with arbitrary hopping and repulsive density interactions. For harmonically confined quantum Coulomb gases and molecular systems, we prove exponential entropy decay for every fixed finite-rank low-energy truncation of the interaction. Finally, we exhibit a periodic charging model with a positive spectral gap but no finite-time $L^2$-to-$L^p$ bound for any $p>2$.
\end{abstract}

\maketitle
\enlargethispage{4pt}

\section{Introduction and main results}

Thermal-state preparation is a basic task in quantum statistical mechanics
and quantum computation.  A Gibbs sampler drives an initial state towards thermal equilibrium
through a dissipative evolution.  A spectral gap gives exponential decay of
deviations from equilibrium in a weighted $L^2$ norm.  In infinite
dimension, this estimate requires initial data in the weighted $L^2$
space.  Relative entropy measures distinguishability from equilibrium
and, for a Gibbs state, equals the excess free energy divided by the
temperature whenever these quantities are finite.  Exponential
relative-entropy estimates can also apply to initial states outside the
weighted $L^2$ space and control convergence in trace distance.

Boundedness of the semigroup from $L^2$ to $L^p$, $p>2$, connects $L^2$ relaxation with relative-entropy decay. A finite $L^2$-to-$L^p$ bound is referred to as hyperboundedness, while a bound with operator norm at most one is the contractive form of hypercontractivity.
Both express an improvement in integrability after a finite time.
This smoothing property is not implied by a spectral gap alone.  In classical probability, the
connection between hypercontractivity, logarithmic Sobolev inequalities,
and convergence to equilibrium has been studied extensively since
Gross's foundational work; see \cite{ledoux2004logarithmic}.
These inequalities give estimates for diffusion and Markov-chain mixing.
For reversible quantum dynamics, we show that one bounded
$L^2$-to-$L^q$ map gives an explicit exponential
relative-entropy bound. In the atomic setting considered here, this
endpoint already implies compact resolvent and, under primitivity, a
positive spectral gap. Its value enters the rate when the endpoint norm
exceeds one.

Efficient finite-dimensional quantum Gibbs samplers satisfying exact KMS detailed balance were developed in \cite{ChenKastoryanoGilyen2023,DingLiLin2025}, and the infinite-dimensional construction of \cite{BeckerRouzeSalzmannToAppearcmp} extends this framework to systems governed by unbounded Hamiltonians and jump operators.
We study when these samplers improve integrability and how this
gives entropy decay. We prove sufficient conditions and examples of
failure, showing what must be added to a spectral-gap estimate.  The behavior depends on the choice of jumps, even when the
Hamiltonian is fixed.

\subsection{Main results}\label{subsec:main-results}
We first state two criteria that turn a bounded $L^2$-to-$L^q$ map
into exponential relative-entropy decay.  A state is a positive trace-class operator of trace
one; it is faithful if its kernel is zero.  Let $\Phi_t$, $t\ge0$, be a
strongly continuous semigroup of completely positive, trace-preserving
linear maps on trace-class operators, and let $\sigma$ be a faithful
invariant state.  Write $T_t$ for the dual evolution of observables.
For $1\le p<\infty$, the equilibrium-weighted space $L^p(\sigma)$ is
the completion of bounded observables in the norm
\[
 \|X\|_{p,\sigma}
 =\|\sigma^{1/(2p)}X\sigma^{1/(2p)}\|_p,
\]
where $\|\cdot\|_p$ is the Schatten $p$-norm.
The space $L^\infty(\sigma)$ consists of bounded observables with the
operator norm.  We assume that $T_t$ extends consistently to positive
contractions on these spaces, is self-adjoint on $L^2(\sigma)$, and
has only scalar observables in its $L^2$ fixed space.  These are the
reversibility and primitivity assumptions used below; reversibility
is in the Kubo--Martin--Schwinger (KMS) sense.

With $\Tr$ denoting the operator trace, the spectral gap $\lambda$ is
the optimal exponential contraction rate on observables
$X\in L^2(\sigma)$ of equilibrium mean zero:
\[
 \|T_tX\|_{2,\sigma}\le e^{-\lambda t}\|X\|_{2,\sigma},
 \qquad \Tr(\sigma X)=0.
\]
For $X\in L^2(\sigma)$, the mean $\Tr(\sigma X)$ denotes its
continuous extension from bounded observables.
The gap may vanish in infinite dimension.  For a state $\rho$, its
relative entropy is
\[
 D(\rho\|\sigma)=\Tr\bigl(\rho(\log\rho-\log\sigma)\bigr)
 \in[0,\infty],
\]
interpreted in the extended sense when the separate traces are not
finite.  The spaces and operator domains are specified in \cref{sec:proof-main}.

\Cref{thm:main} uses the Stroock--Varopoulos comparison: entropy
production is at least four times the Dirichlet energy of the square
root of the state. The domain requirement is included in
\eqref{eq:SV-assumption} below.

\begin{theorem}\label{thm:main}
In the reversible setting above, assume the Stroock--Varopoulos
comparison \eqref{eq:SV-assumption}.  Suppose that, for some $s>0$
and $C\ge1$,
\begin{equation}\label{eq:24-main}
 \|T_s:L^2(\sigma)\to L^4(\sigma)\|\le C.
\end{equation}
If $C>1$, let $\lambda>0$ denote the spectral gap, which is
positive by the endpoint bound and primitivity. Then every state $\rho$
satisfies, for $t\ge0$,
\[
 D(\Phi_t(\rho)\|\sigma)
 \le e^{-\alpha_{\rm HC}t}D(\rho\|\sigma),
 \qquad
 \alpha_{\rm HC}
 =\left(s+\frac{\log C+\frac12}{\lambda}\right)^{-1}.
\]
If $C=1$, the same conclusion holds with $\alpha_{\rm HC}=s^{-1}$
without a gap assumption.
\end{theorem}

The comparison can fail for the samplers considered here.  The following
result extends the finite-dimensional endpoint criterion of
\cite[Theorem~3.3]{GaoWang2026Entropy} to the present setting, with the
same entropy-decay rate after normalization.

\begin{theorem}\label{thm:KMS-endpoint}
In the same reversible setting, let $q_0>2$ and suppose that, for some
$s>0$ and $C\ge1$,
\[
 \|T_s:L^2(\sigma)\to L^{q_0}(\sigma)\|\le C.
\]
If $C>1$, let $\lambda>0$ denote the spectral gap, which is
positive by the endpoint bound and primitivity. Then every state $\rho$
satisfies, for $t\ge0$,
\begin{equation}\label{eq:KMS-entropy-decay}
 D(\Phi_t(\rho)\|\sigma)
 \le e^{-\alpha_{\rm KMS}t}D(\rho\|\sigma),
 \qquad
 \alpha_{\rm KMS}
 =\frac{1-2/q_0}{s+\lambda^{-1}(\log C+1-2/q_0)}.
\end{equation}
If $C=1$, one may take $\alpha_{\rm KMS}=(1-2/q_0)/s$ without a gap.
\end{theorem}

In \cref{sec:perturbations}, we prove two stability results.
First, adding a bounded operator to the generator preserves the
$L^p$ bounds along a hypercontractive curve, provided its norm is
uniformly bounded on those spaces. We give the resulting norm estimate.
Second, a change of the Gibbs weight on a finite-dimensional energy
subspace, apart from normalization, preserves a tight $2$-logarithmic
Sobolev inequality when the Dirichlet forms are comparable and the new
generator has a positive gap.

The main applications are Bose--Hubbard and trapped Coulomb samplers.
In both cases we prove the smoothing estimate and the positive gap
needed by \cref{thm:KMS-endpoint}.

\paragraph{\emph{Bose--Hubbard models.}}
On a fixed lattice of $k$ modes, \cref{thm:multimode-BH} treats
\[
 H=\sum_{j=1}^k\left[\frac{U_j}{2}N_j(N_j-I)-\mu_jN_j\right]
 +\sum_{i<j}V_{ij}N_iN_j
 -\sum_{i<j}\left(J_{ij}a_i^\dagger a_j+
 \overline{J_{ij}}a_j^\dagger a_i\right),
\]
where $a_j,a_j^\dagger$ are the mode operators,
$N_j=a_j^\dagger a_j$, $U_j>0$, $V_{ij}\ge0$,
$\mu_j\in\R$, and $J_{ij}\in\C$.
With the bare jumps $a_j,a_j^\dagger,a_j^3,(a_j^\dagger)^3$ in
each mode and the filter \eqref{eq:filter}, the sampler has a unique
invariant Gibbs state, a positive spectral gap, and a bounded map
$L^2(\sigma_\beta)\to L^p(\sigma_\beta)$ for every $t>0$ and
finite $p>2$. It therefore satisfies
\[
 D(\Phi_t(\rho)\|\sigma_\beta)
 \le e^{-\alpha t}D(\rho\|\sigma_\beta),\qquad \alpha>0.
\]
The theorem allows arbitrary hopping amplitudes and uses the full
bosonic space, without an occupation cutoff. The cubic jumps provide
losses of order $n^3$ on the sector with $n$ particles; these dominate
the quadratic energy and the errors caused by hopping. The same theorem gives a bounded endpoint after a positive waiting
time with the smaller jump family $a_j,a_j^\dagger,a_j^2,(a_j^\dagger)^2$.
The rate may depend on the lattice size and all model parameters.

\paragraph{\emph{Trapped Coulomb systems.}}
For $n$ particles in $d=2$ or $3$ dimensions, \cref{cor:coulomb-hc}
considers the Hamiltonian
\[
 H_{n,M}=-\Delta+|x|^2+
 \Pi_M\left(\sum_{i<j}\alpha_{n,i,j}w_d(x_i-x_j)\right)\Pi_M,
 \qquad w_2(y)=-\log|y|,\quad w_3(y)=|y|^{-1},
\]
where $\Pi_M=P_M^{\otimes n}$ and $P_M$ projects onto the first $M$
one-particle Hermite functions, and the couplings $\alpha_{n,i,j}$
are real. Only the interaction is truncated;
the Hamiltonian acts on $L^2((\R^d)^n)$. With coordinate creation
and annihilation jumps and the Gaussian filter
$\widehat f_\eta(\nu)=e^{-\eta\nu^2-\beta\nu/4}$, $\eta>0$, the sampler
has a unique invariant Gibbs state and a positive gap. For each fixed
$n,M$, we prove a contractive $L^2$-to-$L^4$ bound after a finite
time and exponential relative-entropy decay. The proof controls the
terms containing one ladder operator and one finite-energy coefficient
by a weighted oscillator estimate. The constants can depend on $n$
and $M$.

The other examples explain when such smoothing estimates can fail.
For one-mode polynomial number Hamiltonians,
\cref{thm:mixed-polynomial} compares the degree of the Hamiltonian
with the largest jump order. Bounded $L^2$-to-$L^p$ maps then occur
at every positive time, after a positive waiting time, or never;
the stated classification excludes the boundary time in the critical
case. A periodic charging model has a positive gap but no such bound
at any finite time.

\subsection{Relation to earlier work}
Our results build on finite-dimensional entropy criteria and on the
infinite-dimensional Gibbs samplers constructed in earlier work.
The finite-dimensional version of \cref{thm:KMS-endpoint}
is due to Gao and Wang
\cite[Theorem~3.3]{GaoWang2026Entropy}.  Their entropy-decay rate agrees
with ours after converting their normalization
$D(\Phi_t(\rho)\|\sigma)\le e^{-2\alpha_1t}D(\rho\|\sigma)$.
They combine weighted interpolation, contraction of centred
observables, and a norm-to-entropy estimate.  We extend their criterion to infinite dimension. We approximate
arbitrary states while preserving the limit of the relative entropy,
and prove the required $L^2$-to-$L^q$ bounds for our unbounded models.
Related finite-dimensional results for primitive KMS-symmetric
semigroups appear in \cite{LiuWanWu2026Entropic,WanWu2026Cheeger}.

Noncommutative hypercontractivity and its connection with logarithmic
Sobolev inequalities go back to the weighted $L^p$ theory of
\cite{OlkiewiczZegarlinski1999}; see also
\cite{bardet2017estimating,beigi2020quantum}.  Complete modified
logarithmic Sobolev inequalities and their tensor-stable
finite-dimensional forms were developed in \cite{gao2022complete},
while perturbative change-of-measure principles appear in
\cite{junge2023stability}.  The equivalence between modified
logarithmic Sobolev inequalities and exponential relative-entropy
decay in the von Neumann algebra setting is treated in
\cite{wirth2025exponential}.  Here one $L^2$-to-$L^q$ bound in infinite dimension gives an
explicit entropy-decay rate, involving the automatically positive
spectral gap when the bound has norm greater than one.  One argument uses the
Stroock--Varopoulos comparison; the other uses only the KMS
reversibility assumptions above.  The counterexample in
\cref{prop:generator-SV-counterexample} shows why the KMS argument is needed for general cross-frequency generators.

Our earlier work \cite{BeckerRouzeSalzmannToAppearcmp} provides the
infinite-dimensional Gibbs-sampler construction used here.  In \cite{BeckerRouzeSalzmannBose}, we established
spectral-gap results for mean-field and full multimode
Bose--Hubbard samplers and used them to obtain thermal-state
preparation guarantees.  Here we prove the additional smoothing estimates needed for
exponential decay from every initial state of finite relative entropy.  This requires control of the jump operators
as well as the Hamiltonian.  For example, for $H=N^2$ and the filter
used in \cref{subsec:mixed-polynomial}, \cref{thm:mixed-polynomial}
shows that linear jumps alone have no bounded $L^2$-to-$L^4$
endpoint, whereas adding cubic jumps gives such a bound at every
positive time and yields the entropy estimate in
\cref{thm:multimode-BH}. Quadratic jumps already give an endpoint
after a positive waiting time, including with hopping; cubic jumps
give the every-positive-time conclusion.
The theorem does not give the same entropy estimate for the
linear-jump samplers of \cite{BeckerRouzeSalzmannBose}.

The Coulomb result complements the fixed-cutoff spectral-gap
analysis of \cite{BeckerRouzeSalzmannCoulomb2026}.  We prove a weighted
oscillator estimate for the perturbed semigroup, controlling the mixed
terms that contain one unbounded ladder operator.  This gives an $L^2$-to-$L^4$ bound directly. The form comparison
in Section~3 also preserves a tight $2$-logarithmic Sobolev inequality,
but obtaining hypercontractivity from that inequality would require
an additional regularity assumption.

\Cref{sec:proof-main} sets up the spaces and proves the entropy-decay criteria.  \Cref{sec:perturbations} proves the stability results, and \cref{sec:applications} treats the bosonic examples
and obstructions.

\subsection{AI statement} We used OpenAI’s GPT-5.6 Sol and GPT-6 Astra to help identify assumptions under which generators associated with bosonic Hamiltonians, particularly the Bose–Hubbard Hamiltonian, are hypercontractive. We also used these tools to assist with manuscript preparation. We reviewed and verified the AI-assisted material and take full responsibility for the content of this work.

\section{Framework and entropy decay}\label{sec:proof-main}

We define the weighted spaces and prove the two entropy-decay
criteria from the introduction.  The first proof uses a logarithmic Sobolev
inequality and the Stroock--Varopoulos comparison; the second obtains
entropy contraction directly from KMS symmetry and hypercontractivity.
We then show why the comparison can fail when different frequencies are
coupled, and give approximation conditions that extend it to infinite-dimensional
Davies generators.

\subsection{Spaces and Gibbs samplers}\label{subsec:spaces-gibbs}
We define the operator spaces, Gibbs samplers, and reversibility
assumptions used throughout the paper.

We use $\N=\{1,2,\ldots\}$ and $\N_0=\{0,1,\ldots\}$.
For a separable Hilbert space $\cH$, let $B(\cH)$ be the algebra of
bounded operators and let $I$ be its identity.  The adjoint of an operator
$X$ is written $X^\dagger$.
For $1\le p<\infty$, the Schatten class and its norm are
\[
 \mathscr T_p(\cH)
 =\{X\text{ compact}:\Tr|X|^p<\infty\},
 \qquad |X|=(X^\dagger X)^{1/2},
 \qquad \|X\|_p=(\Tr|X|^p)^{1/p}.
\]
Thus $\mathscr T_1$ and $\mathscr T_2$ are the trace-class and
Hilbert--Schmidt spaces, respectively; $\|X\|_\infty$ is the operator
norm on $B(\cH)$; see \cite{SimonTraceIdeals2005} for trace ideals
and their norm inequalities.
For operators $X,Y$, we write $[X,Y]=XY-YX$ and
$\{X,Y\}=XY+YX$ whenever these products are defined.

Let $H$ be a self-adjoint Hamiltonian bounded from below, let $\beta>0$
be the inverse temperature, and assume that the partition function
$Z_\beta(H)=\Tr(e^{-\beta H})$ is finite.  The Gibbs state is
$\sigma_\beta(H)=Z_\beta(H)^{-1}e^{-\beta H}$.
We omit the Hamiltonian argument when it is fixed.
For the sampler construction, assume in addition that
$\Tr(e^{-bH})<\infty$ for every $b>0$, as in
\cite[Section~2]{BeckerRouzeSalzmannToAppearcmp}, and let $\{A^\alpha\}_{\alpha\in\cA}$ be a
finite adjoint-closed family of bare jump operators: the index set has
an involution $\alpha\mapsto\bar\alpha$ with
$A^{\bar\alpha}=(A^\alpha)^\dagger$. These operators specify the allowed
transitions. The jumps are closed on their respective domains and
contain a common spectral core. We impose the energy-bound and filter
hypotheses of \cite[Condition~A]{BeckerRouzeSalzmannToAppearcmp}; they
are verified for the applications in Section~4.
For a filter $f:\R\to\C$, our Fourier convention is
\[
 \widehat f(\nu)=\int_\R e^{it\nu}f(t)\dd t,
 \qquad \nu\in\R,
\]
with a distributional interpretation when the integral is not absolutely
convergent.  The filtered jump operators $L^\alpha$ are given formally by
\begin{equation}\label{eq:dressed-jump}
 L^\alpha=\int_\R e^{itH}A^\alpha e^{-itH}f(t)\dd t,
 \qquad
 \overline{\widehat f(\nu)}=e^{-\beta\nu/2}\widehat f(-\nu).
\end{equation}
Let $B=B^\dagger$ be the detailed-balance correction Hamiltonian
chosen in the infinite-dimensional construction of
\cite{BeckerRouzeSalzmannToAppearcmp}.
The associated Gorini--Kossakowski--Lindblad--Sudarshan (GKLS)
generator is
\begin{equation}\label{eq:GKLS}
 \cL(\rho)=-i[B,\rho]
 +\sum_{\alpha\in\cA}\left(L^\alpha\rho(L^\alpha)^\dagger
 -\frac12\{(L^\alpha)^\dagger L^\alpha,\rho\}\right),
\end{equation}
Under the stated hypotheses, this choice of $B$ makes
$\sigma_\beta$ stationary.
The examples below use the multiplier $\fhat$, where
$\mathscr M$ labels the following filter:
\begin{equation}\label{eq:filter}
 \fhat(\nu)=\exp\left[-\frac{\sqrt{1+(\beta\nu)^2}+\beta\nu}{4}\right]
\end{equation}
For $k\in\N$ bosonic modes, we work in the occupation basis
$\{|n\rangle:n=(n_1,\ldots,n_k)\in\N_0^k\}$ of
$\ell^2(\N_0^k)$.  With $e_i$ the $i$th coordinate vector,
the annihilation, creation, and number operators are
\[
 a_i|n\rangle=\sqrt{n_i}\,|n-e_i\rangle,
 \qquad a_i^\dagger|n\rangle=\sqrt{n_i+1}\,|n+e_i\rangle,
 \qquad N_i=a_i^\dagger a_i,
\]
where the first expression is zero when $n_i=0$.
The bare jumps in our examples are these operators or their higher
powers.  For unbounded jumps or a nonintegrable filter,
\eqref{eq:dressed-jump}--\eqref{eq:GKLS} are formal expressions on the
common core specified by the construction. The sampler is the conservative quantum Markov semigroup obtained from
the closed Dirichlet form in \cite{BeckerRouzeSalzmannToAppearcmp}; its
$L^2$ realization is a self-adjoint operator.  We compute unbounded products on the common core specified in that
construction and then take closures.

For the general quantum Markov semigroup (QMS) of
\cref{subsec:main-results}, let $\cL$ be its generator.  Then
\[
 \Phi_t=e^{t\cL}\quad\text{on }\mathscr T_1(\cH),
 \qquad T_t=\Phi_t^*\quad\text{on }B(\cH),
\]
For $t\ge0$, the duality between the Schr\"odinger and Heisenberg
pictures is $\Tr(\Phi_t(\rho)X)=\Tr(\rho T_t(X))$.
Assume that $\sigma$ is a faithful invariant state.  We use
\emph{primitive} to mean that the fixed observables are exactly
$\C I$; in this setting $\sigma$ is the unique invariant state.
For $1\le p<\infty$, the map
$\iota_{p,\sigma}(X)=\sigma^{1/(2p)}X\sigma^{1/(2p)}$ on $B(\cH)$
extends to an isometry from the weighted space $L^p(\sigma)$ of
\cref{subsec:main-results} onto $\mathscr T_p(\cH)$.
We call this the symmetric Schatten realization and put
$\iota_{\infty,\sigma}=\mathrm{id}$.
For the weighted interpolation construction, see
\cite{Kosaki1984Interpolation,OlkiewiczZegarlinski1999}.
The $\sigma$-weighted mean is
$\langle X\rangle_\sigma=\Tr(\sigma X)$, extended continuously to
these spaces; $X$ is $\sigma$-centred if this mean is zero.
We also write $T_t$ for its consistent extensions to
$L^p(\sigma)$, writing $T_t^{(p)}$ when the exponent must be specified.
An operator norm $\|T:L^p(\sigma)\to L^q(\sigma)\|$ is the supremum
of $\|TX\|_{q,\sigma}/\|X\|_{p,\sigma}$ over nonzero $X$.

The KMS symmetry assumed in \cref{subsec:main-results}
\cite{FagnolaUmanita2010} is self-adjointness of $T_t$ for the inner product
\[
 \langle X,Y\rangle_{2,\sigma}
 =\Tr(\sigma^{1/2}X^\dagger\sigma^{1/2}Y).
\]
Transporting the $L^2$ semigroup to Hilbert--Schmidt space through
$\iota_{2,\sigma}$ gives its standard form:
\begin{equation}\label{eq:L2-link}
 \iota_{2,\sigma}T_t^{(2)}
 \iota_{2,\sigma}^{-1}=e^{-tL}
 \quad\text{on }\mathscr T_2(\cH).
\end{equation}
Here $L$ is a nonnegative self-adjoint operator on $\mathscr T_2(\cH)$.
It acts on Hilbert--Schmidt operators, whereas the jump operators
$L^\alpha$ in \eqref{eq:dressed-jump} act on $\cH$.
For any operator $A$, $D(A)$ denotes its operator domain.
The closed quadratic form of $L$, its form domain, and its
complex-time semigroup are
\[
 D(\cE)=D(L^{1/2}),\qquad
 \cE(Y)=\|L^{1/2}Y\|_2^2,\qquad
 e^{-zL}\quad(\operatorname{Re}z>0).
\]
A form core is a subspace dense in $D(\cE)$
for the norm $\bigl(\|Y\|_2^2+\cE(Y)\bigr)^{1/2}$.
For closed forms and form cores, see \cite[Sections~10.1--10.2]{Schmdgen2012};
the correspondence with Markov semigroups in standard form is treated
in \cite{Cipriani1997}.
When $\ker L=\C\sqrt\sigma$, the spectral gap from
\cref{subsec:main-results} has the variational expression
\[
 \operatorname{gap}(L)
 =\inf_{\substack{Y\in D(\cE),\ \|Y\|_2=1\\
                   \Tr(\sqrt\sigma\,Y)=0}}\cE(Y).
\]
\paragraph{\emph{Compactness from one endpoint.}}
For $2<q<\infty$, the inclusion $L^q(\sigma)\hookrightarrow L^2(\sigma)$
is compact. In symmetric Schatten coordinates it is
$J_q(Y)=\sigma^aY\sigma^a$, where $a=1/4-1/(2q)>0$.
For finite-rank spectral projections $P_n\uparrow I$ of $\sigma$,
Schatten H\"older with $r=1/a$ and $2/r+1/q=1/2$ gives
\[
 \|J_q-P_nJ_q(\,\cdot\,)P_n\|_{\mathscr T_q\to\mathscr T_2}
 \le2\bigl(\Tr((I-P_n)\sigma)\bigr)^a\longrightarrow0.
\]
The approximating maps have finite rank. Thus a bounded
$T_s:L^2(\sigma)\to L^q(\sigma)$ makes $e^{-sL}$ compact.
The spectral theorem then gives compact resolvent for $L$; if
$\ker L=\C\sqrt\sigma$, zero is simple and isolated, so
$\operatorname{gap}(L)>0$. This argument uses the atomic algebra
$B(\cH)$ and trace-class $\sigma$; it does not assert compactness
for diffuse algebras or a uniform numerical lower bound on the gap.

Let $\mathscr L^{(p)}$ denote the
generator of $T_t^{(p)}$ on $L^p(\sigma)$.  The minus sign in
\eqref{eq:L2-link} gives
\[
 L=-\iota_{2,\sigma}\mathscr L^{(2)}\iota_{2,\sigma}^{-1}.
\]

The Bohr frequencies are the energy differences in
$\Omega(H)=\{E'-E:E,E'\in\Sp(H)\}$.
Let $\sigma_E\in(0,\infty]$ be an energy-resolution width:
the sampler couples energy differences $\nu$ and $\mu$ with
Gaussian factor $\exp[-(\nu-\mu)^2/(8\sigma_E^2)]$, interpreted as
one when $\sigma_E=\infty$.  With the filter $\widehat f$ and
Hamiltonian $H$ fixed, write
$\cL_{\sigma_E}$ for the state-side generator and
$L_{\sigma_E}\ge0$ for its standard-form operator.
The explicit finite-dimensional formula
is \eqref{eq:cross-frequency-generator}.  At the Davies endpoint
$\sigma_E=0$, only equal frequencies remain; its form is denoted by
$\cE_0$ and is generated by $L_0$.

We call the family hyperbounded if for every $1<p<q<\infty$ there
are $\tau_{p,q}\ge0$ and $C_{p,q}\ge1$ such that
\begin{equation}\label{eq:weak-HC}
 \|T_t^{(p)}:L^p(\sigma)\to L^q(\sigma)\|
 \le C_{p,q},\qquad t>\tau_{p,q}.
\end{equation}
We call an $L^p$-to-$L^q$ bound contractive if its norm is at most one.
An endpoint bound means an $L^p$-to-$L^q$ bound at one fixed time.

\subsection{Relative entropy and entropy production}
We introduce entropy production and formulate the entropy-decay
criteria for Gibbs samplers with filtered jumps.

We use the Umegaki relative entropy introduced in \cref{subsec:main-results},
written $D(\rho\|\omega)$ for a faithful reference state $\omega$.
Logarithms are defined by functional
calculus, with $0\log0=0$; in infinite dimension the trace expression is
understood in the extended relative-entropy sense, rather than as a
difference of two possibly infinite traces
\cite[Section~2 and Remark~4.6]{Hiai2018Standard}.
We use trace-norm lower semicontinuity and data processing in this
extended sense; see \cite[Theorem~4.1(i),(iv)]{Hiai2018Standard}.
Let
\[
 B(\omega)=\bigcup_{\Gamma\ge1}
 \{\rho\in\mathscr T_1(\cH):\rho\ge0,\ \Tr\rho=1,
             \ \Gamma^{-1}\omega\le\rho\le\Gamma\omega\}.
\]
Here the inequalities are in the operator order and $\Gamma$ is a scalar
comparison constant.  The notation $B(\omega)$ denotes this set of
states, whereas $B(\cH)$ denotes the algebra of bounded operators.
For $\rho\in D(\cL)\cap B(\omega)$ and an invariant state
$\omega$, the entropy production is
\[
 \EP_{\cL}(\rho)
 =-\Tr\bigl(\cL(\rho)(\log\rho-\log\omega)\bigr).
\]
The logarithmic difference has a bounded extension on $B(\omega)$,
so this trace is defined.  A modified logarithmic Sobolev inequality
(MLSI) with rate $\alpha>0$ is
$\alpha D(\rho\|\omega)\le\EP_{\cL}(\rho)$ on this domain.
The precise Stroock--Varopoulos (SV) assumption in \cref{thm:main}
is that, for every $\rho\in D(\cL)\cap B(\sigma)$,
\begin{equation}\label{eq:SV-assumption}
 \sqrt\rho\in D(\cE),\qquad
 4\cE(\sqrt\rho)\le\EP_{\cL}(\rho).
\end{equation}
KMS symmetry alone does not imply this comparison; see
\cref{prop:generator-SV-counterexample}.  Under the assumptions of
\cref{thm:main}, the proof also gives the equivalent MLSI,
\[
 \alpha_{\rm HC}D(\rho\|\sigma)\le\EP_{\cL}(\rho),
 \qquad\rho\in D(\cL)\cap B(\sigma),
\]
with the same rate $\alpha_{\rm HC}$ as in \cref{thm:main}.

\Cref{thm:KMS-endpoint} applies to the filtered samplers as follows.

\begin{corollary}
\label{cor:generator-KMS-endpoint}
Let $\sigma_E\in(0,\infty]$ and let
$\Phi_t=e^{t\cL_{\sigma_E}}$ be any of the
conservative KMS-symmetric semigroups constructed in
\cite{BeckerRouzeSalzmannToAppearcmp}, with invariant Gibbs state
$\sigma_\beta$.  Assume the irreducibility hypotheses giving
$\ker L_{\sigma_E}=\C\sqrt{\sigma_\beta}$, and put
$L=L_{\sigma_E}\ge0$ in the sign convention of
\eqref{eq:L2-link}. If
\[
 \|T_s^{(2)}:L^2(\sigma_\beta)\to L^4(\sigma_\beta)\|\le C>1,
\]
then $\lambda=\operatorname{gap}(L)>0$ and
\eqref{eq:KMS-entropy-decay} holds with
\begin{equation}\label{eq:generator-KMS-rate}
 \alpha_{\rm KMS}
 =\left(2s+\frac{2\log C+1}{\lambda}\right)^{-1}.
\end{equation}
For $C=1$, the rate $(2s)^{-1}$ holds without a gap.  This includes the
original generator at $\sigma_E=\infty$ and every Gaussian width $0<\sigma_E<\infty$; it does not use the Davies limit $\sigma_E=0$.
\end{corollary}

For the same $L^2$-to-$L^4$ bound at time $s$, with constant $C>1$
and spectral gap $\lambda>0$, the two entropy-decay rates are
\[
 \alpha_{\rm HC}
 =\left(s+\frac{\log C+1/2}{\lambda}\right)^{-1},
 \qquad
 \alpha_{\rm KMS}=\frac12\alpha_{\rm HC}.
\]
\Cref{thm:main} gives the first rate when the Stroock--Varopoulos
comparison holds.  \Cref{thm:KMS-endpoint} gives the second rate
without that comparison.  The KMS argument therefore gives half the rate obtained with SV.  \Cref{prop:generator-SV-counterexample}
shows why the comparison cannot be assumed for cross-frequency
generators; it does not show that this factor of two is optimal.

Hyperboundedness gives an $L^2$-to-$L^4$ bound after the
waiting time $\tau_{2,4}$ in \eqref{eq:weak-HC}. Under primitivity,
compactness gives a positive gap, and \cref{thm:KMS-endpoint}
therefore gives exponential relative-entropy decay.  If the Stroock--Varopoulos comparison also
holds, \cref{thm:main} gives the faster rate above.
A unique invariant state does not by itself ensure a positive gap in
infinite dimension.  When the $L^2$-to-$L^4$ bound has norm at most
one, \cref{thm:KMS-endpoint} gives the rate $(2s)^{-1}$ without a gap.
If the Stroock--Varopoulos comparison holds as well,
\cref{thm:main} improves this rate to $s^{-1}$.

The following table compares the two entropy-decay criteria.  In the
first two rows the endpoint bound has constant $C>1$, and its
automatically positive spectral gap enters the rate; in the last two
rows the endpoint is contractive ($C=1$), so the rate needs no gap value.

\begin{center}
\begin{tabular}{@{}p{0.23\textwidth}p{0.27\textwidth}p{0.38\textwidth}@{}}
\toprule
Norm bound & Other assumptions & Entropy-decay rate \\
\midrule
$L^2\to L^4$ at $(s,C)$ & SV &
$\displaystyle\left(s+\frac{\log C+1/2}{\lambda}\right)^{-1}$ \\
$L^2\to L^{q_0}$ at $(s,C)$ & KMS symmetry &
$\displaystyle\frac{1-2/q_0}{s+\lambda^{-1}(\log C+1-2/q_0)}$ \\
contractive $L^2\to L^4$ & SV & $s^{-1}$ \\
contractive $L^2\to L^{q_0}$ & KMS symmetry & $(1-2/q_0)s^{-1}$ \\
\bottomrule
\end{tabular}
\end{center}

To prove \cref{thm:main}, we first derive a logarithmic Sobolev
inequality.  We interpolate the $L^2$-to-$L^4$ bound to small times and
differentiate at time zero. This gives a $2$-logarithmic Sobolev
inequality with an additional squared-norm term, called the defect.  The Rothaus
estimate transfers the defect to the component orthogonal to the
invariant vector, and the Poincar\'e inequality bounds this component
by the Dirichlet energy.  Combining the tight inequality with the Stroock--Varopoulos comparison
gives the modified logarithmic Sobolev inequality.  Exponential
relative-entropy decay then follows.  When $C=1$, differentiation already gives a tight inequality,
so the gap is unnecessary.

The proof of \cref{thm:KMS-endpoint} derives entropy contraction
directly from the weighted $L^p$ estimates.  For $C>1$, interpolation,
noncommutative $L^q$ smoothness, and decay of centred observables under
the spectral gap yield contractive $L^2$-to-$L^q$ bounds after an
explicit additional time.  KMS symmetry then allows interpolation
with the $L^1$ contraction to give a strong data-processing inequality
for relative entropy.  We divide the time interval into equal pieces and iterate the
short-time estimate. Letting the number of pieces tend to infinity
gives exponential decay.  Optimizing the interpolation
parameter gives the rate in \cref{thm:KMS-endpoint}.  For $C=1$, the argument starts from
the given contractive endpoint and again requires no gap.

The remainder of \cref{sec:proof-main} proves these criteria and
identifies the cross-frequency obstruction to the
Stroock--Varopoulos comparison.  Its infinite-dimensional Davies
criterion is conditional on the stated exhaustion hypotheses, which
must be verified for any proposed application.
\Cref{sec:perturbations} proves stability of the norm bounds under
bounded perturbations, and of $2$-logarithmic Sobolev inequalities
under changes of the Gibbs weight on a finite energy subspace.
\Cref{sec:applications} treats polynomial jumps, Bose--Hubbard
samplers with hopping, and finite-rank Coulomb perturbations of tensor
quantum Ornstein--Uhlenbeck (qOU) semigroups. It also gives the
periodic charging example.

\subsection{Modified logarithmic Sobolev inequality}
We relate the modified logarithmic Sobolev inequality to entropy decay
and express relative entropy in the Hilbert--Schmidt representation.

If
$\Gamma^{-1}\omega\le\rho\le\Gamma\omega$, operator monotonicity of
the logarithm shows that the quadratic-form difference
$\log\rho-\log\omega$ extends to a bounded self-adjoint operator, with
\[
 \|\log\rho-\log\omega\|_\infty\le\log\Gamma;
\]
see \cite{Schmdgen2012}.  Thus its pairing with
$\cL(\rho)\in\mathscr T_1(\cH)$ is defined for
$\rho\in D(\cL)\cap B(\omega)$.

\begin{lemma}\label{lem:MLSI-equivalence}
Let $\omega$ be a faithful invariant state of $\Phi_t=e^{t\cL}$ and let $\alpha>0$.  The estimate $D(\Phi_t(\rho)\|\omega)\le e^{-\alpha t}D(\rho\|\omega)$ for all states and all $t\ge0$ is equivalent to
\begin{equation*}
 \alpha D(\rho\|\omega)\le\EP_{\cL}(\rho),
 \qquad \rho\in D(\cL)\cap B(\omega).
\end{equation*}
\end{lemma}
This is the equivalence between MLSI and entropy decay in infinite
dimension; see \cite[Lemma~A.1 and Theorem~2.7]{wirth2025exponential}.

The hypercontractive estimates use the spaces $L^p(\sigma)$, whereas
the Dirichlet form acts on the fixed Hilbert space
$\mathscr T_2(\cH)$.  We identify an observable $X$
with
\[
 Y=\iota_{2,\sigma}(X)=\sigma^{1/4}X\sigma^{1/4}.
\]
This identification is isometric: $\|Y\|_2=\|X\|_{2,\sigma}$.
The norm of the same observable in $L^p(\sigma)$ is obtained by
changing the weights from
$\sigma^{1/4}$ to $\sigma^{1/(2p)}$.  In terms of $Y$, this gives
\begin{equation}\label{eq:transported-p-norm}
 \|Y\|_{p\leftarrow2,\sigma}
 :=\|\sigma^{1/(2p)-1/4}Y\sigma^{1/(2p)-1/4}\|_p
 =\|X\|_{p,\sigma},\qquad 1\le p<\infty,
\end{equation}
whenever these norms are finite.  The subscript $p\leftarrow2$ denotes
the weighted $L^p$ norm in the Hilbert--Schmidt coordinates fixed at
$p=2$.  At $p=2$, it is the ordinary Hilbert--Schmidt norm.

For $p>2$, the exponent of $\sigma$ in
\eqref{eq:transported-p-norm} is negative.  Since the eigenvalues of
$\sigma$ can approach zero, this expression need not be finite for
every Hilbert--Schmidt operator.  We define it by finite spectral cutoffs.  Choose increasing finite-rank spectral
projections $P_n$ of $\sigma$ with $P_n\to I$ strongly, and set
\[
 \|Y\|_{p\leftarrow2,\sigma}
 =\lim_{n\to\infty}
 \|\sigma^{1/(2p)-1/4}P_nYP_n
       \sigma^{1/(2p)-1/4}\|_p\in[0,\infty].
\]
All powers in each term act on the finite-dimensional range of $P_n$.
The norms increase with $n$, and their limit is independent of the
chosen exhaustion.  The domain consists of those $Y$ for which this
limit is finite; on that domain the reweighted operator belongs to
$\mathscr T_p(\cH)$.  The proof of \cref{thm:Gross} constructs the form core used to
differentiate the semigroup.

For $Y\in\mathscr T_2(\cH)\setminus\{0\}$, define
\begin{align*}
 \Ent_{2,\sigma}(Y)
 :=\frac{\|Y\|_2^2}{2}\bigg[
 &D\left(\frac{Y^\dagger Y}{\|Y\|_2^2}\middle\|\sigma\right)
 +D\left(\frac{YY^\dagger}{\|Y\|_2^2}\middle\|\sigma\right)
 \bigg],
\end{align*}
with $\Ent_{2,\sigma}(0)=0$ and infinite values allowed. Both
normalized operators inside the relative entropies are states, and
this functional is lower semicontinuous in the Hilbert--Schmidt norm.
For $Y$ supported on a finite spectral subspace of $\sigma$,
differentiating \eqref{eq:transported-p-norm} gives
\[
 \Ent_{2,\sigma}(Y)
 =4\|Y\|_2^2
 \left.\frac{\dd}{\dd p}\log\|Y\|_{p\leftarrow2,\sigma}
 \right|_{p=2},\qquad Y\ne0.
\]
Here $Y$ is fixed while $p$ varies. In infinite dimension we use
the relative-entropy formula above; the cutoff argument in
\cref{thm:Gross} below proves the required norm comparison.

In particular, a state $\rho$ is represented in the positive cone of
Hilbert--Schmidt space by its square root $Y=\sqrt\rho$.  Then
$Y^\dagger Y=YY^\dagger=\rho$ and $\|Y\|_2^2=\Tr\rho=1$, so
\begin{equation*}
 \Ent_{2,\sigma}(\sqrt\rho)=D(\rho\|\sigma).
\end{equation*}
The equilibrium state corresponds to $\sqrt\sigma$, which is also
the image of the identity observable under $\iota_{2,\sigma}$.
Thus an inequality for $\Ent_{2,\sigma}(Y)$ gives a relative-entropy
inequality by taking $Y=\sqrt\rho$.

\subsection{The Gross and Rothaus steps}
We differentiate the hypercontractive bound and use the Rothaus and
Poincar\'e inequalities to obtain a tight logarithmic Sobolev inequality.

For standard-form vectors use the complex Hilbert--Schmidt inner
product $\langle X,Y\rangle:=\Tr(X^\dagger Y)$.  Thus
$\langle\sqrt\sigma,X\rangle\sqrt\sigma$ is the orthogonal projection
of $X$ onto the invariant line $\C\sqrt\sigma$.

\begin{theorem}\label{thm:Gross}
Under the semigroup assumptions of \cref{thm:main}, the endpoint \eqref{eq:24-main} implies, for every standard-form vector $X\in D(\cE)\subset\mathscr T_2(\cH)$,
\begin{equation}\label{eq:defective-general}
 \Ent_{2,\sigma}(X)
 \le4s\cE(X)+4\log C\|X\|_2^2.
\end{equation}
Consequently, for every state $\rho$ with $\sqrt\rho\in D(\cE)$,
\begin{equation}\label{eq:Rothaus}
 D(\rho\|\sigma)\le4s\cE(\sqrt\rho)
 +(4\log C+2)\|\sqrt\rho-\langle\sqrt\sigma,\sqrt\rho\rangle\sqrt\sigma\|_2^2.
\end{equation}
If the gap is $\lambda>0$, then
\begin{equation}\label{eq:tight-LSI}
 D(\rho\|\sigma)
 \le c_{\rm LS}\cE(\sqrt\rho),
 \qquad c_{\rm LS}=4s+\frac{4\log C+2}{\lambda}.
\end{equation}
For $C=1$, \eqref{eq:defective-general} directly gives \eqref{eq:tight-LSI} with $c_{\rm LS}=4s$ and no gap assumption.
\end{theorem}
\begin{proof}
\emph{Step 1: a spectral form core.}
Let $E_R=\mathbf1_{[0,R]}(L)$ and set
$\mathcal C=\bigcup_{R>0}E_R\mathscr T_2(\cH)$.
For $X\in D(\cE)$, the spectral truncations $X_R=E_RX$ satisfy
\[
 \|X-X_R\|_{\cE}^2
 =\|(I-E_R)X\|_2^2+\|(I-E_R)L^{1/2}X\|_2^2\longrightarrow0,
\]
where $\|X\|_{\cE}^2=\|X\|_2^2+\cE(X)$.  Thus
$\mathcal C$ is a form core.  If $X\in E_R\mathscr T_2(\cH)$,
then $e^{sL}X$ is well defined and has norm at most $e^{sR}\|X\|_2$.
The assumed endpoint therefore gives
\[
 \|X\|_{4\leftarrow2,\sigma}
 =\|e^{-sL}e^{sL}X\|_{4\leftarrow2,\sigma}
 \le C\|e^{sL}X\|_2\le Ce^{sR}\|X\|_2.
\]
Thus the transported $L^4$ norm is bounded by a constant times
the Hilbert--Schmidt norm on each spectral subspace. We first prove the estimates for a nonzero
$X\in E_R\mathscr T_2(\cH)$; the case $X=0$ is immediate.

\emph{Step 2: endpoint interpolation.}
Apply Stein interpolation \cite[Chapter~V]{SteinWeiss1971},
in its weighted noncommutative form
\cite[Theorem~2.10]{Jencova2018Renyi}, on
$0\le\operatorname{Re}z\le1$ to the analytic family $e^{-szL}$.  On the first boundary $e^{-isyL}$ is
unitary on $L^2$, and on the second boundary the factorization
$e^{-s(1+iy)L}=e^{-sL}e^{-isyL}$ gives the $L^2$-to-$L^4$ bound
$C$, for every $y\in\R$.  Consequently, for $0\le t\le s$,
\[
 \|e^{-tL}X\|_{q(t)\leftarrow2,\sigma}\le C^{t/s}\|X\|_2,
 \qquad \frac1{q(t)}=\frac{1-t/s}{2}+\frac{t/s}{4}.
\]
The exponent is $q(t)=4s/(2s-t)$. Apply the estimate to
$e^{tL}X$, which is well defined on the spectral subspace from
Step~1. Since $e^{-tL}e^{tL}X=X$, it follows that
\[
 \log\frac{\|X\|_{q(t)\leftarrow2,\sigma}}{\|X\|_2}
 \le\frac{t}{s}\log C+
 \log\frac{\|e^{tL}X\|_2}{\|X\|_2},\qquad0<t\le s.
\]

\emph{Step 3: finite cutoffs and the defect.}
Let $P_n$ be the finite spectral projections of $\sigma$ used above
and put $X_n=P_nXP_n$. For all sufficiently large $n$, $X_n\ne0$.
Fix $0<t\le s$ and define, for $0\le u\le t$,
\[
 g_n(u)=\log\|X_n\|_{q(u)\leftarrow2,\sigma}.
\]
Weighted interpolation makes $g_n$ convex, since $1/q(u)$ is affine
in $u$; see \cite{Kosaki1984Interpolation,Jencova2018Renyi}.
We differentiate only finite matrices on the range of $P_n$. The finite-matrix entropy formula and $q'(0)=1/s$ give
\[
 \frac{\Ent_{2,\sigma}(X_n)}{4s\|X_n\|_2^2}
 =g_n'(0)\le\frac{g_n(t)-g_n(0)}{t}.
\]
The weight here is the restriction of $\sigma$, without renormalizing
its trace. Thus this is precisely $\Ent_{2,\sigma}(X_n)$, with no
normalization correction.
By the definition of the transported norm,
\[
 \|X_n\|_2\longrightarrow\|X\|_2,\qquad
 \|X_n\|_{q(t)\leftarrow2,\sigma}
 \longrightarrow\|X\|_{q(t)\leftarrow2,\sigma}<\infty.
\]
Lower semicontinuity of the entropy and the bound on $g_n'(0)$ give
\begin{align*}
 \Ent_{2,\sigma}(X)
 &\le\liminf_{n\to\infty}\Ent_{2,\sigma}(X_n)\\
 &\le\frac{4s\|X\|_2^2}{t}
 \log\frac{\|X\|_{q(t)\leftarrow2,\sigma}}{\|X\|_2}\\
 &\le4\log C\|X\|_2^2+
 \frac{4s\|X\|_2^2}{t}
 \log\frac{\|e^{tL}X\|_2}{\|X\|_2}.
\end{align*}
Only the last Hilbert--Schmidt norm needs to be differentiated.
The operator $L$ is bounded on $E_R\mathscr T_2(\cH)$, so
\[
 \left.\frac{\dd}{\dd t}\log\|e^{tL}X\|_2\right|_{t=0}
 =\frac{\operatorname{Re}\langle X,LX\rangle}{\|X\|_2^2}
 =\frac{\cE(X)}{\|X\|_2^2}.
\]
Letting $t\downarrow0$ proves \eqref{eq:defective-general} on
$\mathcal C$. The argument applies to complex $X$ even when the spectral
projections of $L$ and $\sigma$ do not commute.

\emph{Step 4: the centred estimate.}
Write $a=\langle\sqrt\sigma,X\rangle$ and
$X_0=X-a\sqrt\sigma$.  Since $\|\sqrt\sigma\|_2=1$ and
$\sqrt\sigma\in\ker L$, we have
\[
 \langle\sqrt\sigma,X_0\rangle=0,\qquad
 \|X\|_2^2=|a|^2+\|X_0\|_2^2,\qquad \cE(X_0)=\cE(X).
\]
Moreover $X_0\in\mathcal C$, because every $E_R$ fixes
$\sqrt\sigma$.  The entropy-domain Rothaus inequality proved in
\cref{lem:rothaus-entropy-domain} below gives
\[
 \Ent_{2,\sigma}(X)\le\Ent_{2,\sigma}(X_0)+2\|X_0\|_2^2.
\]
Applying \eqref{eq:defective-general} to $X_0$ and using
$\cE(X_0)=\cE(X)$ now gives
\[
 \Ent_{2,\sigma}(X)\le4s\cE(X)+(4\log C+2)\|X_0\|_2^2.
\]
Thus the defect depends only on the component orthogonal to the
invariant vector.

\emph{Step 5: extension to the form domain and the gap.}
For general $X\in D(\cE)$, use the truncations $X_R=E_RX$
from Step~1.  They satisfy
$\cE(X_R)\to\cE(X)$ and $\|X_R\|_2\to\|X\|_2$.
Since $E_R\sqrt\sigma=\sqrt\sigma$, their centred components are
$(X_R)_0=E_RX_0$ and converge to $X_0$ in $L^2$.
For the defective bound, lower semicontinuity gives
\begin{align*}
 \Ent_{2,\sigma}(X)
 &\le\liminf_{R\to\infty}\Ent_{2,\sigma}(X_R)\\
 &\le\lim_{R\to\infty}\bigl(4s\cE(X_R)+4\log C\|X_R\|_2^2\bigr)
 =4s\cE(X)+4\log C\|X\|_2^2.
\end{align*}
Using the centred bound from Step~4 instead gives
\[
 \Ent_{2,\sigma}(X)
 \le\lim_{R\to\infty}\bigl(4s\cE(X_R)+(4\log C+2)\|E_RX_0\|_2^2\bigr)
 =4s\cE(X)+(4\log C+2)\|X_0\|_2^2.
\]
Thus both inequalities hold throughout $D(\cE)$.  With $X=\sqrt\rho$, the identity
$\Ent_{2,\sigma}(\sqrt\rho)=D(\rho\|\sigma)$ then gives
\eqref{eq:Rothaus}.

If $\lambda>0$, the definition of the spectral gap on the orthogonal
complement of $\sqrt\sigma$ gives
$\lambda\|X_0\|_2^2\le\cE(X_0)=\cE(X)$.
Substitution into the centred bound, again with $X=\sqrt\rho$,
yields
\[
 D(\rho\|\sigma)\le
 \left(4s+\frac{4\log C+2}{\lambda}\right)\cE(\sqrt\rho),
\]
which is \eqref{eq:tight-LSI}.  Finally, when $C=1$, the defect
$4\log C\|X\|_2^2$ in \eqref{eq:defective-general} vanishes.
Applying that inequality directly to $\sqrt\rho$ gives the constant
$4s$, without centering or a spectral-gap assumption.
\end{proof}

\begin{lemma}[Rothaus inequality on the entropy domain]
\label{lem:rothaus-entropy-domain}
Let $\sigma$ be a faithful density matrix on $\cH$.  For every
$Y\in\mathscr T_2(\cH)$ and $a\in\C$,
\begin{equation}\label{eq:rothaus-entropy-domain}
 \Ent_{2,\sigma}(Y+a\sqrt\sigma)
 \le\Ent_{2,\sigma}(Y)+2\|Y\|_2^2.
\end{equation}
Infinite values are allowed.  In particular, the estimate applies to
the centred vector $Y=X-\langle\sqrt\sigma,X\rangle\sqrt\sigma$
whenever its entropy is finite.  
\end{lemma}
\begin{proof}
\emph{Finite dimension.}
For $2\le q\le4$, set
$B_q=\sigma^{1/(2q)-1/4}Y\sigma^{1/(2q)-1/4}$.
The first variation of the squared Schatten norm at
$a\sigma^{1/q}$ is
\[
 \left.\frac{\dd}{\dd u}\|a\sigma^{1/q}+uB_q\|_q^2\right|_{u=0}
 =2\operatorname{Re}\bigl(\bar a\Tr(\sigma^{1-1/q}B_q)\bigr)
 =2\operatorname{Re}\bigl(\bar a\langle\sqrt\sigma,Y\rangle\bigr).
\]
For $a=0$ both sides vanish.  The Ball--Carlen--Lieb smoothness
inequality \cite{BallCarlenLieb1994} therefore gives
\[
 \|Y+a\sqrt\sigma\|_{q\leftarrow2,\sigma}^2
 \le |a|^2+2\operatorname{Re}\bigl(\bar a\langle\sqrt\sigma,Y\rangle\bigr)
 +(q-1)\|Y\|_{q\leftarrow2,\sigma}^2.
\]
This is an equality at $q=2$, without a centering assumption.
The middle term is independent of $q$.  For any fixed $Z$,
\[
 \left.\partial_q\|Z\|_{q\leftarrow2,\sigma}^2\right|_{q=2+}
 =2\|Z\|_2^2\left.\partial_q\log\|Z\|_{q\leftarrow2,\sigma}\right|_{q=2+}
 =\frac12\Ent_{2,\sigma}(Z).
\]
Since equality holds at $q=2$, the right derivatives satisfy
\[
 \frac12\Ent_{2,\sigma}(Y+a\sqrt\sigma)
 \le\|Y\|_2^2+\frac12\Ent_{2,\sigma}(Y).
\]
This is \eqref{eq:rothaus-entropy-domain}.  When $Z=0$ its squared
norm is identically zero, so the same argument includes $Y=0$ and
$Y+a\sqrt\sigma=0$.

\emph{Entropy recovery under spectral cutoffs.}
For a positive trace-class operator $A$ of trace $m>0$, write
\[
 \mathfrak D_\sigma(A):=mD(A/m\|\sigma)+m\log m,
 \qquad \mathfrak D_\sigma(0):=0.
\]
This definition avoids subtracting possibly infinite energy and
entropy traces.  The functional is lower semicontinuous in trace
norm \cite[Theorem~4.1(i)]{Hiai2018Standard} and satisfies
$\mathfrak D_\sigma(A)\ge m\log m$.
Choose increasing finite-rank spectral projections $P_n$ of $\sigma$
with $P_n\to I$ strongly, and put $Q_n=I-P_n$, $Y_n=P_nYP_n$.
We claim that finite $\Ent_{2,\sigma}(Y)$ implies
\begin{equation}\label{eq:entropy-cutoff-recovery}
 \Ent_{2,\sigma}(Y_n)\longrightarrow\Ent_{2,\sigma}(Y).
\end{equation}
To see this, set
\[
 A=Y^\dagger Y,\qquad A_n=Y_n^\dagger Y_n,\qquad
 R_n=P_nY^\dagger Q_nYP_n,\qquad S_n=Q_nAQ_n.
\]
Then $A_n+R_n=P_nAP_n$.  On the finite block, the reference-logarithm
terms cancel in the difference
\begin{align*}
 \mathfrak D_\sigma(A_n+R_n)-\mathfrak D_\sigma(A_n)-\mathfrak D_\sigma(R_n)
 ={}&\Tr\bigl(A_n[\log(A_n+R_n)-\log A_n]\bigr)\\
 &+\Tr\bigl(R_n[\log(A_n+R_n)-\log R_n]\bigr)\ge0.
\end{align*}
The last inequality follows from operator monotonicity of the logarithm
\cite[Chapter~V]{Bhatia1997Matrix}, with regularization at zero.
Data processing \cite[Theorem~4.1(iv)]{Hiai2018Standard} for the
pinching $A\mapsto P_nAP_n+Q_nAQ_n$, which fixes $\sigma$, gives the second:
\[
 \mathfrak D_\sigma(A_n)+\mathfrak D_\sigma(R_n)
 +\mathfrak D_\sigma(S_n)
 \le\mathfrak D_\sigma(P_nAP_n)+\mathfrak D_\sigma(S_n)
 \le\mathfrak D_\sigma(A).
\]
The discarded masses tend to zero because
\[
 r_n:=\Tr R_n=\|Q_nYP_n\|_2^2\le\|Q_nY\|_2^2\longrightarrow0,
 \qquad s_n:=\Tr S_n=\|YQ_n\|_2^2\longrightarrow0.
\]
Since $\mathfrak D_\sigma(R_n)\ge r_n\log r_n$ and
$\mathfrak D_\sigma(S_n)\ge s_n\log s_n$, it follows that
\[
 \mathfrak D_\sigma(A_n)
 \le\mathfrak D_\sigma(A)-r_n\log r_n-s_n\log s_n.
\]
The trace-norm convergence follows from the ideal inequality:
\[
 \|A_n-A\|_1
 \le(\|Y_n\|_2+\|Y\|_2)\|Y_n-Y\|_2\longrightarrow0.
\]
Combining the upper bound with lower semicontinuity gives
\[
 \mathfrak D_\sigma(A)
 \le\liminf_n\mathfrak D_\sigma(A_n)
 \le\limsup_n\mathfrak D_\sigma(A_n)\le\mathfrak D_\sigma(A).
\]
Applying this argument to $Y^\dagger$ also yields
$\mathfrak D_\sigma(Y_nY_n^\dagger)\to\mathfrak D_\sigma(YY^\dagger)$.
The identity
\[
 \Ent_{2,\sigma}(Y)
 =\tfrac12\mathfrak D_\sigma(Y^\dagger Y)
 +\tfrac12\mathfrak D_\sigma(YY^\dagger)
 -\|Y\|_2^2\log\|Y\|_2^2
\]
and $\|Y_n\|_2\to\|Y\|_2$ prove
\eqref{eq:entropy-cutoff-recovery}.

\emph{Passage to the entropy domain.}
Assume $\Ent_{2,\sigma}(Y)<\infty$, since otherwise the inequality is immediate.  Put $z_n=\Tr(P_n\sigma)$ and
$\sigma_n=P_n\sigma P_n/z_n$, considered on $P_n\cH$.
For every $Z=P_nZP_n$,
\[
 \Ent_{2,\sigma_n}(Z)
 =\Ent_{2,\sigma}(Z)+\|Z\|_2^2\log z_n.
\]
With $X=Y+a\sqrt\sigma$, its compression is
$X_n=Y_n+a\sqrt{z_n}\sqrt{\sigma_n}$.  The finite-dimensional
inequality gives
\[
 \Ent_{2,\sigma_n}(X_n)
 \le\Ent_{2,\sigma_n}(Y_n)+2\|Y_n\|_2^2.
\]
Returning to the original reference state, this reads
\[
 \Ent_{2,\sigma}(X_n)
 \le\Ent_{2,\sigma}(Y_n)+2\|Y_n\|_2^2
 +(\|Y_n\|_2^2-\|X_n\|_2^2)\log z_n.
\]
Here $z_n\to1$, both norm sequences are bounded, and $X_n\to X$
in Hilbert--Schmidt norm.  Thus \eqref{eq:entropy-cutoff-recovery} gives
\[
 \Ent_{2,\sigma}(X)
 \le\liminf_n\Ent_{2,\sigma}(X_n)
 \le\Ent_{2,\sigma}(Y)+2\|Y\|_2^2,
\]
as required.
\end{proof}

\subsection{Entropy decay from KMS symmetry}
We prove \cref{thm:KMS-endpoint} by turning the hypercontractive
bound into entropy contraction through weighted interpolation and
noncommutative $L^q$ smoothness.

We use the standard centred smoothness estimate
\begin{equation}\label{eq:centred-smoothness}
 \|aI+Y\|_{q,\sigma}^2\le |a|^2+(q-1)\|Y\|_{q,\sigma}^2,
 \qquad 2\le q<\infty,\quad\langle Y\rangle_\sigma=0.
\end{equation}
It is the Ball--Carlen--Lieb inequality \cite{BallCarlenLieb1994}
applied at $a\sigma^{1/q}$ in direction $\iota_{q,\sigma}(Y)$;
the first variation vanishes by centering.  Schatten-norm density
extends the matrix inequality to the present Hilbert space.  This is
also the estimate used in \cite[Lemma~2.3]{GaoWang2026Entropy}.
The next lemma gives the extra time needed to make the norm at most one.

\begin{lemma}\label{lem:gap-strictification}
Let $T_t$ be a $\sigma$-preserving Markov semigroup with scalar fixed
space and let $\lambda>0$.  Suppose that, for every $\sigma$-centred $Y\in L^2(\sigma)$,
\[
 \|T_tY\|_{2,\sigma}\le e^{-\lambda t}\|Y\|_{2,\sigma}.
\]
If, for some $u>0$, $q>2$, and $K\ge1$,
\[
 \|T_u:L^2(\sigma)\to L^q(\sigma)\|\le K,
\]
then $\|T_{u+r}:L^2(\sigma)\to L^q(\sigma)\|\le1$ for
\begin{equation}\label{eq:generic-strictification-delay}
 r=\lambda^{-1}\bigl(\log K+\tfrac12\log(q-1)\bigr).
\end{equation}
\end{lemma}
\begin{proof}
Write $X=aI+Y$, where $a=\langle X\rangle_\sigma$ and
$\langle Y\rangle_\sigma=0$.  Invariance and orthogonality give
\[
 \langle T_tY\rangle_\sigma=0,\qquad T_tI=I,\qquad
 \|X\|_{2,\sigma}^2=|a|^2+\|Y\|_{2,\sigma}^2.
\]
Apply the gap estimate before using the endpoint:
\[
 \|T_uT_rY\|_{q,\sigma}
 \le K\|T_rY\|_{2,\sigma}
 \le Ke^{-\lambda r}\|Y\|_{2,\sigma}.
\]
The centred smoothness inequality \eqref{eq:centred-smoothness} then gives
\begin{align*}
 \|T_{u+r}X\|_{q,\sigma}^2
 &=\|aI+T_uT_rY\|_{q,\sigma}^2\\
 &\le |a|^2+(q-1)\|T_uT_rY\|_{q,\sigma}^2
 \le |a|^2+(q-1)K^2e^{-2\lambda r}\|Y\|_{2,\sigma}^2.
\end{align*}
The time in \eqref{eq:generic-strictification-delay} solves $(q-1)K^2e^{-2\lambda r}=1$.  The last expression is
therefore $\|X\|_{2,\sigma}^2$.
\end{proof}

\begin{lemma}
\label{lem:entropy-approximation}
For every state $\rho$ there are $\rho_k\in B(\sigma)$ such that
\begin{equation}\label{eq:entropy-approximation}
 \|\rho_k-\rho\|_1\longrightarrow0,
 \qquad D(\rho_k\|\sigma)\longrightarrow D(\rho\|\sigma)
\end{equation}
in the extended nonnegative reals.
\end{lemma}
\begin{proof}
In finite dimension, put $\rho_\varepsilon=(1-\varepsilon)\rho+
\varepsilon\sigma$.  Faithfulness of $\sigma$ makes
$\rho_\varepsilon\in B(\sigma)$, and
\[
 \|\rho_\varepsilon-\rho\|_1\le2\varepsilon,\qquad
 D(\rho_\varepsilon\|\sigma)\le(1-\varepsilon)D(\rho\|\sigma).
\]
Lower semicontinuity gives the matching lower limit.

In infinite dimension, choose increasing finite-rank spectral
projections $E_n\to I$ of $\sigma$, and put $Q_n=I-E_n$.
Faithfulness gives $\Tr(Q_n\sigma)>0$.  Define the normal channel
\[
 \mathsf E_n(\rho)=E_n\rho E_n+
 \Tr(Q_n\rho)\frac{Q_n\sigma Q_n}{\Tr(Q_n\sigma)}.
\]
It replaces the discarded block by the corresponding Gibbs tail.
It preserves the trace and fixes $\sigma$:
\[
 \Tr\mathsf E_n(\rho)=\Tr(E_n\rho)+\Tr(Q_n\rho)=1,
 \qquad \mathsf E_n(\sigma)=\sigma.
\]
Moreover
\[
 \|\mathsf E_n(\rho)-\rho\|_1
 \le\|E_n\rho E_n-\rho\|_1+\Tr(Q_n\rho)\longrightarrow0.
\]
To control the operator order, let $m_n>0$ be the smallest eigenvalue
of $E_n\sigma E_n$ on $E_n\cH$, and set
\[
 c_n=\max\left\{\frac{\|E_n\rho E_n\|_\infty}{m_n},
                  \frac{\Tr(Q_n\rho)}{\Tr(Q_n\sigma)}\right\}.
\]
The two blocks separately satisfy
\[
 E_n\rho E_n\le c_nE_n\sigma E_n,\qquad
 \Tr(Q_n\rho)\frac{Q_n\sigma Q_n}{\Tr(Q_n\sigma)}
 \le c_nQ_n\sigma Q_n,
\]
so $\mathsf E_n(\rho)\le c_n\sigma$.
Choose $\varepsilon_n=1/(n+1)$ and put
$\rho_n=(1-\varepsilon_n)\mathsf E_n(\rho)+\varepsilon_n\sigma$.
Then
\[
 \varepsilon_n\sigma\le\rho_n
 \le\bigl((1-\varepsilon_n)c_n+\varepsilon_n\bigr)\sigma,
 \qquad
 \|\rho_n-\rho\|_1\le\|\mathsf E_n(\rho)-\rho\|_1+2\varepsilon_n\to0.
\]
Thus $\rho_n\in B(\sigma)$.  If $D(\rho\|\sigma)<\infty$,
convexity, data processing, and lower semicontinuity yield
\begin{align*}
 D(\rho_n\|\sigma)
 &\le(1-\varepsilon_n)D(\mathsf E_n(\rho)\|\sigma)
 \le(1-\varepsilon_n)D(\rho\|\sigma),\\
 D(\rho\|\sigma)&\le\liminf_nD(\rho_n\|\sigma).
\end{align*}
These bounds prove entropy convergence.  If $D(\rho\|\sigma)=\infty$,
the second line alone forces $D(\rho_n\|\sigma)\to\infty$.
\end{proof}

The norm-to-entropy implication below is
\cite[Proposition~5.4]{Hirche2022contraction} in finite dimension.
We give the proof for unbounded Gibbs weights and arbitrary
normal states.

\begin{lemma}
\label{lem:weighted-SDPI}
Let $\Phi$ be a normal channel with faithful invariant state $\sigma$,
let $T=\Phi^*$, and assume that $T$ is $\sigma$-KMS symmetric and extends
consistently to positive contractions on $L^r(\sigma)$ for
$1\le r\le\infty$. Suppose that, for some $1\le K<\infty$,
\begin{equation}\label{eq:pq-SDPI-input}
 \|T:L^p(\sigma)\to L^q(\sigma)\|\le K,
 \qquad 1<p\le q<\infty.
\end{equation}
Then
\begin{equation}\label{eq:weighted-SDPI}
 D(\Phi(\rho)\|\sigma)
 \le \frac{(p-1)q}{p(q-1)}D(\rho\|\sigma)
 +\frac q{q-1}\log K
\end{equation}
for every state $\rho$.
\end{lemma}
\begin{proof}
For $\rho\in B(\sigma)$, the observable
$X=\iota_{1,\sigma}^{-1}(\rho)$ is bounded and positive.
We use the inverse only on this range.  KMS symmetry gives
$\Phi\iota_{1,\sigma}=\iota_{1,\sigma}T$.  Weighted complex interpolation
\cite{Kosaki1984Interpolation} of \eqref{eq:pq-SDPI-input} with
the $L^1$ contraction yields, for
$0\le\vartheta<1$,
\[
 \frac1{p_\vartheta}=\vartheta+\frac{1-\vartheta}{p},
 \quad
 \frac1{q_\vartheta}=\vartheta+\frac{1-\vartheta}{q},
 \quad
 \|TX\|_{q_\vartheta,\sigma}
 \le K^{1-\vartheta}\|X\|_{p_\vartheta,\sigma}.
\]
For $r>1$, define the sandwiched R\'enyi relative entropy
\cite[Section~3]{Jencova2018Renyi} by
\[
 \widetilde D_r(\rho\|\sigma)
 :=\frac1{r-1}\log\Tr\!\left[
 \left(\sigma^{(1-r)/(2r)}\rho
 \sigma^{(1-r)/(2r)}\right)^r\right].
\]
For $\rho\in B(\sigma)$, the positive product belongs to the
required Schatten class, and
\[
 \log\|\iota_{1,\sigma}^{-1}(\rho)\|_{r,\sigma}
 =\frac{r-1}{r}\widetilde D_r(\rho\|\sigma).
\]
Since $\Phi$ preserves
$\sigma$ and order, $\Phi(\rho)\in B(\sigma)$ as well.  The interpolated
bound therefore becomes
\[
 \widetilde D_{q_\vartheta}(\Phi(\rho)\|\sigma)
 \le\frac{(p_\vartheta-1)q_\vartheta}{p_\vartheta(q_\vartheta-1)}
 \widetilde D_{p_\vartheta}(\rho\|\sigma)
 +\frac{(1-\vartheta)q_\vartheta}{q_\vartheta-1}\log K.
\]
The interpolation identities imply
\[
 \frac{p_\vartheta-1}{p_\vartheta}
 =(1-\vartheta)\frac{p-1}{p},\qquad
 \frac{q_\vartheta-1}{q_\vartheta}
 =(1-\vartheta)\frac{q-1}{q}.
\]
The coefficients are therefore independent of $\vartheta$:
\[
 \frac{(p_\vartheta-1)q_\vartheta}{p_\vartheta(q_\vartheta-1)}
 =\frac{(p-1)q}{p(q-1)},\qquad
 \frac{(1-\vartheta)q_\vartheta}{q_\vartheta-1}=\frac q{q-1}.
\]
Letting $\vartheta\uparrow1$ proves
\eqref{eq:weighted-SDPI} on $B(\sigma)$ because
$\widetilde D_r(\rho\|\sigma)\to D(\rho\|\sigma)$ as $r\downarrow1$
for comparable states \cite[Proposition~3.8(ii)]{Jencova2018Renyi}.

For a general state with finite input entropy, choose $\rho_k$ from
\cref{lem:entropy-approximation}.  Trace-norm contractivity gives
$\|\Phi(\rho_k)-\Phi(\rho)\|_1\le\|\rho_k-\rho\|_1\to0$.
Consequently
\begin{align*}
 D(\Phi(\rho)\|\sigma)
 &\le\liminf_kD(\Phi(\rho_k)\|\sigma)\\
 &\le\lim_k\left(\frac{(p-1)q}{p(q-1)}D(\rho_k\|\sigma)
                 +\frac q{q-1}\log K\right)\\
 &=\frac{(p-1)q}{p(q-1)}D(\rho\|\sigma)+\frac q{q-1}\log K.
\end{align*}
For infinite input entropy the right-hand side is infinite, so the
assertion holds in the extended sense.
\end{proof}

\begin{proof}[Proof of \cref{thm:KMS-endpoint}]
We use the interpolation argument of
\cite[Theorem~3.3]{GaoWang2026Entropy}, using the preceding lemmas to justify each step in infinite dimension.  Put $\delta=1-2/q_0$,
so $0<\delta<1$.

\emph{Interpolation and removal of the endpoint constant.}
For $0<\theta\le1$, weighted Stein interpolation gives
\[
 \frac1{q_\theta}=\frac{1-\theta}{2}+\frac\theta{q_0}
 =\frac{1-\theta\delta}{2},\qquad
 \|T_{\theta s}:L^2(\sigma)\to L^{q_\theta}(\sigma)\|\le C^\theta.
\]
When $C>1$, \cref{lem:gap-strictification}, with $K=C^\theta$,
gives a contractive endpoint at
\begin{equation}\label{eq:strictification-delay}
 S_\theta=\theta s+
 \frac{\theta\log C+\frac12\log(q_\theta-1)}{\lambda}.
\end{equation}
Interpolate this endpoint with the $L^2$ bound at time zero.  For
$0\le u\le S_\theta$, define $r_\theta(u)$ by
\[
 \frac1{r_\theta(u)}
 =\frac{1-u/S_\theta}{2}+\frac{u/S_\theta}{q_\theta}
 =\frac12\left(1-\frac{\theta\delta u}{S_\theta}\right).
\]
Then $\|T_u:L^2(\sigma)\to L^{r_\theta(u)}(\sigma)\|\le1$.
For $u>0$, \cref{lem:weighted-SDPI} applies with
$p=2$, $q=r_\theta(u)$, and $K=1$.  Its coefficient is
\[
 \frac{r_\theta(u)}{2(r_\theta(u)-1)}
 =\frac1{1+(\theta\delta/S_\theta)u}.
\]
We have therefore obtained the short-time estimate
\[
 D(\Phi_u(\rho)\|\sigma)
 \le\frac{D(\rho\|\sigma)}{1+(\theta\delta/S_\theta)u},
 \qquad0\le u\le S_\theta.
\]

\emph{Iteration and the limiting rate.}
Assume $D(\rho\|\sigma)<\infty$, fix $t>0$, and take an integer
$n$ large enough that $t/n\le S_\theta$.  Applying the last estimate
successively to $\rho,\Phi_{t/n}(\rho),\ldots,\Phi_{(n-1)t/n}(\rho)$ gives
\[
 D(\Phi_t(\rho)\|\sigma)
 \le\left(1+\frac{\theta\delta t}{nS_\theta}\right)^{-n}
 D(\rho\|\sigma)
 \xrightarrow[n\to\infty]{}
 e^{-(\theta\delta/S_\theta)t}D(\rho\|\sigma).
\]
To compute the limit as $\theta\downarrow0$, use
\[
 q_\theta-1=\frac{1+\theta\delta}{1-\theta\delta},\qquad
 \frac{\log(q_\theta-1)}{2\theta}
 =\frac{\log(1+\theta\delta)-\log(1-\theta\delta)}{2\theta}
 \longrightarrow\delta.
\]
It follows that
\[
 \frac{\theta\delta}{S_\theta}
 =\frac{\delta}{s+\lambda^{-1}\left(\log C+
                    \frac{\log(q_\theta-1)}{2\theta}\right)}
 \longrightarrow\frac{\delta}{s+\lambda^{-1}(\log C+\delta)}.
\]
The entropy bound holds for every $\theta>0$, so passing to this
limit proves \eqref{eq:KMS-entropy-decay}.

For $C=1$, interpolation starts directly from $(s,q_0)$:
\[
 \frac1{r(u)}=\frac12\left(1-\frac{\delta u}{s}\right),\qquad
 \|T_u:L^2(\sigma)\to L^{r(u)}(\sigma)\|\le1,
 \qquad0\le u\le s.
\]
The same norm-to-entropy calculation and iteration give the rate
$\delta/s$. 
\end{proof}

\begin{proof}[Proof of \cref{cor:generator-KMS-endpoint}]
The construction is a conservative KMS-symmetric QMS of positive
contractions on the weighted scale, with invariant state $\sigma_\beta$.
Its Hilbert--Schmidt semigroup is $e^{-tL_{\sigma_E}}$.
Thus the nonnegative operator in \eqref{eq:L2-link} is
$L_{\sigma_E}$. The endpoint implies compact resolvent, so the
kernel assumption gives a positive gap and allows us to apply
\cref{thm:KMS-endpoint}; taking $q_0=4$ gives
\eqref{eq:generator-KMS-rate}, without separating the Bohr frequencies.
\end{proof}

\subsection{Stroock--Varopoulos inequality}
We recall the standard finite-dimensional Davies comparison, then
identify the cross-frequency obstruction and give conditions for passing to infinite dimension.
We use the following frequency convention throughout this subsection.  Write the Hamiltonian as
$H=\sum_{E\in\Sp(H)}EP_E$, where $P_E$ is the orthogonal projection
onto its $E$-eigenspace, and set
$\sigma=\sigma_\beta$.  Every faithful matrix state has this Gibbs
representation.
For each bare jump $A^\alpha$, its component at frequency $\nu$ is
\[
 A_\nu^\alpha
 :=\sum_{\substack{E,E'\in\Sp(H)\\E'-E=\nu}}
 P_{E'}A^\alpha P_E,
 \qquad \nu\in\Omega(H).
\]
Thus $[H,A_\nu^\alpha]=\nu A_\nu^\alpha$ and
$[\log\sigma,A_\nu^\alpha]=-\beta\nu A_\nu^\alpha$.
Adjoint closure of the bare family gives
$(A_\nu^\alpha)^\dagger=A_{-\nu}^{\bar\alpha}$.

\begin{proposition}[Finite-dimensional Davies comparison]
\label{thm:SV-finite}
Let $\dim\cH<\infty$ and use the adjoint-closed frequency decomposition
above.  Suppose
\[
 \cL=\sum_{\alpha,\nu}r_{\alpha,\nu}\mathcal D_{A_\nu^\alpha},
 \qquad\mathcal D_A(\rho)=A\rho A^\dagger-\tfrac12\{A^\dagger A,\rho\},
 \qquad r_{\bar\alpha,-\nu}=e^{\beta\nu}r_{\alpha,\nu}\ge0.
\]
Then every faithful state $\rho$ satisfies
$4\cE(\sqrt\rho)\le\EP_{\cL}(\rho)$.
\end{proposition}
This is \cite[Theorem~14]{beigi2020quantum} at $p=1$, $q=2$.
The rate pairing makes $-\cL^*$ self-adjoint for
$\Tr(\sigma X^\dagger Y)$, which is the strong reversibility assumed
there.  Under the standard-form identification, their $2$-Dirichlet
form is $\cE(\sqrt\rho)$ and their $1$-Dirichlet form is
$\EP_{\cL}(\rho)/4$; this accounts for the factor four.
For the filtered Davies generator, $r_{\alpha,\nu}=|\widehat f(\nu)|^2$
and the filter KMS identity supplies the rate pairing.

\subsubsection{The cross-frequency obstruction}
The Davies comparison uses frequency diagonality.  When different frequencies are coupled, the paired identity below
still holds, but the scalar chain inequality does not control the
mixed terms.

In finite dimension, retain the adjoint-closed frequency decomposition
above.  For $0<\sigma_E\le\infty$, set
\begin{align*}
 q_{\nu\mu}^{(\sigma_E)}
 &:={e^{-(\nu-\mu)^2/(8\sigma_E^2)}}
 \overline{\widehat f(\nu)}\widehat f(\mu),
 \qquad
 d_{\nu\mu}:=\frac{\beta(\nu-\mu)}4,\\
 c_{\nu\mu}^{(\sigma_E)}
 &:=q_{\nu\mu}^{(\sigma_E)}
 \frac{e^{\beta(\nu+\mu)/4}}{2\cosh d_{\nu\mu}}.
\end{align*}
The state-side cross-frequency generator is
\begin{equation}\label{eq:cross-frequency-generator}
 \cL_{\sigma_E}(\rho)
 =\sum_{\alpha,\nu,\mu}q_{\nu\mu}^{(\sigma_E)}
 \Bigl(A_\mu^\alpha\rho(A_\nu^\alpha)^\dagger
 -\frac{e^{d_{\nu\mu}}}{2\cosh d_{\nu\mu}}(A_\nu^\alpha)^\dagger A_\mu^\alpha\rho
 -\frac{e^{-d_{\nu\mu}}}{2\cosh d_{\nu\mu}}\rho(A_\nu^\alpha)^\dagger A_\mu^\alpha\Bigr).
\end{equation}
All frequency sums here and below run over $\Omega(H)$.
At $\sigma_E=\infty$ this is the original generator; finite
$\sigma_E$ retains the cross terms with a Gaussian weight.

\begin{lemma}
\label{lem:paired-generator}
Let $\rho$ be a faithful state,
$K_\rho=\log\rho-\log\sigma_\beta$, and write
$a_\nu=\beta\nu/2$.  Define
\begin{align*}
 U_{\alpha,\nu}
 &:=e^{-a_\nu/2}A_\nu^\alpha\sqrt\rho
     -e^{a_\nu/2}\sqrt\rho A_\nu^\alpha,\\
 W_{\alpha,\nu}
 &:=e^{-a_\nu}A_\nu^\alpha\rho
     -e^{a_\nu}\rho A_\nu^\alpha,
 \qquad
 V_{\alpha,\nu}:=[A_\nu^\alpha,K_\rho].
\end{align*}
Let $\cE_{\sigma_E}$ be the standard-form Dirichlet form associated
with $\cL_{\sigma_E}$.  Its energy and the entropy production
satisfy
\begin{align}
 \cE_{\sigma_E}(\sqrt\rho)
 &=\operatorname{Re}\sum_{\alpha,\nu,\mu}
 c_{\nu\mu}^{(\sigma_E)}
 \Tr(U_{\alpha,\nu}^\dagger U_{\alpha,\mu}),
 \label{eq:E-paired-generator}\\
 \EP_{\cL_{\sigma_E}}(\rho)
 &=\operatorname{Re}\sum_{\alpha,\nu,\mu}
 c_{\nu\mu}^{(\sigma_E)}
 \Tr(W_{\alpha,\nu}^\dagger V_{\alpha,\mu}).
 \label{eq:EP-paired-generator}
\end{align}
\end{lemma}
\begin{proof}
Put $Y=\sqrt\rho$.  The four terms in the first pairing are
\begin{align*}
 \Tr(U_{\alpha,\nu}^\dagger U_{\alpha,\mu})
 ={}&e^{-(a_\nu+a_\mu)/2}\Tr(Y(A_\nu^\alpha)^\dagger A_\mu^\alpha Y)\\
 &-e^{(a_\mu-a_\nu)/2}\Tr(Y(A_\nu^\alpha)^\dagger Y A_\mu^\alpha)\\
 &-e^{(a_\nu-a_\mu)/2}\Tr((A_\nu^\alpha)^\dagger Y A_\mu^\alpha Y)\\
 &+e^{(a_\nu+a_\mu)/2}\Tr((A_\nu^\alpha)^\dagger Y^2 A_\mu^\alpha).
\end{align*}
Multiply by $c_{\nu\mu}^{(\sigma_E)}$ and sum.  Using
$\sigma_\beta^tA_\nu^\alpha\sigma_\beta^{-t}
=e^{-2ta_\nu}A_\nu^\alpha$ and pairing the adjoint rows turns this sum into
\[
 -\Tr\!\left(Y\sigma_\beta^{-1/4}
 \cL_{\sigma_E}(\sigma_\beta^{1/4}Y\sigma_\beta^{1/4})
 \sigma_\beta^{-1/4}\right)
 =\cE_{\sigma_E}(Y).
\]
This is the first formula.  For the second, the modular relation
$[\log\sigma_\beta,A_\nu^\alpha]=-2a_\nu A_\nu^\alpha$ gives
\[
 V_{\alpha,\nu}
 =A_\nu^\alpha\log(e^{-a_\nu}\rho)
  -\log(e^{a_\nu}\rho)A_\nu^\alpha.
\]
The KMS identity for the filter reads
$\widehat f(-\nu)=e^{a_\nu}\overline{\widehat f(\nu)}$, hence
\begin{equation}\label{eq:generator-c-involution}
 c_{-\mu,-\nu}^{(\sigma_E)}=c_{\nu\mu}^{(\sigma_E)}
 \quad\text{under}\quad
 (\alpha,\nu,\mu)\longmapsto(\bar\alpha,-\mu,-\nu).
\end{equation}
Expansion and cyclicity give the real part of
\begin{align*}
 \sum_{\alpha,\nu,\mu}c_{\nu\mu}^{(\sigma_E)}\big[{}
 &e^{-a_\nu}\Tr(\rho(A_\nu^\alpha)^\dagger A_\mu^\alpha K_\rho)
 -e^{-a_\nu}\Tr(A_\mu^\alpha\rho(A_\nu^\alpha)^\dagger K_\rho)\\
 &-e^{a_\nu}\Tr((A_\nu^\alpha)^\dagger\rho A_\mu^\alpha K_\rho)
 +e^{a_\nu}\Tr(A_\mu^\alpha(A_\nu^\alpha)^\dagger\rho K_\rho)\big].
\end{align*}
Reindexing the last two terms by \eqref{eq:generator-c-involution}, and writing
\[
 p_{\nu\mu}:=\frac{e^{d_{\nu\mu}}}{2\cosh d_{\nu\mu}},
\]
use
\begin{align*}
 c_{\nu\mu}^{(\sigma_E)}e^{-a_\nu}
 &=q_{\nu\mu}^{(\sigma_E)}(1-p_{\nu\mu}), \
 c_{\nu\mu}^{(\sigma_E)}e^{-a_\mu}
 =q_{\nu\mu}^{(\sigma_E)}p_{\nu\mu},\text{ and }
 c_{\nu\mu}^{(\sigma_E)}(e^{-a_\nu}+e^{-a_\mu})
 =q_{\nu\mu}^{(\sigma_E)}.
\end{align*}
to obtain
\begin{align*}
 \operatorname{Re}\sum_{\alpha,\nu,\mu}q_{\nu\mu}^{(\sigma_E)}
 \big[{}&-\Tr(A_\mu^\alpha\rho(A_\nu^\alpha)^\dagger K_\rho)
 +p_{\nu\mu}\Tr((A_\nu^\alpha)^\dagger A_\mu^\alpha\rho K_\rho)\\
 &+(1-p_{\nu\mu})\Tr(\rho(A_\nu^\alpha)^\dagger
 A_\mu^\alpha K_\rho)\big],
\end{align*}
This is $-\Tr(\cL_{\sigma_E}(\rho)K_\rho)$ by
\eqref{eq:cross-frequency-generator}.
\end{proof}

For a fixed frequency, spectral decomposition and the scalar inequality
$4(\sqrt y-\sqrt x)^2\le(y-x)(\log y-\log x)$ give
\[
 4\|U_{\alpha,\nu}\|_2^2
 \le \operatorname{Re}\langle W_{\alpha,\nu},V_{\alpha,\nu}\rangle.
\]
The diagonal estimates do not control the terms with $\nu\ne\mu$ in
\eqref{eq:E-paired-generator}--\eqref{eq:EP-paired-generator}.
The following example shows that the Stroock--Varopoulos comparison
can fail even for a single self-adjoint bare jump.

\begin{proposition}
\label{prop:generator-SV-counterexample}
For the filter \eqref{eq:filter}, there exist a Hamiltonian $H$, a
self-adjoint bare jump $A$, and a faithful state $\rho$ on $\C^2$ such that
\[
 \EP_{\cL_{\sigma_E}}(\rho)<4\cE_{\sigma_E}(\sqrt\rho),
 \qquad \sigma_E\in\{30,\infty\}.
\]
\end{proposition}
\begin{proof}
Take
\[
 H=\begin{pmatrix}0&0\\0&1\end{pmatrix},\qquad
 A=\begin{pmatrix}1&1\\1&0\end{pmatrix},\qquad \beta=\frac14.
\]
Set $v=(2,-3)^{\mathsf T}/\sqrt{13}$ and
$\rho=\frac1{2000}I+\frac{999}{1000}|v\rangle\langle v|$.
The eigenvalues of $\rho$ are $1/2000$ and $1999/2000$, so $\rho$ is
faithful.  The corresponding matrix functions are
\[
 \sqrt\rho=\frac{I+(\sqrt{1999}-1)|v\rangle\langle v|}{\sqrt{2000}},
 \qquad \log\rho=-(\log2000)I+(\log1999)|v\rangle\langle v|.
\]
Substitution into
\eqref{eq:E-paired-generator}--\eqref{eq:EP-paired-generator} gives
\[
 \EP_{\cL_\infty}(\rho)<\frac9{1000}<4\cE_\infty(\sqrt\rho),
 \qquad
 \EP_{\cL_{30}}(\rho)<\frac1{100}<4\cE_{30}(\sqrt\rho).
\]

This proves both assertions.

For $\sigma_E=\infty$, the filtered jump is
\[
 L^1=\begin{pmatrix}
 \fhat(0)&\fhat(-1)\\ \fhat(1)&0
 \end{pmatrix}.
\]
Here the filter values are
\[
 \fhat(0)=e^{-1/4},\qquad
 \fhat(\pm1)=e^{-(\sqrt{17}\pm1)/16},\qquad
 \langle0|[L^1,(L^1)^\dagger]|0\rangle
 =\fhat(-1)^2-\fhat(1)^2>0.
\]
Thus $L^1$ is non-normal.  Every proper unital $*$-subalgebra of
$M_2(\C)$ is commutative and contains only normal matrices.
Consequently $L^1$ and $(L^1)^\dagger$ generate $M_2(\C)$, proving
irreducibility.
\end{proof}

\subsubsection{Unbounded Davies generators}
We give conditions under which finite-dimensional Davies comparisons
pass to the infinite-dimensional generator.

An unweighted sum over energy pairs diverges already on the diagonal,
since $|\fhat(0)|^2=e^{-1/2}$.  We must instead sum the transition blocks with their Gibbs
weights. The result below gives sufficient approximation conditions;
these must be checked for each model.

Assume now that $H$ has pure point spectrum with finite multiplicities,
$Z_\beta<\infty$, and $P_E\cH\subset D(A^\alpha)$ for all $E,\alpha$.
The component $A_\nu^\alpha$ is initially defined on
\[
 \mathcal D_{\rm fin}(H):=
 \operatorname{span}\{P_E\cH:E\in\Sp(H)\}
\]
by the same transition-block formula.  In the following expressions,
products with states are understood through their Hilbert--Schmidt
extensions whenever these exist.  Write $A$ for the whole family
$(A^\alpha)_{\alpha\in\cA}$, and define its Gibbs-weighted transition
activity by
\begin{equation}\label{eq:transition-activity}
 \mathfrak J_{\widehat f}(H,A)
 :=\frac1{Z_\beta}\sum_{\alpha\in\cA}\sum_{E,E'\in\Sp(H)}
 e^{-\beta E}|\widehat f(E'-E)|^2
 \|P_{E'}A^\alpha P_E\|_2^2.
\end{equation}
Assume $\mathfrak J_{\widehat f}(H,A)<\infty$.  For a state satisfying
$\Gamma^{-1}\sigma_\beta\le\rho\le\Gamma\sigma_\beta$, the expressions
$U_{\alpha,\nu}$ from \cref{lem:paired-generator} have Hilbert--Schmidt
extensions whenever $\widehat f(\nu)\ne0$, and obey the following bound,
with zero-weight frequencies omitted:
\begin{equation*}
 \sum_{\alpha,\nu}e^{\beta\nu/2}|\widehat f(\nu)|^2
 \|U_{\alpha,\nu}\|_2^2
 \le4\Gamma\,\mathfrak J_{\widehat f}(H,A).
\end{equation*}
Indeed, $\rho\le\Gamma\sigma_\beta$ gives, first on spectral cutoffs,
\[
 \|A_\nu^\alpha\sqrt\rho\|_2^2
 \le\Gamma\Tr(\sigma_\beta(A_\nu^\alpha)^\dagger A_\nu^\alpha),
 \qquad
 \|\sqrt\rho A_\nu^\alpha\|_2^2
 \le\Gamma\Tr(\sigma_\beta A_\nu^\alpha(A_\nu^\alpha)^\dagger).
\]
The transition-block decomposition yields
\begin{align*}
 \Tr(\sigma_\beta(A_\nu^\alpha)^\dagger A_\nu^\alpha)
 &=Z_\beta^{-1}\sum_{E'-E=\nu}e^{-\beta E}
        \|P_{E'}A^\alpha P_E\|_2^2,\\
 \Tr(\sigma_\beta A_\nu^\alpha(A_\nu^\alpha)^\dagger)
 &=e^{-\beta\nu}\Tr(\sigma_\beta(A_\nu^\alpha)^\dagger A_\nu^\alpha).
\end{align*}
Consequently the two terms defining $U_{\alpha,\nu}$ satisfy
\begin{align*}
 e^{\beta\nu/2}|\widehat f(\nu)|^2\|U_{\alpha,\nu}\|_2^2
 &\le2|\widehat f(\nu)|^2
 \bigl(\|A_\nu^\alpha\sqrt\rho\|_2^2
       +e^{\beta\nu}\|\sqrt\rho A_\nu^\alpha\|_2^2\bigr)\\
 &\le4\Gamma|\widehat f(\nu)|^2
       \Tr(\sigma_\beta(A_\nu^\alpha)^\dagger A_\nu^\alpha).
\end{align*}
Summing proves the activity bound and the existence of the
Hilbert--Schmidt extensions. This controls the Dirichlet form.
The logarithmic term needs a separate estimate.

\begin{definition}\label{def:SV-exhaustion}
Let $P_n$ be increasing finite-rank spectral projections of $H$ that
converge strongly to $I$.
For $\rho\in D(\cL_0)\cap B(\sigma_\beta)$, use $K_\rho$ as in
\cref{lem:paired-generator}.  Discard any initial indices with
$\Tr(P_n\rho)=0$ and set
\[
 \sigma_n=\frac{P_n\sigma_\beta P_n}{\Tr(P_n\sigma_\beta)},
 \qquad
 \rho_n=\frac{P_n\rho P_n}{\Tr(P_n\rho)},
 \qquad
 K_n=\log\rho_n-\log\sigma_n
\]
on $P_n\cH$.  Let $\cL_n$ be the finite-dimensional Davies generator
obtained by keeping the transitions whose initial and final energy
spaces both lie in $P_n\cH$.  Write $\cE_n$ for its standard-form Dirichlet form
and take $\EP_{\cL_n}$ relative to $\sigma_n$.  Extend the matrices
$\rho_n,\sigma_n$ and their square roots by zero to $\cH$ when
comparing different $n$.  We interpret a closed Dirichlet form as
$+\infty$ outside its domain.  The exhaustion is
\emph{SV-admissible for $\rho$} if
\begin{align}
 &\sqrt{\rho_n}\longrightarrow\sqrt\rho\quad\text{in }\mathscr T_2(\cH),
 \qquad
 \cE_0(\sqrt\rho)
 \le\liminf_{n\to\infty}\cE_n(\sqrt{\rho_n}),
 \label{eq:form-exhaustion}\\
 &\sup_n\|K_n\|_\infty<\infty,
 \qquad
 P_nK_nP_n\longrightarrow K_\rho
 \quad\text{strongly},
 \label{eq:log-exhaustion}\\
 &\cL_n(\rho_n)\longrightarrow\cL_0(\rho)
 \quad\text{in trace norm}.
 \label{eq:generator-exhaustion}
\end{align}
Here $K_n$ is extended by zero on $(P_n\cH)^\perp$. The exhaustion is SV-admissible on a class of states if these conditions hold for every state in that class.
\end{definition}

Uniform boundedness and strong convergence of $K_n$, together with
trace-norm generator convergence, imply
\begin{equation}\label{eq:EP-exhaustion}
 \EP_{\cL_n}(\rho_n)=-\Tr(\cL_n(\rho_n)K_n)
 \longrightarrow-\Tr(\cL_0(\rho)K_\rho)
 =\EP_{\cL_0}(\rho).
\end{equation}
To check the passage to the limit, extend $K_n$ by zero as above
and subtract the two trace pairings:
\begin{align*}
 |\Tr(\cL_n(\rho_n)K_n)-\Tr(\cL_0(\rho)K_\rho)|
 \le{}&\|\cL_n(\rho_n)-\cL_0(\rho)\|_1\|K_n\|_\infty\\
 &+|\Tr(\cL_0(\rho)(K_n-K_\rho))|\longrightarrow0.
\end{align*}
The first term vanishes by \eqref{eq:generator-exhaustion}.
For the second, approximate $\cL_0(\rho)$ in trace norm by a
finite-rank operator.  Strong convergence handles that finite-rank
pairing, and the uniform bound on $K_n$ controls the approximation error.

\begin{theorem}\label{thm:SV-exhaustion}
Suppose that an SV-admissible exhaustion exists on
$D(\cL_0)\cap B(\sigma_\beta)$. Then the Stroock--Varopoulos hypothesis \eqref{eq:SV-assumption} of \cref{thm:main} holds on this class. Consequently, \cref{thm:main} applies whenever its remaining hypotheses---the identification \eqref{eq:L2-link}, positive contractions on the weighted spaces, the bound \eqref{eq:24-main}, and a positive gap when $C>1$---are satisfied.
\end{theorem}
\begin{proof}
By \cref{thm:SV-finite},
$4\cE_n(\sqrt{\rho_n})\le\EP_{\cL_n}(\rho_n)$. Taking lower limits and using
\eqref{eq:form-exhaustion} and \eqref{eq:EP-exhaustion} gives
\[
 4\cE_0(\sqrt\rho)
 \le4\liminf_n\cE_n(\sqrt{\rho_n})
 \le\lim_n\EP_{\cL_n}(\rho_n)
 =\EP_{\cL_0}(\rho)<\infty.
\]
Finiteness follows because $\cL_0(\rho)$ is trace class and
$K_\rho$ is bounded.  Since the closed form is interpreted as
$+\infty$ outside its domain, the same inequality proves
$\sqrt\rho\in D(\cE_0)$ as well as the SV bound.
\end{proof}

\begin{corollary}\label{cor:Davies-main}
For a primitive finite-dimensional Davies generator,
\cref{thm:SV-finite} proves the SV inequality. Hence \cref{thm:main}
applies if the weighted extensions are positive contractions, satisfy
\eqref{eq:L2-link}, and obey \eqref{eq:24-main}. The same conclusion
holds in infinite dimension under \cref{def:SV-exhaustion}.
\end{corollary}

\begin{proof}[Proof of \cref{thm:main}]
Combining \eqref{eq:tight-LSI} and the assumed inequality \eqref{eq:SV-assumption} gives
\[
 \frac4{c_{\rm LS}}D(\rho\|\sigma)
 \le4\cE(\sqrt\rho)
 \le\EP_{\cL}(\rho).
\]
This holds on $D(\cL)\cap B(\sigma)$: the SV assumption places $\sqrt\rho$ in the form domain and gives
the last inequality.  Thus \cref{lem:MLSI-equivalence} applies with
\[
 \alpha_{\rm HC}=\frac4{c_{\rm LS}}
 =\begin{cases}
 \displaystyle\left(s+\frac{\log C+1/2}{\lambda}\right)^{-1},&C>1,\\[4pt]
 s^{-1},&C=1.
 \end{cases}
\]
Under \eqref{eq:weak-HC}, fix its constant $C_{2,4}$ and apply the
estimate at each $s>\tau_{2,4}$.  For finite input entropy the
right-hand side is continuous in $s$, so letting
$s\downarrow\tau_{2,4}$ gives the corresponding limiting rate.
Infinite input entropy requires no limit argument.
\end{proof}
\section{Perturbative hypercontractivity}\label{sec:perturbations}

We prove two stability results. With the Gibbs weight fixed, a
product formula controls changes of the generator that are bounded
along a hypercontractive curve. When the weight changes on a finite
energy subspace, entropy and form comparisons preserve a tight
logarithmic Sobolev inequality, provided the new generator has a
positive gap.

\subsection{Interpolation and bounded perturbations}
We transfer bounds along a hypercontractive curve under perturbations
that are bounded on each weighted space along the curve.

We use weighted Stein interpolation
\cite{SteinWeiss1971,Kosaki1984Interpolation} in the form established
in Step~2 of \cref{thm:Gross}; a weighted three-lines theorem is
\cite[Theorem~2.10]{Jencova2018Renyi}.  For a KMS-symmetric
semigroup as in \cref{sec:proof-main}, an endpoint
$\|T_\tau:L^2(\sigma)\to L^q(\sigma)\|\le C$ with
$\tau>0$ and $2<q<\infty$ gives
\[
 \|T_t:L^2(\sigma)\to L^{q(t)}(\sigma)\|\le C^{t/\tau},
 \qquad\frac1{q(t)}=\frac{1-t/\tau}{2}+\frac{t/\tau}{q},
 \quad0\le t\le\tau.
\]
The next theorem assumes a bound between every two points on
the exponent curve.  It is related to the Holley--Stroock
principles of \cite{ledoux2004logarithmic,junge2023stability}.
For the bounded perturbation theorem and the Trotter product formula
on a fixed Banach space, see
\cite[Theorem~III.1.3 and Corollary~III.5.8]{engel1999one}.

\begin{theorem}
\label{thm:bounded-perturbation}
Let $\sigma$ be a faithful density matrix, fix $t_*>0$ and
$2\le p_*<\infty$, and let $T_t^{(r)}$ and $S_t^{(r)}$ be strongly
continuous, consistent semigroups on $L^r(\sigma)$ for
$2\le r\le p_*$.  Fix $\omega\in\R$ and a continuous
nondecreasing exponent curve $p:[0,t_*]\to[2,p_*]$ with $p(0)=2$.
Suppose that
\begin{equation*}
 \|T_{t-s}:L^{p(s)}(\sigma)\to L^{p(t)}(\sigma)\|
 \le e^{\omega(t-s)},\qquad0\le s\le t\le t_*.
\end{equation*}
Let $\mathscr C$ be a common dense subspace of
$\bigcap_{2\le r\le p_*}L^r(\sigma)$ on which the realizations of
the semigroups agree.  One may take $B(\cH)$ when all maps are
initially defined there.  Let $\mathcal V$ be a linear map with bounded
realizations $\mathcal V^{(r)}:L^r(\sigma)\to L^r(\sigma)$ that
agree on $\mathscr C$.  We assume that adding $\mathcal V$ to the
generator of $T_t$ generates the semigroup $S_t$, in the precise sense
that, for $0\le t\le t_*$, the following $L^2$ Trotter formula holds:
\begin{equation*}
 S_t^{(2)}=\operatorname*{s-lim}_{n\to\infty}
 \bigl(T_{t/n}^{(2)}e^{(t/n)\mathcal V^{(2)}}\bigr)^n.
\end{equation*}
Here $\operatorname*{s-lim}$ denotes convergence on each vector of
$L^2(\sigma)$ in its norm.  Assume that the perturbation bound
\[
 M_{t_*}:=\sup_{2\le r\le p(t_*)}
 \|\mathcal V^{(r)}:L^r(\sigma)\to L^r(\sigma)\|
\]
is finite.
Then, for $0\le s\le t\le t_*$,
\begin{equation}\label{eq:bounded-perturbation}
 \|S_{t-s}:L^{p(s)}(\sigma)\to L^{p(t)}(\sigma)\|
 \le e^{(\omega+M_{t_*})(t-s)}.
\end{equation}
In particular,
$\|S_t:L^2(\sigma)\to L^{p(t)}(\sigma)\|
\le e^{(\omega+M_{t_*})t}$.  If the reverse Trotter formula holds with
$T$ and $S$ interchanged and perturbation $-\mathcal V$, the same
statement transfers a curve from $S$ back to $T$.
\end{theorem}
\begin{proof}
Fix $0\le s<t\le t_*$ and a positive integer $n$.  Put
$h=(t-s)/n$ and $s_k=s+kh$.  For $X\in\mathscr C$, define
\[
 X_{0,n}=X,\qquad X_{k,n}=T_he^{h\mathcal V}X_{k-1,n}
 \quad(1\le k\le n).
\]
The products agree wherever two of the weighted spaces overlap,
by consistency.  The exponential series gives
\[
 \|e^{h\mathcal V^{(p(s_{k-1}))}}\|
 \le\sum_{j=0}^\infty\frac{h^jM_{t_*}^j}{j!}
 =e^{hM_{t_*}}.
\]
Apply the assumed two-time bound to the other factor.  At the $k$th step,
\[
 \|X_{k,n}\|_{p(s_k),\sigma}
 \le e^{\omega h}\|e^{h\mathcal V}X_{k-1,n}\|_{p(s_{k-1}),\sigma}
 \le e^{(\omega+M_{t_*})h}\|X_{k-1,n}\|_{p(s_{k-1}),\sigma}.
\]
Iteration yields a bound independent of $n$:
\[
 \|X_{n,n}\|_{p(t),\sigma}
 \le e^{(\omega+M_{t_*})(t-s)}\|X\|_{p(s),\sigma}.
\]

We now pass to the limit in $L^{p(t)}(\sigma)$.  For $2\le q<\infty$,
H\"older's inequality and $\Tr\sigma=1$ give the continuous inclusion
$L^q(\sigma)\subset L^2(\sigma)$:
\[
 \|X\|_{2,\sigma}
 =\|\sigma^{1/4-1/(2q)}\iota_{q,\sigma}(X)
       \sigma^{1/4-1/(2q)}\|_2
 \le\|X\|_{q,\sigma}.
\]
The bounded sequence $(X_{n,n})_n$ has a weakly convergent subsequence
in the reflexive space $L^{p(t)}(\sigma)$, say with limit $Y$.
The same subsequence converges weakly to $Y$ in $L^2(\sigma)$.
The Trotter formula also gives
\[
 X_{n,n}=(T_he^{h\mathcal V})^nX
 \longrightarrow S_{t-s}^{(2)}X\quad\hbox{in }L^2(\sigma).
\]
Uniqueness of the weak limit therefore gives $Y=S_{t-s}X$.
Weak lower semicontinuity now proves
\[
 \|S_{t-s}X\|_{p(t),\sigma}
 \le\liminf_j\|X_{n_j,n_j}\|_{p(t),\sigma}
 \le e^{(\omega+M_{t_*})(t-s)}\|X\|_{p(s),\sigma}.
\]
By density of $\mathscr C$, this bound holds on all of
$L^{p(s)}(\sigma)$. Consistency identifies the limit with $S_{t-s}$.
For $s=t$ the bound is the identity estimate.  Replacing
$\mathcal V$ by $-\mathcal V$ proves the reverse statement whenever
the reverse product formula holds.
\end{proof}

\subsection{Finite-core changes of the Gibbs weight}
We now allow the Gibbs weight to change on a finite-dimensional
energy subspace and compare the resulting entropies and Dirichlet forms.

\begin{lemma}
\label{lem:finite-core-entropy}
Let $\sigma_0$ and $\sigma$ be faithful density matrices on $\cH$.
Assume that $P$ is a finite-rank orthogonal projection satisfying
$P\sigma_0=\sigma_0P$ and $P\sigma=\sigma P$.  Set $Q=I-P$ and suppose that
\[
 Q\sigma Q=c\,Q\sigma_0Q
\]
for a scalar $c>0$.  Then there is a constant $0\le\kappa<\infty$
such that, for every Hilbert--Schmidt operator
$Y\in\mathscr T_2(\cH)$ with finite $\Ent_{2,\sigma_0}(Y)$,
\begin{equation}
\label{eq:finite-core-entropy-comparison}
 \Ent_{2,\sigma}(Y)
 \le \Ent_{2,\sigma_0}(Y)+\kappa\|Y\|_2^2.
\end{equation}
\end{lemma}
\begin{proof}
The quadratic-form difference $K=\log\sigma_0-\log\sigma$ has a
bounded self-adjoint extension.  Indeed, block functional calculus gives
\[
 K=P(\log\sigma_0-\log\sigma)P-(\log c)Q,
 \qquad
 \|K\|_\infty\le\max\{\|P(\log\sigma_0-\log\sigma)P\|_\infty,|\log c|\}<\infty.
\]
The $P$-block is finite dimensional and both compressed states are
strictly positive there.  Since the difference of the reference logarithms is bounded,
relative entropy satisfies, whenever $D(\rho\|\sigma_0)<\infty$,
\[
 D(\rho\|\sigma)=D(\rho\|\sigma_0)+\Tr(\rho K).
\]
This identity holds also when the individual energy and von Neumann
entropy traces are infinite, by bounded perturbation of the reference
logarithm in the extended relative-entropy definition;
see \cite[Appendix~B.1]{Hiai2018Standard} for bounded relative Hamiltonians.
Apply it to $Y^\dagger Y/\|Y\|_2^2$ and $YY^\dagger/\|Y\|_2^2$
in the explicit relative-entropy formula for $\Ent_{2,\sigma}$.
For $Y\ne0$ this gives
\[
 \Ent_{2,\sigma}(Y)-\Ent_{2,\sigma_0}(Y)
 =\frac12\Tr\bigl((Y^\dagger Y+YY^\dagger)K\bigr)
 \le\|K\|_\infty\|Y\|_2^2.
\]
For $Y=0$ the claim is immediate.  Thus one may take
$\kappa=\|\log\sigma_0-\log\sigma\|_\infty$.
\end{proof}

\begin{proposition}
\label{prop:finite-core-hc}
Let $T_t^0$ and $T_t$ be primitive KMS-symmetric Markov semigroups with
faithful invariant states $\sigma_0$ and $\sigma$, and let
$L_0,L\ge0$ be their standard-form generators with closed forms
$\cE_0,\cE$ on $\mathscr T_2(\cH)$.
Assume the finite-core relation of \cref{lem:finite-core-entropy},
and choose its comparison constant $\kappa\ge0$.  Suppose that,
for some finite $c_0>0$,
\[
 \Ent_{2,\sigma_0}(Y)\le c_0\cE_0(Y),
 \qquad Y\in D(\cE_0).
\]
Assume also that $D(\cE)\subset D(\cE_0)$ and that there are finite
constants $A>0$ and $B\ge0$ such that
\begin{equation}
\label{eq:finite-core-form-comparison}
 \cE_0(Y)\le A\cE(Y)+B\|Y\|_2^2,
 \qquad Y\in D(\cE).
\end{equation}
If $L$ has gap $\lambda>0$, define
$c_{\rm fc}=c_0A+(c_0B+\kappa+2)/\lambda$.  Then
\begin{equation}
\label{eq:finite-core-tight-lsi}
 \Ent_{2,\sigma}(Y)\le c_{\rm fc}\cE(Y),
 \qquad Y\in D(\cE).
\end{equation}
This conclusion concerns the $2$-logarithmic Sobolev inequality.
We do not deduce a hypercontractive bound from these assumptions.
\end{proposition}
\begin{proof}
Fix $X\in D(\cE)$ and define its centred vector by
$Y=X-\langle\sqrt\sigma,X\rangle\sqrt\sigma$.
Since $\sqrt\sigma\in\ker L$, both $X$ and $Y$ belong to
$D(\cE)\subset D(\cE_0)$.  The free logarithmic Sobolev
inequality and \cref{lem:finite-core-entropy} first give
\[
 \Ent_{2,\sigma}(Y)
 \le\Ent_{2,\sigma_0}(Y)+\kappa\|Y\|_2^2
 \le c_0\cE_0(Y)+\kappa\|Y\|_2^2<\infty.
\]
Thus \cref{lem:rothaus-entropy-domain} applies to $Y$ without an
endpoint assumption.  Combining it with \eqref{eq:finite-core-form-comparison} gives
\[
 \Ent_{2,\sigma}(X)
 \le\Ent_{2,\sigma}(Y)+2\|Y\|_2^2
 \le c_0A\cE(X)+(c_0B+\kappa+2)\|Y\|_2^2.
\]
The centering does not change the energy, and the gap controls its norm:
\[
 L^{1/2}Y=L^{1/2}X,\qquad
 \lambda\|Y\|_2^2\le\|L^{1/2}Y\|_2^2=\cE(X).
\]
Substitution removes the squared-norm term:
\[
 \Ent_{2,\sigma}(X)
 \le\left(c_0A+\frac{c_0B+\kappa+2}{\lambda}\right)\cE(X)
 =c_{\rm fc}\cE(X).
\]  To deduce hypercontractivity from a logarithmic Sobolev inequality,
one needs regularity of the $L^p$ Dirichlet forms
\cite{OlkiewiczZegarlinski1999,bardet2017estimating}.
We prove an $L^2$-to-$L^4$ bound directly for the Coulomb sampler
in \cref{subsec:coulomb}.
\end{proof}
\section{Applications and obstructions}\label{sec:applications}

We apply the entropy-decay criteria to bosonic Gibbs samplers and
examine how the jump operators affect hypercontractivity.  We classify one-mode polynomial jumps by their order and then
treat Bose--Hubbard models with hopping. For finite-rank Coulomb
perturbations, we use a weighted oscillator estimate.  A periodic charging
model then shows that a positive spectral gap can coexist with failure
of every $L^2$-to-$L^p$ endpoint with $p>2$.

\subsection{Polynomial jumps and multimode Bose--Hubbard models}\label{subsec:mixed-polynomial}
We first compare the growth of a one-mode number Hamiltonian with the
order of the bosonic jumps.  This gives a classification of the
hypercontractive endpoints and explains why cubic jumps suffice for
quadratic number Hamiltonians.  We then prove exponential entropy
decay for a full multimode Bose--Hubbard sampler with hopping and
repulsive density interactions in \cref{thm:multimode-BH}.

Entropy decay for quantum Ornstein--Uhlenbeck semigroups, with
Hamiltonians quadratic in the creation and annihilation operators,
is well understood; see, for example,
\cite{carbone2008hypercontractivity,gao2021ricci,beigi2025meta}.
Here the Hamiltonian is a nonlinear polynomial of the number operator,
and the filtered transition rates depend strongly on occupation number.

We use the bosonic operators of \cref{subsec:spaces-gibbs} on the
one-mode space $\cH=\ell^2(\N_0)$, omitting their mode index.
Let $H=h(N)$, where $h$ is a real polynomial of degree $m\ge2$ with
positive leading coefficient, and fix $\beta>0$.  The occupation number
is denoted by $n$, and $Z_\beta=\sum_{n\ge0}e^{-\beta h(n)}$.
A jump of order $d\ge1$ creates or annihilates $d$ particles; we use the
adjoint pair $\cA_d=\{a^d,(a^\dagger)^d\}$.  Choose a finite nonempty set $D\subset\{1,2,\ldots\}$ of jump orders,
put $d_*:=\max D$, and take the bare family
$\cA_D:=\bigcup_{d\in D}\cA_d$.
The set $D$ is unrelated to the relative entropy $D(\rho\|\sigma)$.

\Cref{thm:mixed-polynomial} gives three regimes: no hypercontractive
endpoint when $d_*<m$, a positive waiting time when $d_*=m$, and an
endpoint at every positive time when $d_*>m$.  When $d_*=m$, the theorem gives the threshold but excludes
the boundary time from its statement.
The proof compares the growth of the Gibbs weight, governed by
$h(n)$, with the diagonal damping of order $n^{d_*}$.  The remaining terms
are uniformly bounded over compact intervals of finite exponents and
can be treated by \cref{thm:bounded-perturbation}.  A no-jump lower
bound in the invariant diagonal sector gives the converse estimates.

\paragraph{\emph{Separating the diagonal loss.}}
We first work with one jump pair and identify the unbounded part of
its generator.  We first prove the matrix formulas on the
finite-rank core $\operatorname{span}\{E_{nk}:n,k\ge0\}$, where
$E_{nk}=|n\rangle\langle k|$.
For the filter \eqref{eq:filter}, the dressed creation and annihilation
jumps are denoted by
$L_d^+$ and $L_d^-$, respectively:
\begin{equation*}
 L_d^+=(a^\dagger)^d\fhat(h(N+d)-h(N)),
 \qquad
 L_d^-=a^d\fhat(h(N-d)-h(N)).
\end{equation*}
Put
\[
 \Delta_d^+h(n)=h(n+d)-h(n),\qquad
 \Delta_d^-h(n)=h(n-d)-h(n),
\]
Here $h(n-d)$ means evaluation of the polynomial at the integer
$n-d$. For $n<d$, the lowering jump is zero, so no transition to
a negative occupation occurs.  Every falling factorial below
is defined to be zero when $n<d$.
To use the perturbation theorem, we need bounds on the weighted
generator.  Conjugation changes the gain operators and leaves the
diagonal losses unchanged.  On the symmetric Schatten realization of
$L^p(\sigma_\beta)$, for $2\le p<\infty$ define
\[
 g_{d,\pm,p}(n)
 =\fhat(\Delta_d^\pm h(n))e^{\beta(p-1)\Delta_d^\pm h(n)/(2p)}
 =\exp\left[-\frac{\sqrt{1+(\beta\Delta_d^\pm h(n))^2}}4
 +\frac{(p-2)\beta\Delta_d^\pm h(n)}{4p}\right]
\]
and
\begin{align*}
 B_{d,+,p}&=(a^\dagger)^d g_{d,+,p}(N),
 &B_{d,-,p}&=a^d g_{d,-,p}(N),\\
 b_{d,+}(N)&=\prod_{j=1}^d(N+j)|\fhat(\Delta_d^+h(N))|^2,
 &b_{d,-}(N)&=\prod_{j=0}^{d-1}(N-j)|\fhat(\Delta_d^-h(N))|^2.
\end{align*}
The symbols $B_{d,\pm,p}$ denote the transported Heisenberg gain
operators, and $b_{d,\pm}(N)$ denote the diagonal loss operators.
Indeed, transport of the gain $L^\dagger X L$ gives
$\sigma_\beta^{1/(2p)}L^\dagger\sigma_\beta^{-1/(2p)}$.
Relabelling the adjoint pair and using the KMS relation gives the
factor $(p-1)/(2p)$ above.  In particular the transported generator
annihilates $\sigma_\beta^{1/p}$, the vector representing $I$.  The weighted
generator of the $d$th jump pair, acting on its symmetric Schatten
realization, is
\begin{equation*}
 \mathscr L_{p,h,d}(X)
 =B_{d,+,p}XB_{d,+,p}^\dagger
 +B_{d,-,p}XB_{d,-,p}^\dagger
 -\frac12\{b_{d,+}(N)+b_{d,-}(N),X\}.
\end{equation*}

The next lemma bounds the gains and the upward loss. The downward
loss is the only unbounded part of the reference generator.
For nonnegative sequences $(u_n)$ and $(v_n)$ in this subsection,
$u_n\lesssim v_n$ means $u_n\le C v_n$ for all
sufficiently large $n$, with a constant $C$ independent of $n$;
$u_n\gtrsim v_n$ means $v_n\lesssim u_n$, and $u_n\asymp v_n$ means
both inequalities.  For a linear map $\mathscr V$ on $\mathscr T_p(\cH)$,
write $\|\mathscr V\|_{p\to p}$ for its Schatten-$p$ operator norm.

\begin{lemma}
\label{lem:single-complement}
For every $2\le p_*<\infty$, the operators $B_{d,\pm,p}$ and
$b_{d,+}(N)$ are bounded uniformly for $2\le p\le p_*$.  Hence
\begin{equation*}
 \mathscr V_{d,p}(X):=B_{d,+,p}XB_{d,+,p}^\dagger
 +B_{d,-,p}XB_{d,-,p}^\dagger
 -\frac12\{b_{d,+}(N),X\}
\end{equation*}
satisfies
\begin{equation}\label{eq:single-complement-bound}
 \sup_{2\le p\le p_*}\|\mathscr V_{d,p}\|_{p\to p}<\infty.
\end{equation}
\end{lemma}
\begin{proof}
For $x=\beta\Delta_d^\pm h(n)$ and $2\le p\le p_*$,
\begin{equation*}
 |g_{d,\pm,p}(n)|\le e^{-|x|/(2p_*)}.
\end{equation*}
For $x\ge0$ the exponent is at most $-x/(2p)$.
For $x<0$ it is at most $-(p-1)|x|/(2p)\le-|x|/(2p_*)$.
The polynomial expansion gives
\[
 \Delta_d^+h(n)=ma_md\,n^{m-1}+O(n^{m-2}),\qquad
 \Delta_d^-h(n)=-ma_md\,n^{m-1}+O(n^{m-2}),
\]
where $a_m>0$ is the leading coefficient.  Thus, for all sufficiently
large $n$, $|\Delta_d^\pm h(n)|\ge c_dn^{m-1}$ and
$\Delta_d^+h(n)\ge c_dn^{m-1}$.  This exponential decay dominates
the polynomial ladder coefficients.  For the positive increment,
put $x=\beta\Delta_d^+h(n)$; then
\[
 |\fhat(\Delta_d^+h(n))|^2
 =e^{-(\sqrt{1+x^2}+x)/2}\le e^{-x},
\]
so $b_{d,+}(n)\lesssim n^de^{-cn^{m-1}}$ for some $c>0$.
In particular,
\[
 \|B_{d,\pm,p}\|_\infty
 \le C\sup_{n\ge0}(1+n)^{d/2}e^{-cn^{m-1}/p_*}<\infty.
\]
The ideal property and the anticommutator bound now give
\[
 \|\mathscr V_{d,p}\|_{p\to p}
 \le\|B_{d,+,p}\|_\infty^2+\|B_{d,-,p}\|_\infty^2
 +\|b_{d,+}(N)\|_\infty,
\]
uniformly for $2\le p\le p_*$.
\end{proof}

\paragraph{\emph{Combining the jump orders.}}
We now sum the generators over $d\in D$ and include all downward
losses in one diagonal operator.  This separates the unbounded part, so
\cref{thm:bounded-perturbation} applies to the full family.  Set
\begin{align}
 A_D(N)&:=\frac12\sum_{d\in D}b_{d,-}(N),
 &\mathscr A_D(X)&:=-(A_D(N)X+XA_D(N)),\label{eq:mixed-diagonal}\\
 \mathscr V_{D,p}&:=\sum_{d\in D}\mathscr V_{d,p},
 &\mathscr L_{p,h,D}&=\mathscr A_D+\mathscr V_{D,p}.
 \label{eq:mixed-bounded-perturbation}
\end{align}
Here $\mathscr L_{p,h,D}=\sum_{d\in D}\mathscr L_{p,h,d}$.
By \cref{lem:single-complement},
\begin{equation*}
 M_{D,p_*}:=\sup_{2\le p\le p_*}
 \|\mathscr V_{D,p}\|_{p\to p}<\infty.
\end{equation*}
The diagonal loss satisfies
\begin{equation}\label{eq:mixed-lambda-asymptotic}
 A_D(n)\asymp n^{d_*},
\end{equation}
because for each fixed $d$,
\[
 \frac{b_{d,-}(n)}{n^d}
 =\prod_{j=0}^{d-1}\left(1-\frac jn\right)
   |\fhat(\Delta_d^-h(n))|^2\longrightarrow1,
 \qquad \frac{2A_D(n)}{n^{d_*}}\longrightarrow1.
\]
Thus the largest jump order determines the damping at large
occupation. The estimate \eqref{eq:mixed-lambda-asymptotic} concerns
large $n$ only; lowering jumps can vanish at small occupations.
The diagonal semigroup is
\[
 T_t^D(X)=e^{-tA_D(N)}Xe^{-tA_D(N)}.
\]
Its gain in integrability is measured by the following estimate,
which reduces the weighted norm bound to a scalar supremum: for
$2\le r\le q<\infty$,
\begin{equation}\label{eq:mixed-diagonal-rq}
 \|T_t^D:L^r(\sigma_\beta)\to L^q(\sigma_\beta)\|
 \le Z_\beta^{1/r-1/q}
 \sup_{n\ge0}\exp\!\left[
 \beta(1/r-1/q)h(n)-2tA_D(n)\right].
\end{equation}
Indeed, if $X\in L^r(\sigma_\beta)$ and
$Y=\iota_{r,\sigma_\beta}(X)\in\mathscr T_r(\cH)$, the transported
output is $K_{t,r,q}YK_{t,r,q}$, where
\[
 K_{t,r,q}=e^{-tA_D(N)}\sigma_\beta^{1/(2q)-1/(2r)}.
\]
The estimate follows from
$\|K_{t,r,q}YK_{t,r,q}\|_q\le\|K_{t,r,q}\|_\infty^2\|Y\|_r$.
The positive term in the exponent of \eqref{eq:mixed-diagonal-rq}
is the cost of changing the Gibbs weight from $L^r$ to $L^q$; the
negative term is the damping supplied by the jumps.  Comparing their
growth gives the three cases above.

\paragraph{\emph{Classification and entropy decay.}}
We pass from the diagonal estimate to the full sampler, then use
the gap criterion of \cref{thm:KMS-endpoint} to deduce entropy decay.
Write $S_t^D$ for the consistent weighted semigroup whose Schatten
realizations satisfy
\[
 \iota_{p,\sigma_\beta}S_t^D\iota_{p,\sigma_\beta}^{-1}
 =e^{t\mathscr L_{p,h,D}}.
\]
Let $\Phi_t^D$ be its Schr\"odinger-picture semigroup on trace-class
operators.

\begin{theorem}
\label{thm:mixed-polynomial}
Let $h$ be a real polynomial of degree $m\ge2$ with positive leading
coefficient, let $D\subset\{1,2,\ldots\}$ be finite and nonempty, and put
$d_*=\max D$.  For the sampler with bare jumps $\cA_D$ and filter
\eqref{eq:filter}, the following statements hold.
\begin{enumerate}
\item If $m>d_*$, the sampler has no bounded $L^2$-to-$L^4$ map at any finite
positive time.
\item If $m=d_*$ and $a_m>0$ is the leading coefficient of $h$,
then, for every $p>2$, the sampler is bounded
$L^2(\sigma_\beta)\to L^p(\sigma_\beta)$ for
\begin{equation}\label{eq:mixed-threshold}
 t>\beta a_m(1/2-1/p),
\end{equation}
and unbounded between these spaces for
$0<t<\beta a_m(1/2-1/p)$.
The boundary time is excluded from this statement.
\item If $m<d_*$, it is bounded $L^2(\sigma_\beta)\to
L^p(\sigma_\beta)$ for every $p>2$ and every $t>0$.
\end{enumerate}
\emph{Quantitative bound.} Fix a slope $\delta>0$ and define
\[
 p_\delta(u)=(1/2-\delta u)^{-1},\qquad
 0\le u<(2\delta)^{-1}.
\]
For every time $s\in(0,(2\delta)^{-1})$, define
\begin{equation*}
 \omega_{D,\delta,s}:=
 \max\!\left\{0,\,
 \delta\log Z_\beta+
 \sup_{n\ge0}\{\beta\delta h(n)-2A_D(n)\}
 +M_{D,p_\delta(s)}\right\}.
\end{equation*}
Whenever the supremum is finite, the full semigroup $S_t^D$ satisfies
\begin{equation}\label{eq:mixed-full-curve}
 \|S_{t-u}^D:L^{p_\delta(u)}(\sigma_\beta)
 \to L^{p_\delta(t)}(\sigma_\beta)\|
 \le e^{\omega_{D,\delta,s}(t-u)},
 \qquad0\le u\le t\le s.
\end{equation}

\emph{Kernel and entropy decay.} The mixed semigroup is conservative,
and its nonnegative standard-form operator $L_D=-\mathscr L_{2,h,D}$
has compact resolvent. Put $g=\gcd(D)$ and let $P_r$ project onto
occupations $n\equiv r\pmod g$. Then
\[
 \ker L_D=\operatorname{span}\{\sqrt{\sigma_\beta}P_r:0\le r<g\}.
\]
In particular, the sampler is primitive exactly when $g=1$.
Assume henceforth $g=1$, and let $\lambda_D>0$ be its smallest
positive eigenvalue. If $d_*=m$, \eqref{eq:mixed-full-curve} reaches
$p=4$ for every $s>\beta a_m/4$; if $d_*>m$, it does so for every $s>0$.
For either of these two regimes, fix such an admissible $s$
and, with $\delta=1/(4s)$, put
\begin{equation}\label{eq:mixed-entropy-constant}
 C_D(s):=e^{s\omega_{D,1/(4s),s}}.
\end{equation}
Then every state $\rho$ satisfies
\begin{equation}\label{eq:mixed-entropy-decay}
 D(\Phi_t^D(\rho)\|\sigma_\beta)
 \le e^{-\alpha_D(s)t}D(\rho\|\sigma_\beta),
 \qquad
 \alpha_D(s)=\left(
 2s+\frac{2\log C_D(s)+1}{\lambda_D}
 \right)^{-1}.
\end{equation}
\end{theorem}
\begin{proof}
\emph{Step 1: construction of the sampler.}
The sampler is obtained from the infinite-dimensional form
construction of \cite[Section~2]{BeckerRouzeSalzmannToAppearcmp}.
Indeed, choose a scalar $c\in\R$ that makes the Hamiltonian positive.
Put $\widetilde H=h(N)+cI\asymp(N+1)^m$.  For either
$A=a^d$ or $A=(a^\dagger)^d$, the ladder coefficients give
\[
 \|A\widetilde H^{-d/(2m)}\|_\infty
 +\|\widetilde H^{d/(2m)}A\widetilde H^{-d/m}\|_\infty<\infty.
\]
All jumps are closed on their respective maximal domains and have
the finite occupation span as a common core. For the finite family,
the same ladder estimates give the common energy-bound exponents $\gamma=d_*/(2m)$ and $\mu=d_*/m$.
The filter
\eqref{eq:filter} is nonvanishing and satisfies the required KMS and
boundedness conditions.  Thus the closed form of
\cite[Propositions~2.4 and~2.8]{BeckerRouzeSalzmannToAppearcmp}
generates a conservative KMS-symmetric QMS for every polynomial $h$
covered by \cref{thm:mixed-polynomial}.

We identify its generator before applying the perturbation theorem.
On Hilbert--Schmidt matrix units,
\[
 \mathscr A_D E_{nk}=-(A_D(n)+A_D(k))E_{nk}.
\]
Finite matrices form an operator core for this self-adjoint diagonal
operator. The remainder $\mathscr V_{D,2}$ is bounded by
\cref{lem:single-complement} and symmetric by KMS reversibility.
Consequently $\mathscr A_D+\mathscr V_{D,2}$ is essentially
self-adjoint on that core. The standard-form generator supplied by
the construction agrees there with
\eqref{eq:mixed-bounded-perturbation}, so uniqueness of the
self-adjoint closure identifies it with this bounded perturbation.
The bounded perturbation product formula
\cite[Corollary~III.5.8]{engel1999one} now gives
\[
 e^{t\mathscr L_{2,h,D}}
 =\operatorname*{s-lim}_{n\to\infty}
 \bigl(e^{(t/n)\mathscr A_D}e^{(t/n)\mathscr V_{D,2}}\bigr)^n
 \quad\text{on }\mathscr T_2(\cH).
\]
For the other exponents, transporting $\mathscr V_{D,p}$ back by
$\iota_{p,\sigma_\beta}$ gives the same map on finite matrix
observables. These observables are dense in every $L^p(\sigma_\beta)$
with $p<\infty$; the diagonal semigroups are consistent there as well.
Together with the QMS extensions, this verifies the consistency
and product-formula assumptions of \cref{thm:bounded-perturbation}.

\emph{Step 2: hypercontractive upper bounds.}
Fix the exponent curve $p_\delta$ and a time $s$ for which the
supremum in $\omega_{D,\delta,s}$ is finite.  The identity
$1/p_\delta(u)-1/p_\delta(t)=\delta(t-u)$ turns
\eqref{eq:mixed-diagonal-rq} into a bound along this curve, with
growth constant
\[
 \omega_0=\delta\log Z_\beta+
 \sup_{n\ge0}\{\beta\delta h(n)-2A_D(n)\},\qquad
 \|T_{t-u}^D:L^{p_\delta(u)}\to L^{p_\delta(t)}\|
 \le e^{\omega_0(t-u)}.
\]
By \cref{lem:single-complement}, the perturbation in
\eqref{eq:mixed-bounded-perturbation} is bounded uniformly for all exponents
along this curve.  Applying
\cref{thm:bounded-perturbation} adds $M_{D,p_\delta(s)}$ to this
constant and proves \eqref{eq:mixed-full-curve}.

It remains to choose the curve.  If $a_m$ is the leading coefficient
of $h$, then
\[
 \beta\delta h(n)-2A_D(n)
 =\beta\delta a_m n^m-n^{d_*}+o(n^m+n^{d_*}).
\]
For $d_*>m$, this tends to $-\infty$ for every $\delta>0$.
For $d_*=m$, its leading coefficient is $\beta\delta a_m-1$.
To reach $p>2$ at time $t$, choose
\[
 \delta=\frac{1/2-1/p}{t},\qquad
 \beta\delta a_m<1\quad\Longleftrightarrow\quad
 t>\beta a_m(1/2-1/p).
\]
This proves the claimed upper bounds in both regimes.

\emph{Step 3: the no-jump obstruction.}
To rule out earlier bounds, we use the invariant classical sector of diagonal observables,
whose outgoing rate is
$r(n)=\sum_{d\in D}(b_{d,+}(n)+b_{d,-}(n))\sim n^{d_*}$.
Writing $\pi_n=Z_\beta^{-1}e^{-\beta h(n)}$, test the semigroup on
the normalized observable $\pi_n^{-1/2}|n\rangle\langle n|$.
The probability of remaining at occupation $n$ without a jump until
time $t$ is $e^{-tr(n)}$, so
\[
 \|S_t^D:L^2(\sigma_\beta)\to L^p(\sigma_\beta)\|
 \ge e^{-tr(n)}\pi_n^{1/p-1/2}.
\]
The logarithm of this lower bound is
\[
 -tr(n)+(1/2-1/p)(\beta h(n)+\log Z_\beta)
 =\beta a_m(1/2-1/p)n^m-tn^{d_*}+o(n^m+n^{d_*}).
\]
It tends to $+\infty$ for $m>d_*$ at every finite $t$, and for
$m=d_*$ whenever $t<\beta a_m(1/2-1/p)$.  This proves both obstructions directly
for the full sampler.

\emph{Step 4: the kernel, gap and entropy decay.}
By Step~1, $L_D=-\mathscr A_D-\mathscr V_{D,2}$ is a bounded
symmetric perturbation of the nonnegative diagonal operator
$-\mathscr A_D$.  The latter has eigenvectors $E_{nk}$ and eigenvalues
$A_D(n)+A_D(k)\to\infty$ as $n+k\to\infty$.
If $\mathsf P_J$ projects onto the span of $E_{nk}$ with $n+k\le J$,
then
\[
 \|(I-\mathsf P_J)(I-\mathscr A_D)^{-1}\|_{2\to2}
 =\sup_{n+k>J}\frac1{1+A_D(n)+A_D(k)}\longrightarrow0.
\]
Thus $(I-\mathscr A_D)^{-1}$ is a norm limit of finite-rank maps.
The bounded perturbation transfers this compactness through
\[
 (L_D+I)^{-1}=(I-\mathscr A_D)^{-1}
 +(I-\mathscr A_D)^{-1}\mathscr V_{D,2}(L_D+I)^{-1}.
\]
Hence $L_D$ has compact resolvent.

For $Y\in\ker L_D$, every nonnegative paired form vanishes.
The paired identity \eqref{eq:E-paired-generator} extends from finite
matrices by the loss bound $C(N+I)^{d_*}$ and the form-core property
proved in Step~1. Since its integration weight is strictly positive,
each paired derivation vanishes. Each matrix entry involves finitely
many Bohr frequencies; extracting their coefficients, using the
nonvanishing filter and its KMS relation, gives
\[
 [X,a^d]=[X,(a^\dagger)^d]=0\quad(d\in D),\qquad
 X=\sigma_\beta^{-1/4}Y\sigma_\beta^{-1/4},
\]
entrywise. These products are defined entrywise because the jumps
are weighted shifts; no boundedness of $X$ is required. For any fixed
$d\in D$, $X$ therefore commutes with
$a^d(a^\dagger)^d=\prod_{j=1}^d(N+j)$, whose eigenvalues are strictly
increasing. Hence $X$ is diagonal, and its entries obey
$x_{n+d}=x_n$ for all $d\in D$. B\'ezout's identity writes $g=A-B$
with $A,B$ nonnegative integer combinations of elements of $D$, so
$x_n=x_{n+A}=x_{n+g+B}=x_{n+g}$. Thus $X$ is constant on residue
classes modulo $g$. Conversely, each $\sqrt{\sigma_\beta}P_r$ belongs
to the form domain by its Gibbs tail and is annihilated by all paired
derivations. This proves the stated kernel formula. When $g=1$,
compact resolvent and simplicity of zero give $\lambda_D>0$.

When $d_*=m$, the critical threshold in
\eqref{eq:mixed-threshold} is $\beta a_m/4$ at $p=4$.
When $d_*>m$, every $s>0$ is allowed.  Finally apply
\cref{thm:KMS-endpoint} with $q_0=4$ and norm bound
\eqref{eq:mixed-entropy-constant} to obtain
\eqref{eq:mixed-entropy-decay}.
\end{proof}

\paragraph{\emph{The full multimode Bose--Hubbard Hamiltonian.}}
With hopping, the filtered jumps are no longer explicit in the
occupation basis. Total particle number is still conserved, so we
can estimate the filters on finite-dimensional particle sectors.
The cubic losses dominate the errors in these estimates.

\begin{theorem}[Multimode Bose--Hubbard sampler with hopping]
\label{thm:multimode-BH}
Let $k\ge1$, $\beta>0$, and let $a_j,a_j^\dagger,N_j=a_j^\dagger a_j$
be the usual mode operators on $\ell^2(\N_0^k)$. Consider
\[
 H=\sum_{j=1}^k\left[\frac{U_j}{2}N_j(N_j-I)-\mu_jN_j\right]
 +\sum_{i<j}V_{ij}N_iN_j
 -\sum_{i<j}\left(J_{ij}a_i^\dagger a_j+
 \overline{J_{ij}}a_j^\dagger a_i\right),
\]
where $U_j>0$, $\mu_j\in\R$, $V_{ij}\ge0$, and $J_{ij}\in\C$ are
arbitrary. Put $\sigma_\beta=Z_\beta^{-1}e^{-\beta H}$, with
$Z_\beta=\Tr(e^{-\beta H})<\infty$. In each mode take the bare jumps
$a_j,a_j^\dagger,a_j^3,(a_j^\dagger)^3$, with filter
\[
 \fhat(\nu)=\exp\left[-\frac{\sqrt{1+(\beta\nu)^2}+\beta\nu}{4}\right].
\]
Writing $P_E$ for the spectral projections of $H$, define on finite
particle sectors
\[
 L^\alpha=\sum_{E,E'}\fhat(E'-E)P_{E'}A^\alpha P_E,\qquad
 C_H=\sum_\alpha(L^\alpha)^\dagger L^\alpha,
\]
and choose the correction Hamiltonian by
\[
 P_EBP_{E'}=\frac{i}{2}\tanh\!\left(\frac{\beta(E-E')}4\right)P_EC_HP_{E'}.
\]
The generator
\[
 \cL(\rho)=-i[B,\rho]+\sum_\alpha
 \left(L^\alpha\rho(L^\alpha)^\dagger
       -\frac12\{(L^\alpha)^\dagger L^\alpha,\rho\}\right)
\]
has a conservative KMS-symmetric realization with $\sigma_\beta$ as
its unique invariant state. Its nonnegative standard-form generator
$L$ has compact resolvent and a positive gap $\lambda$. If
$0=\lambda_1\le\lambda_2\le\cdots$ are its eigenvalues with multiplicity,
then, for the fixed model,
\[
 1+\lambda_j\asymp j^{3/(2k)},\qquad
 \Tr_{\mathscr T_2}(e^{-tL})\asymp t^{-2k/3}\quad(0<t\le1).
\]
Here the trace is on Hilbert--Schmidt space; the comparison constants
may depend on all model parameters. Its Heisenberg semigroup satisfies
\[
 \|T_t:L^2(\sigma_\beta)\to L^p(\sigma_\beta)\|<\infty
 \qquad(2<p<\infty,\ t>0).
\]
For any $s>0$, choose $C(s)\ge1$ bounding the $L^2$-to-$L^4$ norm
of $T_s$. Then every state $\rho$ satisfies
\[
 D(\Phi_t(\rho)\|\sigma_\beta)
 \le e^{-\alpha(s)t}D(\rho\|\sigma_\beta),\qquad
 \alpha(s)=\left(2s+\frac{2\log C(s)+1}{\lambda}\right)^{-1}>0.
\]
Alternatively, take the bare jumps $a_j,a_j^\dagger,a_j^2,(a_j^\dagger)^2$.
The resulting sampler is again conservative and primitive, with compact
resolvent and a positive gap. Let $W=(w_{ij})$ be real symmetric with
$w_{ii}=U_i$, $w_{ij}=V_{ij}$ for $i\ne j$, and put
$\Lambda=\lambda_{\max}(W)>0$. In this case a bounded
$L^2(\sigma_\beta)$-to-$L^p(\sigma_\beta)$ endpoint is available whenever
\[
 2<p<\infty,\qquad t>\frac{\beta\Lambda}{2}\left(\frac12-\frac1p\right).
\]
The same entropy estimate holds with its own gap and endpoint constant
at every $s>\beta\Lambda/8$. No assertion is made at the boundary time.
There is no occupation cutoff or smallness assumption on the hopping.
The constants may depend on all model parameters, including $k$.
\end{theorem}
\begin{proof}
The main estimate is that the cubic losses grow like the cube of
the total particle number, whereas the weighted gains grow only
linearly. We first prove these bounds on particle sectors, then use
them to control the full semigroup.

\emph{Step 1:}
Write $H=H_0+T$, where $H_0$ contains the number-diagonal terms and
$T$ is the hopping term. Set $\mathsf N=\sum_jN_j$, let
$Q_n=\mathbf1_{\{n\}}(\mathsf N)$, and put
$\cH^{[n]}=Q_n\cH$. Each term of $H$ preserves $\cH^{[n]}$,
whose dimension is $\binom{n+k-1}{k-1}$. Let $h_{\rm hop}$ be the Hermitian $k\times k$ matrix with
zero diagonal and off-diagonal entries $-J_{ij}$ and
$-\overline{J_{ij}}$. On the
symmetric $n$-particle space,
\[
 T|_{\cH^{[n]}}=\sum_{\ell=1}^n(h_{\rm hop})_\ell,
 \qquad \|TQ_n\|\le n\|h_{\rm hop}\|.
\]
If $U_{\min}=\min_jU_j>0$, the number-diagonal part satisfies
\[
 H_0\ge\frac{U_{\min}}2\sum_jN_j^2-C\mathsf N,
 \qquad H_0\le C\left(\sum_jN_j^2+\mathsf N+I\right).
\]
Moreover,
\[
 \sum_jN_j^2\ge k^{-1}\mathsf N^2,\qquad
 N_iN_j\le\tfrac12(N_i^2+N_j^2).
\]
The positive onsite coefficients and nonnegative density couplings
therefore give, after a scalar shift,
\[
 c(\mathsf N+I)^2\le\widetilde H:=H+c_0I
 \le C(\mathsf N+I)^2,\qquad
 \Tr(e^{-\beta\widetilde H})
 \le\sum_{n\ge0}\binom{n+k-1}{k-1}e^{-\beta c(n+1)^2}<\infty.
\]
The inequalities hold separately on every particle sector. For a lowering or raising jump of order $d\in\{1,3\}$,
\[
 \|a_i^dQ_n\|^2\le n(n-1)\cdots(n-d+1),\qquad
 \|(a_i^\dagger)^dQ_n\|^2\le(n+1)\cdots(n+d),
\]
where the first expression is zero for $n<d$. Each jump changes
$n$ by at most three. Combining these estimates with the lower and
upper bounds on $\widetilde H$ gives
\[
 \|A^\alpha\widetilde H^{-3/4}\|+
 \|\widetilde H^{3/4}A^\alpha\widetilde H^{-3/2}\|<\infty.
\]
The jumps are closed on their respective maximal domains and share
the finite particle span as a spectral core. These are the common
energy bounds required in
\cite[Condition~A and Proposition~2.4]{BeckerRouzeSalzmannToAppearcmp}.
The filter satisfies its KMS and boundedness conditions, so the form
construction gives a conservative KMS-symmetric sampler.
Every spectral sum defining a jump between fixed particle sectors
is finite. Below, constants are independent of the sector index $n$.

\emph{Step 2:}
Extend $V_{ij}$ symmetrically and set
$w_{ii}=U_i$, $w_{ij}=V_{ij}$ for $i\ne j$. Fix a mode $i$ and put
\[
 A=a_i^3,\qquad D_i=b_iI+3\sum_jw_{ij}N_j,\qquad
 c_i=b_i+6U_i+3\mu_i,
\]
where $b_i$ is sufficiently large that $D_i\ge4I$ and its values
at occupations differing by one hop are comparable. Write $\varepsilon(\mathbf n)$ for the value of $H_0$ at the
occupation vector $\mathbf n$. For $n_i\ge3$,
\begin{align*}
 \varepsilon(\mathbf n-3e_i)-\varepsilon(\mathbf n)
 &=\frac{U_i}{2}\bigl[(n_i-3)(n_i-4)-n_i(n_i-1)\bigr]
   +3\mu_i-3\sum_{j\ne i}V_{ij}n_j\\
 &=-3\sum_jw_{ij}n_j+6U_i+3\mu_i
 =-D_i(\mathbf n)+c_i.
\end{align*}
Applying this identity to each occupation vector gives
\[
 [H_0,A]=-A(D_i-c_iI).
\]
For $\operatorname{Re}z\ge-1$, put $K_z=A(z+D_i)^{-1}$.
Since $D_i(\mathbf n)\ge c(n_i+1)$ and
$|z+D_i(\mathbf n)|\ge D_i(\mathbf n)-1\ge\tfrac34D_i(\mathbf n)$,
\[
 \|K_zQ_n\|
 =\max_{|\mathbf n|=n}
 \frac{\sqrt{n_i(n_i-1)(n_i-2)}}{|z+D_i(\mathbf n)|}
 \le C\sqrt{n+1}.
\]
For a hopping term $T_{rs}=a_r^\dagger a_s$, $r\ne s$, put
$R_i(z)=(z+D_i)^{-1}$ and $\Delta_{rs}=3(w_{ir}-w_{is})$.
The product rule gives
\[
 [T_{rs},K_z]=[T_{rs},A]R_i(z)+A[T_{rs},R_i(z)],
 \qquad [T_{rs},A]=-3\mathbf1_{\{r=i\}}a_i^2a_s.
\]
The first term has norm at most $C\sqrt{n+1}$ on $\cH^{[n]}$:
its coefficient is bounded by
$C(n_i+1)\sqrt{n_s}/D_i(\mathbf n)$ when $r=i$.
For the second term, evaluate the commutator on $|\mathbf n\rangle$:
\begin{align*}
 [T_{rs},R_i(z)]|\mathbf n\rangle
 &=\sqrt{(n_r+1)n_s}
 \left(\frac1{z+D_i(\mathbf n)}
       -\frac1{z+D_i(\mathbf n)+\Delta_{rs}}\right)
 |\mathbf n+e_r-e_s\rangle\\
 &=\frac{\Delta_{rs}\sqrt{(n_r+1)n_s}}
 {(z+D_i(\mathbf n))(z+D_i(\mathbf n)+\Delta_{rs})}
 |\mathbf n+e_r-e_s\rangle.
\end{align*}
The choice of $b_i$ bounds the modulus of each denominator below by
$cD_i(\mathbf n)$. After applying $A$, the coefficient is bounded by
\[
 C\frac{(n_i+1)^{3/2}|w_{ir}-w_{is}|\sqrt{(n_r+1)n_s}}
 {(b_i+3\sum_jw_{ij}n_j)^2}\le C\sqrt{n+1}.
\]
Indeed, nonnegativity of the $w_{ij}$ and $w_{ii}>0$ bounds the first
factor $(n_i+1)^{3/2}$ by $CD_i(\mathbf n)^{3/2}$ and gives
\[
 |w_{ir}-w_{is}|\sqrt{(n_r+1)n_s}
 \le C D_i(\mathbf n)^{1/2}\sqrt{n+1}.
\]
Each term is a weighted occupation shift, so a bound on its
coefficients is an operator-norm bound. Summing the finitely many hopping terms therefore proves
\[
 \|K_zQ_n\|+\|[T,K_z]Q_n\|\le C\sqrt{n+1},
 \qquad \operatorname{Re}z\ge-1.
\]

On maps from $\cH^{[n]}$ to $\cH^{[n-3]}$, let
$\mathcal D_nX=H_{n-3}X-XH_n$, where $H_n=H|_{\cH^{[n]}}$.
Since $H_0$ commutes with $D_i$, we have
\begin{align*}
 [H,K_z]
 &=-A(D_i-c_iI)(z+D_i)^{-1}+[T,K_z]\\
 &=-A+zK_z+c_iK_z+[T,K_z].
\end{align*}
Thus
\[
 (z-\mathcal D_n)K_zQ_n=AQ_n-R_zQ_n,\qquad
 R_z=c_iK_z+[T,K_z],\qquad \|R_zQ_n\|\le C\sqrt{n+1}.
\]
In particular, with $K=AD_i^{-1}$ and $R=c_iK+[T,K]$,
\[
 A=-[H,K]+R,\qquad
 \|KQ_n\|+\|RQ_n\|\le C\sqrt{n+1}.
\]
The finitely many sectors $n<3$ contribute zero to $A$.

\emph{Step 3:}
We use the resolvent identity to estimate the energy filter on
each half of the frequency axis. The group
$e^{it\mathcal D_n}X=e^{itH_{n-3}}Xe^{-itH_n}$ consists of
operator-norm isometries. Hence
\[
 \|(z-\mathcal D_n)^{-1}\|_{\infty\to\infty}
 \le|\operatorname{Im}z|^{-1},\qquad
 \|g(\mathcal D_n)\|_{\infty\to\infty}
 \le\|\check g\|_{L^1(\R)},
\]
where $g(\nu)=\int_\R\check g(t)e^{it\nu}\dd t$ in the second
inequality. Also $\operatorname{spec}(\mathcal D_n)\subset[-M_n,M_n]$,
where $M_n=C(n+1)^2$.

Choose a smooth $\chi$ equal to zero on $(-\infty,-1]$ and to
one on $[0,\infty)$. Restricted to this spectral interval, it is a
sum of $O(1+\log(n+2))$ smooth functions $\chi_\ell$ supported
in $[-1,\infty)$ on dyadic intervals of length $2^\ell$, with
$|\chi_\ell^{(r)}|\le C_r2^{-\ell r}$. Each has an almost analytic
extension $\widetilde\chi_\ell$ with
\[
 \int_{\C}\frac{|\partial_{\bar z}\widetilde\chi_\ell(z)|}
 {|\operatorname{Im}z|}\dd x\dd y\le C,
 \qquad\operatorname{supp}\widetilde\chi_\ell
 \subset\{\operatorname{Re}z\ge-1\}.
\]
For example, use the first-order extension
$(\chi_\ell(x)+iy\chi_\ell'(x))\eta(y/2^\ell)$ with a fixed
smooth cutoff $\eta$. On $|y|\lesssim2^\ell$, its derivative obeys
\[
 |\partial_{\bar z}\widetilde\chi_\ell(x+iy)|
 \le C|y|2^{-2\ell}
 +C2^{-\ell}\mathbf1_{\{|y|\asymp2^\ell\}},
\]
on an $x$-interval of length $O(2^\ell)$. Dividing by $|y|$ and
integrating therefore gives a constant independent of $\ell$.
The Cauchy--Green, or Helffer--Sj\"ostrand, formula
\cite[Section~2]{GesztesyNichols2015} and Step~2 give
\begin{align*}
 \chi_\ell(\mathcal D_n)AQ_n
 &=A\chi_\ell(-D_i)Q_n
 -\frac1\pi\int_{\C}\partial_{\bar z}\widetilde\chi_\ell(z)
 (z-\mathcal D_n)^{-1}R_zQ_n\dd x\dd y,\\
 \|\chi_\ell(\mathcal D_n)AQ_n\|&\le C\sqrt{n+1}.
\end{align*}
The first term vanishes because $D_i\ge4I$. The formula can also
be checked entry by entry in the two finite energy eigenbases;
the resolvent estimate makes its bound independent of their dimensions.
Summing the dyadic pieces yields
\[
 \|\chi(\mathcal D_n)AQ_n\|
 \le C\sqrt{n+1}\log(n+2).
\]

Write $f=\fhat$ in this proof and set
$g_-(\nu)=f(\nu)-1+\chi(\nu)$ and
$g_+(\nu)=f(-\nu)-\chi(\nu)$. For every integer $r\ge0$,
\[
 |g_\pm^{(r)}(\nu)|\le C_r(1+|\nu|)^{-1-r}.
\]
A dyadic partition and the Fourier estimate above imply
\[
 \|g_\pm(\mathcal D_n)\|_{\infty\to\infty}\le C,
 \qquad
 \|(\nu g_\pm)(\mathcal D_n)\|_{\infty\to\infty}
 \le C\log(n+2).
\]
Here is the dyadic estimate. For a smooth $u$ supported in a fixed
bounded interval, put $u_\ell(\nu)=u(2^{-\ell}\nu)$.
Plancherel, Cauchy--Schwarz, and a change of variables give
\[
 \|\check u\|_1\le C(\|u\|_2+\|u'\|_2),\qquad
 \check u_\ell(t)=2^\ell\check u(2^\ell t),\qquad
 \|\check u_\ell\|_1=\|\check u\|_1.
\]
The rescaled pieces of $g_\pm$ have size $O(2^{-\ell})$, while
those of $\nu g_\pm$ have size $O(1)$. If $L_n$ is the number of
dyadic scales meeting $[-M_n,M_n]$, the two bounds follow from
\[
 \sum_{\ell=0}^{L_n}C2^{-\ell}\le2C,
 \qquad \sum_{\ell=0}^{L_n}C\le C\log(n+2).
\]
Apply the multiplier to $A=-\mathcal D_nK+R$ on this block:
\begin{align*}
 g_\pm(\mathcal D_n)AQ_n
 &=-(\nu g_\pm)(\mathcal D_n)KQ_n+g_\pm(\mathcal D_n)RQ_n,\\
 \|g_\pm(\mathcal D_n)AQ_n\|
 &\le C\log(n+2)\|KQ_n\|+C\|RQ_n\|.
\end{align*}
Step~2 now gives
\[
 \|g_\pm(\mathcal D_n)AQ_n\|\le C\sqrt{n+1}\log(n+2).
\]
Thus, denoting the filtered cubic lowering and raising jumps by
$L_{i,3}^-$ and $L_{i,3}^+$,
\[
 \|(L_{i,3}^--a_i^3)Q_n\|
 +\|(L_{i,3}^+)^\dagger Q_n\|
 \le C\sqrt{n+1}\log(n+2).
\]
The second term uses
$(L_{i,3}^+)^\dagger Q_n=f(-\mathcal D_n)AQ_n$.
Taking adjoints and shifting $n$ by three bounds $L_{i,3}^+Q_n$
as well. The same dyadic Fourier argument for $f$ itself gives
$\|L_{i,1}^{\pm}Q_n\|\le C\sqrt{n+1}\log(n+2)$ for the linear
jumps. Consequently, for $\psi\in\cH^{[n]}$,
\begin{align*}
 \langle\psi,C_H\psi\rangle
 &\ge\frac12\sum_i\|a_i^3\psi\|^2
       -C(n+1)\log^2(n+2)\|\psi\|^2,\\
 \sum_i(a_i^\dagger)^3a_i^3
 &=\sum_iN_i(N_i-I)(N_i-2I)
 \ge k^{-2}\mathsf N^3-C(\mathsf N+I)^2.
\end{align*}
The last lower bound uses
$\sum_i n_i^3\ge k^{-2}(\sum_i n_i)^3$ and
$n_i(n_i-1)(n_i-2)=n_i^3-3n_i^2+2n_i$.
For every $\varepsilon>0$,
\[
 (n+1)^2+(n+1)\log^2(n+2)
 \le\varepsilon(n+1)^3+C_\varepsilon.
\]
This absorbs the error into the cubic term. For the upper bound,
the filtered lowering jumps have norm at most $C(n+1)^{3/2}$,
and the remaining jumps satisfy the smaller bounds above. Thus
\[
 c(\mathsf N+I)^3-CI\le C_H\le C(\mathsf N+I)^3.
\]

\emph{Step 4:}
For a bare jump $A^\alpha$, its standard-form gain operator is
\[
 B_\alpha=e^{\beta H/4}L^\alpha e^{-\beta H/4}
 =F_0(\operatorname{ad}_H)A^\alpha,\qquad
 F_\theta(\nu)=\exp\!\left[-\frac{\sqrt{1+(\beta\nu)^2}}4+\theta\nu\right].
\]
Here $\operatorname{ad}_H(X)=[H,X]$ is used only on the finite
particle blocks. Fix $0\le\theta_*<\beta/4$.
Both $F_\theta$ and $\nu F_\theta$ have inverse Fourier transforms
with uniformly bounded $L^1$ norms for $|\theta|\le\theta_*$.
For a cubic lowering jump, the identity in Step~2 gives
\begin{align*}
 B_{\alpha,\theta}Q_n
 &:=(e^{\theta H}B_\alpha e^{-\theta H})Q_n
 =F_\theta(\mathcal D_n)AQ_n=-(\nu F_\theta)(\mathcal D_n)KQ_n
   +F_\theta(\mathcal D_n)RQ_n.
\end{align*}
The Fourier bounds therefore imply
\[
 \sup_{|\theta|\le\theta_*}\|B_{\alpha,\theta}Q_n\|
 \le C_{\theta_*}(\|KQ_n\|+\|RQ_n\|)
 \le C_{\theta_*}\sqrt{n+1}.
\]
For cubic raising jumps this follows by adjunction and
$F_\theta(-\nu)=F_{-\theta}(\nu)$; for linear jumps it follows
directly from their $C\sqrt{n+1}$ bare bound.

The correction $B$ in the statement is essential when hopping is
present. Transporting the full generator, including this correction,
gives on finite particle matrices
\[
 LX=DX+XD-\sum_\alpha B_\alpha X B_\alpha^\dagger,
 \qquad
 D=\frac1{2\cosh((\beta/4)\operatorname{ad}_H)}(C_H).
\]
To check the drift and its sign, set
$M=-iB-C_H/2$ and $x=\beta(E-E')/4$. The definition of $B$ gives
\[
 P_EMP_{E'}=-\frac{1-\tanh x}{2}P_EC_HP_{E'}
 =-\frac{e^{-x}}{2\cosh x}P_EC_HP_{E'}.
\]
Conjugating by $e^{\beta H/4}$ multiplies this block by $e^x$.
Hence $e^{\beta H/4}Me^{-\beta H/4}=-D$, and the two drift terms
in the nonnegative standard-form generator are $DX+XD$.
For $D_\theta=e^{\theta H}De^{-\theta H}$, its Hermitian part is
\[
 \operatorname{Re}D_\theta
 =\frac{\cosh(\theta\operatorname{ad}_H)}
 {2\cosh((\beta/4)\operatorname{ad}_H)}(C_H)
 =\int_\R m_\theta(t)e^{itH}C_He^{-itH}\dd t,
\]
where an elementary Fourier transform gives
\[
 m_\theta(t)=
 \frac{\cos(2\pi\theta/\beta)\cosh(2\pi t/\beta)}
 {\beta\,[\sinh^2(2\pi t/\beta)+\cos^2(2\pi\theta/\beta)]}
 \ge0,\qquad \int_\R m_\theta(t)\dd t=\frac12.
\]
Since $H$ commutes with $\mathsf N$, Step~3 implies, uniformly for
$|\theta|\le\theta_*$,
\[
 \operatorname{Re}D_\theta\ge c(\mathsf N+I)^3-CI,
 \qquad \|D_\theta Q_n\|\le C_{\theta_*}(n+1)^3.
\]
For the norm estimate, use the integrable inverse Fourier transform
of $e^{\theta\nu}/(2\cosh(\beta\nu/4))$ and the upper bound on $C_H$.

\emph{Step 5:}
On Hilbert--Schmidt space set
\[
 \mathcal K_rX=(\mathsf N+I)^rX+X(\mathsf N+I)^r
 \text{ for }r=1,2,3,\qquad
 W_\theta X=e^{\theta\widetilde H}Xe^{\theta\widetilde H}.
\]
The loss operator $X\mapsto DX+XD$ is a self-adjoint direct sum
on the blocks $Q_nXQ_m$, with domain $D(\mathcal K_3)$ and finite
particle matrices as an operator core. The gain operator and its
conjugates are lower order. In fact, each bare row has a fixed
particle-number change $d_\alpha\in\{-3,-1,1,3\}$, and
\[
 W_\theta\bigl(B_\alpha X B_\alpha^\dagger\bigr)
 =B_{\alpha,\theta}(W_\theta X)B_{\alpha,\theta}^\dagger.
\]
Write $X_{nm}=Q_nXQ_m$, with a block understood to be zero if
an index is negative. For one gain term,
\[
 \|B_{\alpha,\theta}X_{nm}B_{\alpha,\theta}^\dagger\|_2
 \le C\sqrt{(n+1)(m+1)}\|X_{nm}\|_2,
\]
and its output lies in block $(n+d_\alpha,m+d_\alpha)$.
For fixed $\alpha$, distinct input blocks have orthogonal outputs.
Squaring and summing gives
\[
 \|B_{\alpha,\theta}XB_{\alpha,\theta}^\dagger\|_2^2
 \le C\sum_{n,m}(n+1)(m+1)\|X_{nm}\|_2^2
 \le C\|\mathcal K_1X\|_2^2.
\]
For the form bound, put $x_{nm}=\|X_{nm}\|_2$ and use
\begin{align*}
 |\langle X,B_{\alpha,\theta}XB_{\alpha,\theta}^\dagger\rangle|
 &\le C\sum_{n,m}\sqrt{(n+1)(m+1)}\,
 x_{nm}x_{n+d_\alpha,m+d_\alpha}\\
 &\le C\sum_{n,m}(n+m+2)x_{nm}^2
 =C\langle X,\mathcal K_1X\rangle.
\end{align*}
The second inequality uses $2ab\le a^2+b^2$ and the comparability
of the weights before and after a shift of at most three.
Summing over $\alpha$ proves
\begin{align*}
 \left\|\sum_\alpha B_{\alpha,\theta}X
 B_{\alpha,\theta}^\dagger\right\|_2
 &\le C_{\theta_*}\|\mathcal K_1X\|_2,\\
 \left|\sum_\alpha\langle X,B_{\alpha,\theta}X
 B_{\alpha,\theta}^\dagger\rangle\right|
 &\le C_{\theta_*}\langle X,\mathcal K_1X\rangle.
\end{align*}
At $\theta=0$ the gain is symmetric because each $B_\alpha$ is
paired with its adjoint. For every $\varepsilon>0$, the scalar
inequality $r\le\varepsilon r^3+C_\varepsilon$ gives
\[
 \|\mathcal K_1X\|_2
 \le\varepsilon\|\mathcal K_3X\|_2+C_\varepsilon\|X\|_2.
\]
Thus the gain has relative operator bound zero with respect to
the loss operator. The self-adjoint
perturbation theorem \cite{Schmdgen2012} therefore identifies the closure on finite
particle matrices with the standard-form generator from Step~1;
the Gibbs vector belongs to its domain by the Gaussian particle-number
tail. Its form domain is $D(\mathcal K_3^{1/2})$. For the weighted form, Step~4 first gives
\begin{align*}
 \operatorname{Re}\langle X,W_\theta LW_\theta^{-1}X\rangle
 &\ge c\langle X,\mathcal K_3X\rangle
       -C\langle X,\mathcal K_1X\rangle-C\|X\|_2^2,\\
 \langle X,\mathcal K_1X\rangle
 &\le\varepsilon\langle X,\mathcal K_3X\rangle
       +C_\varepsilon\|X\|_2^2.
\end{align*}
Choose $\varepsilon$ so that $C\varepsilon<c/2$. We obtain
\[
 \operatorname{Re}\langle X,W_\theta LW_\theta^{-1}X\rangle
 \ge c_*\langle X,\mathcal K_3X\rangle-C_*\|X\|_2^2,
 \qquad 0\le\theta\le\theta_*<\beta/4,
\]
on the finite particle core. The constants are uniform in $\theta$
on this interval.

\emph{Step 6: smoothing at every positive time.}
Fix $s>0$ and $2<p<\infty$, and choose
\[
 \delta_p=\frac12-\frac1p,\qquad
 \theta_*=\frac{\beta\delta_p}{2}<\frac\beta4,\qquad
 \theta(t)=\frac{\theta_*t}{s}\text{ for }0\le t\le s.
\]
Let $\Pi_MX=Q_{\le M}XQ_{\le M}$, with
$Q_{\le M}=\sum_{n=0}^MQ_n$, and let $L_M$ be the form compression
of $L$ to its range. These projections commute with $W_\theta$.
For $U_M(t)=W_{\theta(t)}e^{-tL_M}\Pi_MX$, differentiation
in the finite-dimensional range gives
\[
 U_M'(t)=\frac{\theta_*}{s}
 (\widetilde H U_M(t)+U_M(t)\widetilde H)
 -W_{\theta(t)}L_MW_{\theta(t)}^{-1}U_M(t).
\]
The compression keeps every loss term: it is $\Pi_ML\Pi_M$ on
its range, rather than a new sampler with truncated jumps.
Since $\Pi_M$ commutes with the weight, Step~5 gives
\begin{align*}
 \frac12\frac{\dd}{\dd t}\|U_M(t)\|_2^2
 &\le\frac{\theta_*}{s}
 \langle U_M,\widetilde H U_M+U_M\widetilde H\rangle
 -c_*\langle U_M,\mathcal K_3U_M\rangle+C_*\|U_M\|_2^2\le b_s\|U_M(t)\|_2^2,
\end{align*}
where one may take
\[
 b_s=C_*+2\sup_{r\ge1}
 \left\{\frac{C\theta_*}{s}r^2-c_*r^3\right\}<\infty.
\]
Here $\widetilde H\le C(\mathsf N+I)^2$ is from Step~1.
Integration gives $\|U_M(s)\|_2\le e^{sb_s}\|X\|_2$,
uniformly in $M$. To remove $M$, set $u=(L+I)^{-1}X$ and
$u_M=(L_M+I)^{-1}\Pi_MX$. The variational equations give
\[
 \cE(u-u_M,v)+\langle u-u_M,v\rangle=0
 \quad(v\in\operatorname{Ran}\Pi_M),\qquad
 \|u-u_M\|_{\cE}
 \le\inf_{v\in\operatorname{Ran}\Pi_M}\|u-v\|_{\cE}\longrightarrow0.
\]
Here $\cE(\cdot,\cdot)$ is the sesquilinear form associated with
$\cE$. The limit follows because finite particle matrices are a form
core. Strong resolvent convergence then gives
$e^{-sL_M}\Pi_MX\to e^{-sL}X$ in Hilbert--Schmidt norm
\cite[Theorem~III.4.8]{engel1999one}.
Choose a joint eigenbasis $(\psi_a)_a$ of $\widetilde H$ and
$\mathsf N$, with $\widetilde H\psi_a=\widetilde E_a\psi_a$.
If $y_{ab}^{(M)}$ and $y_{ab}$ are the corresponding coefficients
of these two outputs, Fatou's lemma gives
\[
 \sum_{a,b}e^{2\theta_*(\widetilde E_a+\widetilde E_b)}|y_{ab}|^2
 \le\liminf_M\sum_{a,b}
 e^{2\theta_*(\widetilde E_a+\widetilde E_b)}|y_{ab}^{(M)}|^2
 \le e^{2sb_s}\|X\|_2^2.
\]
In particular, the weighted output is Hilbert--Schmidt and
\[
 \|e^{\theta_*\widetilde H}e^{-sL}X
 e^{\theta_*\widetilde H}\|_2\le e^{sb_s}\|X\|_2.
\]
The bound therefore holds for the full semigroup. Put $\widetilde Z=\Tr(e^{-\beta\widetilde H})$.
Since $\|Y\|_p\le\|Y\|_2$ for $p\ge2$,
\[
 \|e^{-sL}X\|_{p\leftarrow2,\sigma_\beta}
 =\widetilde Z^{\delta_p}
 \|e^{\theta_*\widetilde H}e^{-sL}X
 e^{\theta_*\widetilde H}\|_p
 \le\widetilde Z^{\delta_p}e^{sb_s}\|X\|_2.
\]
This proves the required bounds for every $s>0$ and finite $p>2$.

\emph{Step 7: the kernel and the gap.}
Step~5 gives a quantitative tail estimate:
\[
 \|(I-\Pi_M)Y\|_2^2
 \le\frac{\langle Y,\mathcal K_3Y\rangle}{(M+2)^3}
 \le\frac{\cE(Y)+C_*\|Y\|_2^2}{c_*(M+2)^3}.
\]
Every $\Pi_M$ has finite rank. Thus the form unit ball is relatively
compact in Hilbert--Schmidt space, and $L$ has compact resolvent.
Moreover, the upper loss and gain bounds in Steps~4--5 give the
matching upper form bound, so
\[
 c\mathcal K_3-CI\le L\le C\mathcal K_3+CI.
\]
There are $\binom{n+k}{k}\asymp n^k$ occupation states with at most
$n$ particles. Counting pairs of such states gives
$\#\{\text{eigenvalues of }\mathcal K_3\le E\}\asymp E^{2k/3}$.
The min--max principle yields $1+\lambda_j\asymp j^{3/(2k)}$;
summing $e^{-t\lambda_j}$ gives the asserted heat-trace bounds.
These are matching powers, not exact Weyl constants.

To identify its kernel, put
$J_\alpha(t)=e^{itH}L^\alpha e^{-itH}$. The paired Dirichlet-form
identity \eqref{eq:E-paired-generator}, at $\sigma_E=\infty$, reads
\[
 \cE(Y)=\sum_\alpha\int_\R
 \frac{\|J_\alpha(t)Y-YJ_{\bar\alpha}(t)^\dagger\|_2^2}
 {\beta\cosh(2\pi t/\beta)}\dd t.
\]
We first justify the identity for complex finite particle matrices.
For a self-adjoint $Y$, the difference between the two sides is a
real quadratic polynomial in its matrix entries. It vanishes on
the positive definite matrices by homogeneity and \cref{lem:paired-generator},
and hence on all self-adjoint matrices. One may work in a particle
cutoff larger than the support of $Y$ by three, so all jump images
and loss terms are retained. For complex $Y=Y_1+iY_2$, with
$Y_j=Y_j^\dagger$, put
$\mathcal D_\alpha(t)Y=J_\alpha(t)Y-YJ_{\bar\alpha}(t)^\dagger$.
Adjoint pairing gives
\[
 \mathcal D_\alpha(t)(Y^\dagger)
 =-\bigl(\mathcal D_{\bar\alpha}(t)Y\bigr)^\dagger,
 \qquad
 \cE(Y)=\cE(Y_1)+\cE(Y_2).
\]
The mixed terms on the right also cancel after summing over
$\alpha$. This proves the identity on the complex finite particle core.

The growth bounds from Step~3 give, uniformly in $t$,
\[
 \sum_\alpha\|\mathcal D_\alpha(t)Y\|_2^2
 \le C\langle Y,\mathcal K_3Y\rangle.
\]
Both sides therefore extend continuously to
$D(\mathcal K_3^{1/2})$. Approximation by finite particle matrices
also shows that each $\mathcal D_\alpha(t)Y$ is continuous in $t$.
The integration weight is strictly positive. Hence
\[
 Y\in\ker L\quad\Longrightarrow\quad
 \mathcal D_\alpha(t)Y=0
 \quad\text{for every }\alpha\text{ and }t\in\R.
\]
On each block $Q_n(\cdot)Q_m$, this is a finite trigonometric sum,
because a jump has a fixed particle-number change. Indeed, the finite frequency expansion on such a block is
\[
 \mathcal D_\alpha(t)Y
 =\sum_\nu e^{it\nu}
 \left[f(\nu)A_\nu^\alpha Y-f(-\nu)YA_\nu^\alpha\right].
\]
Distinct exponentials are linearly independent. Since
$f(-\nu)=e^{\beta\nu/2}f(\nu)$ and $f(\nu)>0$, each coefficient
vanishes and
\[
 A_\nu^\alpha Y=e^{\beta\nu/2}YA_\nu^\alpha.
\]
Define $X=\sigma_\beta^{-1/4}Y\sigma_\beta^{-1/4}$ blockwise;
these are finite matrices, so no boundedness of $X$ is assumed.
The modular relation gives, blockwise,
\[
 A_\nu^\alpha\sigma_\beta^{1/4}
 =e^{\beta\nu/4}\sigma_\beta^{1/4}A_\nu^\alpha,
 \qquad
 \sigma_\beta^{1/4}A_\nu^\alpha
 =e^{-\beta\nu/4}A_\nu^\alpha\sigma_\beta^{1/4}.
\]
Substituting $Y=\sigma_\beta^{1/4}X\sigma_\beta^{1/4}$ into the
preceding identity and cancelling the finite invertible weights
therefore gives $[A_\nu^\alpha,X]=0$. Summing the finitely many frequencies in each
particle block, the linear rows yield
\[
 [a_j,X]=[a_j^\dagger,X]=0,\qquad [N_j,X]=0\quad(1\le j\le k).
\]
Writing $X_{\mathbf n\mathbf m}=\langle\mathbf n|X|\mathbf m\rangle$,
the number commutators give
\[
 (n_j-m_j)X_{\mathbf n\mathbf m}=0\quad\text{for all }j.
\]
Thus $X$ is diagonal. If $x_{\mathbf n}=X_{\mathbf n\mathbf n}$,
the creation commutators give
\[
 0=\langle\mathbf n+e_j|[a_j^\dagger,X]|\mathbf n\rangle
 =\sqrt{n_j+1}\,(x_{\mathbf n}-x_{\mathbf n+e_j}).
\]
Every occupation vector is connected to the vacuum by these steps.
Hence $X$ is scalar and $\ker L=\C\sqrt{\sigma_\beta}$.
Compact resolvent now gives a positive gap $\lambda$.

\emph{Step 8: entropy decay and uniqueness among states.}
Apply \cref{thm:KMS-endpoint} with $p=4$ and
$C(s)=\max\{1,\widetilde Z^{1/4}e^{sb_s}\}$ to obtain the claimed
relative-entropy estimate. To check uniqueness also for invariant
states of infinite relative entropy, approximate any state $\tau$ in
trace norm by its normalized finite particle compressions $\rho_M$.
These have finite relative entropy because each particle cutoff
has finite dimension and the Gibbs state is faithful there. Pinsker's
inequality and the entropy estimate give
\[
 \|\Phi_t(\rho_M)-\sigma_\beta\|_1^2
 \le2D(\Phi_t(\rho_M)\|\sigma_\beta)
 \le2e^{-\alpha(s)t}D(\rho_M\|\sigma_\beta)
 \longrightarrow0.
\] If $\Phi_t(\tau)=\tau$, trace-norm contractivity gives
\[
 \|\tau-\sigma_\beta\|_1
 \le\|\tau-\rho_M\|_1+\|\Phi_t(\rho_M)-\sigma_\beta\|_1.
\]
Letting $t\to\infty$ and then $M\to\infty$ proves $\tau=\sigma_\beta$.

\emph{Step 9: the quadratic-jump alternative.}
We record the changes needed when each cubic pair is replaced by a
quadratic pair. Step~1 applies with common energy exponents $1/2$ and
$1$. In Step~2 take
\[
 A=a_i^2,\qquad D_i=b_iI+2\sum_jw_{ij}N_j,\qquad
 c_i=b_i+3U_i+2\mu_i.
\]
Then $[H_0,A]=-A(D_i-c_iI)$ exactly, and
$\|A(z+D_i)^{-1}Q_n\|\le C$ for $\operatorname{Re}z\ge-1$.
For a hopping term the direct commutator is
$[a_r^\dagger a_s,a_i^2]=-2\mathbf1_{\{r=i\}}a_ia_s$.
The resolvent-commutator coefficient is bounded by
\[
 C\frac{(n_i+1)|w_{ir}-w_{is}|\sqrt{(n_r+1)n_s}}
 {D_i(\mathbf n)^2}\le C\sqrt{n+1},
\]
using the same nonnegative-interaction estimate as in Step~2.
Thus its resolvent identity holds with a remainder bounded by
$C\sqrt{n+1}$. Step~3 consequently gives
\[
 \|(L_{i,2}^--a_i^2)Q_n\|+\|L_{i,2}^+Q_n\|
 \le C\sqrt{n+1}\log(n+2).
\]
For every $\varepsilon>0$, expanding the squared lowering norms and
using $2ab\le\varepsilon a^2+\varepsilon^{-1}b^2$ now gives
\[
 C_H\ge(1-\varepsilon)\sum_iN_i^2
 -C_\varepsilon(\mathsf N+I)\log^2(\mathsf N+2I),
 \qquad C_H\le C(\mathsf N+I)^2.
\]
The quadratic part of $H$ is $\frac12\sum_{i,j}w_{ij}N_iN_j$;
hopping and the remaining terms have sector norm $O(n+1)$. Hence
\[
 \widetilde H\le\frac\Lambda2\sum_iN_i^2+C(\mathsf N+I),
 \qquad C_H\ge b\widetilde H-C_bI\quad(0<b<2/\Lambda),
\]
where the smaller logarithmic error is absorbed using
$\widetilde H\ge c(\mathsf N+I)^2$.
Step~4 still gives $\|B_{\alpha,\theta}Q_n\|\le C_{\theta_*}\sqrt{n+1}$.
The positive average defining $\operatorname{Re}D_\theta$ has mass
$1/2$ and leaves $\widetilde H$ invariant. Together with the gain
bound of Step~5, which is infinitesimal relative to the quadratic
energy, it follows that for every $0<c<1/\Lambda$,
\[
 \operatorname{Re}\langle X,W_\theta LW_\theta^{-1}X\rangle
 \ge c\langle X,\widetilde HX+X\widetilde H\rangle
 -C_{c,\theta_*}\|X\|_2^2,
 \qquad 0\le\theta\le\theta_*<\beta/4.
\]
At $\theta=0$ the same relative-bound argument identifies the operator
and form domains as $D(\mathcal K_2)$ and $D(\mathcal K_2^{1/2})$,
with finite particle matrices as a core.
Set $\theta_*=(\beta/2)(1/2-1/p)$ and choose
$\theta_*/s<c<1/\Lambda$. This is possible under the strict time
condition in the statement. For the weighted compressions in Step~6,
\[
 \frac12\frac{\dd}{\dd t}\|U_M(t)\|_2^2
 \le(\theta_*/s-c)
 \langle U_M,\widetilde HU_M+U_M\widetilde H\rangle
 +C_{c,\theta_*}\|U_M(t)\|_2^2.
\]
Integration and the same form-core limit and Fatou argument give
$\|T_s:L^2(\sigma_\beta)\to L^p(\sigma_\beta)\|
\le\widetilde Z^{1/2-1/p}e^{sC_{c,\theta_*}}<\infty$.
The quadratic form domain embeds compactly into Hilbert--Schmidt
space. Step~7's kernel argument uses only the linear rows and applies
unchanged, giving a simple kernel and a positive gap. Step~8 at any
$s>\beta\Lambda/8$ proves entropy decay and uniqueness among states.
\end{proof}

\begin{remark}
For $k=1$, $U_1=2$, and $\mu_1=-1$, the Hamiltonian is $N^2$.
The linear pair alone has no bounded $L^2$-to-$L^4$ endpoint by
\cref{thm:mixed-polynomial}; the quadratic pair gives a finite
waiting time, and the cubic pair gives smoothing at every positive
time. The theorem treats both enlarged families with full hopping. It does not give the same entropy estimate
for the samplers with only linear jumps in \cite{BeckerRouzeSalzmannBose},
nor a rate uniform in the number of modes.
\end{remark}

\subsection{Finite-rank Coulomb truncations}\label{subsec:coulomb}
We next treat a trapped Coulomb system in two or three dimensions.
The interaction is compressed to a finite Hermite subspace, while the
Hamiltonian and sampler act on the full infinite-dimensional space.
For each fixed particle number and cutoff, we prove a bounded
$L^2$-to-$L^4$ map and exponential entropy decay.

\begin{corollary}[Coulomb sampler with a finite-rank interaction]
\label{cor:coulomb-hc}
Let $d\in\{2,3\}$, $n,M\ge1$, $\beta,\eta>0$, and
$\alpha_{n,i,j}\in\R$.  On $\cH_n=L^2((\R^d)^n)$, with
$x=(x_1,\ldots,x_n)$, define
\[
 H_{0,n}=-\Delta+|x|^2,\qquad
 W_n(x)=\sum_{1\le i<j\le n}\alpha_{n,i,j}w_d(x_i-x_j),
 \qquad w_2(y)=-\log|y|,\quad w_3(y)=|y|^{-1}.
\]
Let $P_M$ project onto the first $M$ one-particle Hermite basis
vectors, ordered by energy with a fixed ordering of degeneracies, and
put
\[
 \Pi_M=P_M^{\otimes n},\qquad
 W_{n,M}=\Pi_MW_n\Pi_M,\qquad
 H_{n,M}=H_{0,n}+W_{n,M},\qquad
 \sigma_\beta=\frac{e^{-\beta H_{n,M}}}{\Tr(e^{-\beta H_{n,M}})}.
\]
Consider the coherent Gibbs sampler with bare jumps
$a_{ij}=(x_{ij}+\partial_{x_{ij}})/\sqrt2$ and
$a_{ij}^\dagger=(x_{ij}-\partial_{x_{ij}})/\sqrt2$,
$1\le i\le n$, $1\le j\le d$, and filter
$\widehat f_\eta(\nu)=e^{-\eta\nu^2-\beta\nu/4}$.
Its conservative KMS-symmetric semigroup $\Phi_t$ has
$\sigma_\beta$ as its unique invariant state, and its nonnegative
standard-form generator has compact resolvent and positive gap.
For its Heisenberg semigroup $T_t$, there are finite constants
$s_{n,M}>0$ and $\alpha_{n,M}>0$ such that
\[
 \|T_{s_{n,M}}:L^2(\sigma_\beta)\to L^4(\sigma_\beta)\|\le1,
 \qquad
 D(\Phi_t(\rho)\|\sigma_\beta)
 \le e^{-\alpha_{n,M}t}D(\rho\|\sigma_\beta)
\]
for every state $\rho$ and $t\ge0$.
The conclusion also holds with any bounded self-adjoint
$W=\Pi_MW\Pi_M$ in place of $W_{n,M}$.
The constants may depend on $n,M,\beta,\eta$, and the interaction;
the estimates need not be uniform in the cutoff.
\end{corollary}

The compressed interaction is well defined: both Coulomb kernels
are locally integrable, and Hermite functions decay rapidly, so
\[
 \langle\varphi_{\boldsymbol\ell},W_{n,M}
 \varphi_{\boldsymbol r}\rangle
 =\sum_{i<j}\alpha_{n,i,j}
 \int_{(\R^d)^n}\overline{\varphi_{\boldsymbol\ell}(x)}
 w_d(x_i-x_j)\varphi_{\boldsymbol r}(x)\,\dd x
 \quad(1\le\ell_i,r_i\le M)
\]
is finite, where $\varphi_{\boldsymbol\ell}
=\varphi_{\ell_1}\otimes\cdots\otimes\varphi_{\ell_n}$.
The resulting matrix is self-adjoint, and
$H_{n,M}=H_{0,n}$ on $(I-\Pi_M)\cH_n$.
This finite-energy support lets us prove the $L^2$-to-$L^4$
bound directly. The form comparison in \cref{prop:finite-core-hc}
alone would not give that bound.

\paragraph{\emph{An oscillator estimate.}}
Changing the Hamiltonian produces terms with one unbounded ladder
operator. The next estimate controls them.
For a general $k$-mode space $\mathcal H^{(k)}=\ell^2(\N_0^k)$, put
$\mathsf N=\sum_{j=1}^kN_j$.  In the Coulomb application,
$k=dn$ and $H_{0,n}=dn+2\mathsf N$ under the Hermite identification.
On $\mathscr T_2(\mathcal H^{(k)})$ define the nonnegative operator
\[
 \mathcal KX=\mathsf NX+X\mathsf N,\qquad
 e^{\theta\mathcal K}X=e^{\theta \mathsf N}Xe^{\theta \mathsf N}.
\]
We first use these formulas on the span $\mathscr F$ of the
oscillator matrix units.  For $r_->r_+>0$, set
$c=(r_-+r_+)/2$, $v=\sqrt{r_-r_+}$, and define the free standard-form
qOU operator by
\[
 L_0X=c\mathcal KX+kr_+X
 -v\sum_{j=1}^k(a_jXa_j^\dagger+a_j^\dagger Xa_j).
\]
This is the sum of $k$ one-mode qOU operators.

\begin{lemma}\label{lem:oscillator-smoothing}
Let $L\ge0$ be self-adjoint on $\mathscr T_2(\mathcal H^{(k)})$,
with $\mathscr F$ a form core contained in $D(L)$, and write
$L=L_0+R$ on $\mathscr F$.  Suppose that $R\mathscr F\subset\mathscr F$
and that, for some $\theta_*>0$,
\[
 \kappa_*:=c-v\cosh(2\theta_*)>0,\qquad
 \sup_{0\le\theta\le\theta_*}
 \|e^{\theta\mathcal K}Re^{-\theta\mathcal K}X\|_2
 \le C_R\|(\mathcal K+I)^{1/2}X\|_2
 \quad(X\in\mathscr F).
\]
Then there is a finite $C_*>0$ such that, with
$s_*=2\theta_*/\kappa_*$,
\begin{equation}\label{eq:oscillator-smoothing}
 \|e^{\theta_*\mathcal K}e^{-s_*L}X\|_2
 \le e^{C_*s_*}\|X\|_2,\qquad X\in\mathscr T_2(\mathcal H^{(k)}).
\end{equation}
\end{lemma}
\begin{proof}
Put $R_\theta=e^{\theta\mathcal K}Re^{-\theta\mathcal K}$ on
$\mathscr F$.  The conjugated free operator is
\[
 e^{\theta\mathcal K}L_0e^{-\theta\mathcal K}X
 =c\mathcal KX+kr_+X
 -v\sum_j\bigl(e^{-2\theta}a_jXa_j^\dagger
              +e^{2\theta}a_j^\dagger Xa_j\bigr).
\]
The two gain maps are Hilbert--Schmidt adjoints.  Their common
quadratic-form estimate follows from
\begin{align*}
 2\left|\sum_j\langle X,a_jXa_j^\dagger\rangle\right|
 &=2\left|\sum_j\langle a_j^\dagger X,Xa_j^\dagger\rangle\right|\le\sum_j\bigl(\|a_j^\dagger X\|_2^2+\|Xa_j^\dagger\|_2^2\bigr)
 =\langle X,\mathcal KX\rangle+k\|X\|_2^2.
\end{align*}
Consequently
\begin{align*}
 \operatorname{Re}\langle X,e^{\theta\mathcal K}L_0e^{-\theta\mathcal K}X\rangle
 \ge{}&(c-v\cosh(2\theta))\langle X,\mathcal KX\rangle+k(r_+-v\cosh(2\theta))\|X\|_2^2.
\end{align*}
Hence, uniformly for $0\le\theta\le\theta_*$,
\[
 \operatorname{Re}\langle X,e^{\theta\mathcal K}L_0
 e^{-\theta\mathcal K}X\rangle
 \ge \kappa_*\langle X,\mathcal KX\rangle
 -kv\cosh(2\theta_*)\|X\|_2^2.
\]
For the perturbation, Cauchy--Schwarz and Young's inequality give
\begin{align*}
 |\langle X,R_\theta X\rangle|
 &\le C_R\|X\|_2\|(\mathcal K+I)^{1/2}X\|_2\\
 &\le C_R\|X\|_2\langle X,\mathcal KX\rangle^{1/2}+C_R\|X\|_2^2\\
 &\le\frac{\kappa_*}{2}\langle X,\mathcal KX\rangle
    +\left(C_R+\frac{C_R^2}{2\kappa_*}\right)\|X\|_2^2.
\end{align*}
Thus one may take
$C_*=kv\cosh(2\theta_*)+C_R+C_R^2/(2\kappa_*)$, or any larger
positive constant, to obtain
\begin{equation}\label{eq:weighted-oscillator-coercivity}
 \operatorname{Re}\langle X,e^{\theta\mathcal K}Le^{-\theta\mathcal K}X\rangle
 \ge \frac{\kappa_*}{2}\langle X,\mathcal KX\rangle-C_*\|X\|_2^2.
\end{equation}

We justify the evolution estimate by finite-dimensional approximation.
Let $\mathsf P_J=\mathbf1_{[0,J]}(\mathcal K)$ and let
$L_J=\mathsf P_JL\mathsf P_J$ on its finite-dimensional range.
For $X_J(t)=e^{-tL_J}\mathsf P_JX$ and
$\theta(t)=\kappa_*t/2$, set $U_J(t)=e^{\theta(t)\mathcal K}X_J(t)$.
Since $\mathsf P_J$ commutes with $\mathcal K$,
\cref{eq:weighted-oscillator-coercivity} applies to the compressed
operator with the same constants.  For $0\le t\le s_*$ it yields
\[
 \frac12\frac{\dd}{\dd t}\|U_J(t)\|_2^2
 =\theta'(t)\langle U_J(t),\mathcal KU_J(t)\rangle
 -\operatorname{Re}\langle U_J(t),
 e^{\theta(t)\mathcal K}L_Je^{-\theta(t)\mathcal K}U_J(t)\rangle
 \le C_*\|U_J(t)\|_2^2.
\]
Multiplying by $e^{-2C_*t}$ and integrating gives
\[
 \|U_J(t)\|_2^2\le e^{2C_*t}\|\mathsf P_JX\|_2^2
 \le e^{2C_*t}\|X\|_2^2,\qquad0\le t\le s_*.
\]
To pass to the limit, put $u=(L+I)^{-1}X$ and
$u_J=(L_J+I)^{-1}\mathsf P_JX$.
The two variational equations imply orthogonality of $u-u_J$ to
$\operatorname{Ran}\mathsf P_J$ for the inner product of the form
$\cE+\|\cdot\|_2^2$.  Therefore
\[
 \|u-u_J\|_{\cE}
 \le\inf_{v\in\operatorname{Ran}\mathsf P_J}\|u-v\|_{\cE}
 \longrightarrow0,
\]
because the increasing ranges exhaust the form core $\mathscr F$.
Strong resolvent convergence gives
$e^{-s_*L_J}\mathsf P_JX\to e^{-s_*L}X$ in Hilbert--Schmidt norm
\cite[Theorem~III.4.8]{engel1999one}.
Write the respective matrix coefficients as $y^{(J)}_{\mathbf m\mathbf n}$
and $y_{\mathbf m\mathbf n}$.  Since $\theta(s_*)=\theta_*$, Fatou's lemma yields
\[
 \sum_{\mathbf m,\mathbf n}e^{2\theta_*(|\mathbf m|+|\mathbf n|)}
 |y_{\mathbf m\mathbf n}|^2
 \le\liminf_J\sum_{\mathbf m,\mathbf n}e^{2\theta_*(|\mathbf m|+|\mathbf n|)}
 |y^{(J)}_{\mathbf m\mathbf n}|^2
 \le e^{2C_*s_*}\|X\|_2^2.
\]
The finite left-hand side shows that the output belongs to
$D(e^{\theta_*\mathcal K})$, and its square root proves
\eqref{eq:oscillator-smoothing}.
\end{proof}

\begin{proof}[Proof of \cref{cor:coulomb-hc}]
Write $H=H_{0,n}+W$, $\sigma=\sigma_\beta(H)$, and
$\sigma_0=\sigma_\beta(H_{0,n})$.
Choose $c_0$ so that $\widetilde H=H+c_0I\ge I$. Outside the
finite oscillator block, $\widetilde H$ agrees with $H_{0,n}+c_0I$.
On the finite block all powers of $\widetilde H$ are bounded.
Since a coordinate ladder changes the oscillator number by one,
each bare jump $A$ satisfies
\[
 \|A\widetilde H^{-1/2}\|
 +\|\widetilde H^{1/2}A\widetilde H^{-1}\|<\infty.
\]
The finite Hermite span is a common invariant core, and
$\Tr(e^{-b\widetilde H})<\infty$ for every $b>0$.
These bounds verify the energy and domain assumptions of
\cite[Condition~A and Proposition~2.4]{BeckerRouzeSalzmannToAppearcmp}.
The Gaussian filter satisfies the required KMS and boundedness
conditions. The construction in that paper therefore gives a
conservative KMS-symmetric sampler.
Its standard-form operator $L\ge0$ is obtained by closing the form on
$\mathscr F+\C\sqrt\sigma$.  Here the finite $H$-spectral matrices
span $\mathscr F$, since $W$ acts inside a fixed finite oscillator
block.  We will show that $\mathscr F$ alone is a form core.

\emph{Step 1: finite-energy changes of the coefficients.}
Let $L^\ell$ and $L_0^\ell$ denote the filtered jumps for $H$ and
$H_{0,n}$, respectively, with $\ell$ ranging over the $2k$ bare
ladders.  There is a finite oscillator spectral projection $Q$ such
that
\begin{equation}\label{eq:coulomb-finite-coefficients}
 L^\ell=L_0^\ell+F_\ell,\qquad F_\ell=QF_\ell Q.
\end{equation}
Choose $J_0$ with $\Pi_M\le\mathbf1_{[0,J_0]}(\mathsf N)$ and
take $Q=\mathbf1_{[0,J_0+1]}(\mathsf N)$.
For every bare coordinate ladder $A$,
\[
 (I-Q)A\Pi_M=0,\qquad \Pi_MA(I-Q)=0,
\]
because it changes total occupation by one.  On the complement of
$\Pi_M$, $H$ and $H_{0,n}$ have the same spectral decomposition.
All changed matrix elements of the filtered jump therefore have
both source and target in $Q\cH_n$:
\[
 (I-Q)(L^\ell-L_0^\ell)=0,
 \qquad (L^\ell-L_0^\ell)(I-Q)=0.
\]
This proves \eqref{eq:coulomb-finite-coefficients}.  Diagonalizing
$H$ inside $\Pi_M\cH_n$ does not enlarge these spaces.
We enlarge $Q$ below when further ladder actions occur.

Write the state generator as
$\cL(\rho)=\sum_\ell L^\ell\rho (L^\ell)^\dagger
+M_H\rho+\rho M_H^\dagger$.
We also need the finite-energy property for the drift, including
its correction Hamiltonian.  For the bare ladder $A^\ell$, let
$A^\ell_{\nu,H}$ denote its frequency component relative to $H$.
Its free frequency is $\nu_\ell\in\{-2,2\}$.  The same support argument gives, for each frequency,
\[
 A^\ell_{\nu,H}=\mathbf1_{\{\nu=\nu_\ell\}}A^\ell+G^\ell_\nu,
 \qquad G^\ell_\nu=QG^\ell_\nu Q,
\]
with only finitely many nonzero $G^\ell_\nu$.
Equation \eqref{eq:cross-frequency-generator} at $\sigma_E=\infty$ reads
\[
 M_H=-\sum_{\ell,\nu,\mu}
 \frac{\overline{\widehat f_\eta(\nu)}\widehat f_\eta(\mu)
       e^{\beta(\nu-\mu)/4}}{2\cosh(\beta(\nu-\mu)/4)}
 (A^\ell_{\nu,H})^\dagger A^\ell_{\mu,H}.
\]
Subtracting the free expression leaves the products
\begin{align*}
 (A^\ell_{\nu,H})^\dagger A^\ell_{\mu,H}
 &-\mathbf1_{\{\nu=\nu_\ell\}}\mathbf1_{\{\mu=\nu_\ell\}}
       (A^\ell)^\dagger A^\ell=\mathbf1_{\{\nu=\nu_\ell\}}(A^\ell)^\dagger G^\ell_\mu
   +\mathbf1_{\{\mu=\nu_\ell\}}(G^\ell_\nu)^\dagger A^\ell
   +(G^\ell_\nu)^\dagger G^\ell_\mu.
\end{align*}
Every term has finite-energy support: multiplication by a ladder
adds at most one shell to the support of $G^\ell_\nu$.
After enlarging $Q$ by that shell, we have
$M_H=M_0+F_0$ with $F_0=QF_0Q$.

Put $Z=\Tr(e^{-\beta H})$ and $Z_0=\Tr(e^{-\beta H_{0,n}})$.
Since both Hamiltonians are reduced by $\Pi_M$ and agree outside it,
for every real $a$,
\[
 (I-\Pi_M)\sigma^a\sigma_0^{-a}(I-\Pi_M)
 =(Z_0/Z)^a(I-\Pi_M).
\]
Inside the block the same product is a finite invertible matrix.
Conjugating to standard form still changes only finite-energy
coefficients:
\[
 C_\ell:=\sigma^{-1/4}L^\ell\sigma^{1/4}
 =C_\ell^0+\widetilde F_\ell,\qquad
 D_H:=\sigma^{-1/4}M_H\sigma^{1/4}
 =M_0+\widetilde F_0.
\]
Here every $\widetilde F_\ell$ has finite-energy support; the
normalizing scalar cancels in each conjugation.
KMS symmetry identifies the standard-form operator with
\[
 LX=-\sum_\ell C_\ell XC_\ell^\dagger-D_HX-XD_H^\dagger.
\]
For the free operator the two rates are
$r_\pm=|\widehat f_\eta(\pm2)|^2$, so
$r_+=e^{-2\beta}r_-$ and the formula for $L_0$ above applies.
In particular, subtraction gives the explicit perturbation
\begin{align*}
 RX={}&-\sum_\ell\bigl(C_\ell^0X\widetilde F_\ell^\dagger
 +\widetilde F_\ell X(C_\ell^0)^\dagger
 +\widetilde F_\ell X\widetilde F_\ell^\dagger\bigr)
 -\widetilde F_0X-X\widetilde F_0^\dagger.
\end{align*}
The terms with two $\widetilde F$ factors and the two drift
terms have bounded coefficients.  Each cross
term contains one free ladder and one finite-energy coefficient, so
the only unbounded maps to estimate have the form
\begin{equation}\label{eq:coulomb-mixed-maps}
 X\longmapsto F X a_j^\sharp,\qquad
 X\longmapsto a_j^\sharp X F,
 \qquad F=QFQ,\qquad a_j^\sharp\in\{a_j,a_j^\dagger\}.
\end{equation}

\emph{Step 2: the weighted perturbation estimate.}
Choose $J_Q$ so that $Q\le\mathbf1_{[0,J_Q]}(\mathsf N)$.
For $F=QFQ$ and $\theta\ge0$, spectral calculus gives
\[
 \|e^{\pm\theta \mathsf N}Fe^{\mp\theta \mathsf N}\|_\infty
 \le e^{\theta J_Q}\|F\|_\infty.
\]
These factors are uniformly bounded for
$[0,\beta/4]$.  For the ladder factors one has
\[
 e^{\theta \mathsf N}a_je^{-\theta \mathsf N}=e^{-\theta}a_j,
 \qquad e^{\theta \mathsf N}a_j^\dagger e^{-\theta \mathsf N}
 =e^\theta a_j^\dagger.
\]
For the left multiplications, the ladder relations give
\[
 \|a_jX\|_2^2=\Tr(X^\dagger N_jX),\qquad
 \|a_j^\dagger X\|_2^2=\Tr(X^\dagger(N_j+I)X).
\]
Right multiplication yields the corresponding terms with $XN_j$.
Since $N_j\le \mathsf N$, each norm is at most
$\|(\mathcal K+I)^{1/2}X\|_2$, and hence
\[
 \|a_j^\sharp X\|_2+\|Xa_j^\sharp\|_2
 \le2\|(\mathcal K+I)^{1/2}X\|_2.
\]
For example, with $F_\theta=e^{\theta \mathsf N}Fe^{-\theta \mathsf N}$,
\[
 \|e^{\theta\mathcal K}(F(e^{-\theta\mathcal K}X)a_j^\sharp)\|_2
 \le e^{|\theta|}\|F_\theta\|_\infty
 \|(\mathcal K+I)^{1/2}X\|_2.
\]
Summing the finitely many terms gives
\[
 \sup_{0\le\theta\le\beta/4}
 \|e^{\theta\mathcal K}Re^{-\theta\mathcal K}X\|_2
 \le C_R\|(\mathcal K+I)^{1/2}X\|_2.
\]
At $\theta=0$, the lower and upper free estimates give
\[
 (c-v)\langle X,\mathcal KX\rangle-kv\|X\|_2^2
 \le\langle X,L_0X\rangle
 \le(c+v)\langle X,\mathcal KX\rangle+k(r_++v)\|X\|_2^2.
\]
The bound on $R$ implies, for any $\varepsilon>0$,
\[
 |\langle X,RX\rangle|
 \le\varepsilon\langle X,\mathcal KX\rangle
   +\left(C_R+\frac{C_R^2}{4\varepsilon}\right)\|X\|_2^2.
\]
Choose $\varepsilon=(c-v)/2>0$ and absorb this term into the lower
bound.  After adding a sufficiently large multiple of $\|X\|_2^2$,
there are $c_1,c_2,c_3>0$ such that
\[
 c_1\|(\mathcal K+I)^{1/2}X\|_2^2
 \le\cE(X)+c_2\|X\|_2^2
 \le c_3\|(\mathcal K+I)^{1/2}X\|_2^2.
\]
Hence the closure on $\mathscr F$ has form domain
$D(\mathcal K^{1/2})$.  The Gibbs vector belongs even to
$D(\mathcal K)$: outside the finite block its diagonal coefficients
are $Z^{-1/2}e^{-\beta(k+2|\mathbf n|)/2}$, so
\[
 \|(\mathcal K+I)\sqrt\sigma\|_2^2
 \le C+C\sum_{\mathbf n\in\N_0^k}(1+|\mathbf n|)^2e^{-2\beta|\mathbf n|}
 <\infty.
\]
Every term in the generator therefore maps $\sqrt\sigma$ into
Hilbert--Schmidt space, and invariance
gives $L\sqrt\sigma=0$ in this coefficient formula.  Thus the closed
form on $D(\mathcal K^{1/2})$ also annihilates $\sqrt\sigma$.
For $\mathsf P_J=\mathbf1_{[0,J]}(\mathcal K)$, this gives
\[
 \|(\mathcal K+I)^{1/2}(I-\mathsf P_J)\sqrt\sigma\|_2\longrightarrow0,
\]
so adding $\sqrt\sigma$ does not change the closed form.  This proves that
$\mathscr F$ is a form core for the constructed operator $L$;
the coefficient formula also gives $\mathscr F\subset D(L)$.
The parameters in \cref{lem:oscillator-smoothing} can now be chosen as
\[
 \theta_*=\beta/4,\qquad
 \kappa_*=v\bigl(\cosh\beta-\cosh(\beta/2)\bigr)>0,
 \qquad s_*=\frac{\beta}{2\kappa_*},
\]
because $c/v=\cosh\beta$.  The lemma gives
\eqref{eq:oscillator-smoothing}.

\emph{Step 3: the endpoint and entropy decay.}
The operator
$B_*=\sigma^{-1/8}e^{-\beta \mathsf N/4}$
has a bounded extension.  Since $H_{0,n}=k+2\mathsf N$,
its two blocks are
\[
 B_*=\Pi_MB_*\Pi_M+Z^{1/8}e^{\beta k/8}(I-\Pi_M).
\]
The first block is a finite matrix and the second is scalar.
For $Y\in D(e^{(\beta/4)\mathcal K})$, first on finite matrices and
then by closure,
\begin{align*}
 \|Y\|_{4\leftarrow2,\sigma}
 &=\|\sigma^{-1/8}Y\sigma^{-1/8}\|_4
 =\|B_*(e^{(\beta/4)\mathcal K}Y)B_*^\dagger\|_4\\
 &\le\|B_*\|_\infty^2\|e^{(\beta/4)\mathcal K}Y\|_4
 \le\|B_*\|_\infty^2\|e^{(\beta/4)\mathcal K}Y\|_2.
\end{align*}
Putting $Y=e^{-s_*L}X$ and using the oscillator estimate gives
\[
 \|e^{-s_*L}X\|_{4\leftarrow2,\sigma}
 \le\|B_*\|_\infty^2e^{C_*s_*}\|X\|_2.
\]
It follows that $\|T_{s_*}:L^2(\sigma)\to L^4(\sigma)\|\le C$,
where $C=\max\{1,\|B_*\|_\infty^2e^{C_*s_*}\}<\infty$.
The form comparison also proves compactness.  For
$\mathsf P_J=\mathbf1_{[0,J]}(\mathcal K)$ and integer $J\ge0$,
\[
 \|(I-\mathsf P_J)X\|_2^2
 \le\frac{\|(\mathcal K+I)^{1/2}X\|_2^2}{J+1}
 \le\frac{\cE(X)+c_2\|X\|_2^2}{c_1(J+1)}.
\]
Since $\mathsf P_J$ has finite rank, the form unit ball is relatively
compact in Hilbert--Schmidt space.  This is the compact embedding
criterion for compact resolvent
\cite[Sections~10.3 and 12.1]{Schmdgen2012}.
The nonvanishing-filter kernel argument for finite-rank perturbations
in \cite[Section~3.3]{BeckerRouzeSalzmannCoulomb2026} gives
$\ker L=\C\sqrt\sigma$.  Hence $\lambda=\operatorname{gap}(L)>0$.
By \cref{lem:gap-strictification}, one may take
\[
 s_{n,M}=s_*+\frac{\log C+\frac12\log3}{\lambda}.
\]
Finally \cref{thm:KMS-endpoint} gives the stated entropy estimate,
for example with
$\alpha_{n,M}=(2s_*+(2\log C+1)/\lambda)^{-1}>0$.
\end{proof}

\subsection{A bounded-rate obstruction}
We show that bounded escape rates can prevent hypercontractivity
even when the Gibbs state is strongly confining and the sampler has a
positive spectral gap.

\begin{proposition}
\label{prop:holding-obstruction}
Let $T_t$ be a conservative continuous-time Markov semigroup on a countable
set $\Omega$ with faithful invariant probability $\pi$.  For distinct $x,y\in\Omega$, let $q(x,y)\ge0$ be the transition rate
from $x$ to $y$, and put $q(x)=\sum_{y\ne x}q(x,y)$, the total rate
out of $x$.  Here $L^p(\pi)$ has norm
$\|f\|_{p,\pi}=(\sum_{x\in\Omega}\pi(x)|f(x)|^p)^{1/p}$.  If there are points $x_k$ such
that
\[
 \pi(x_k)\longrightarrow0,
 \qquad \sup_k q(x_k)\le\Lambda<\infty,
\]
then for every $p>2$ and every finite $t>0$,
\[
 \|T_t:L^2(\pi)\to L^p(\pi)\|=\infty.
\]
\end{proposition}
\begin{proof}
Let $\mathbf1_{\{x_k\}}$ be the indicator of the singleton $\{x_k\}$
and set $f_k=\pi(x_k)^{-1/2}\mathbf1_{\{x_k\}}$, so
$\|f_k\|_{2,\pi}=1$.  The event that the chain makes no jump before time
$t$ has probability at least $e^{-\Lambda t}$ when it starts from
$x_k$.  Hence
\[
 (T_tf_k)(x_k)\ge e^{-\Lambda t}\pi(x_k)^{-1/2}
\]
and therefore $\|T_tf_k\|_{p,\pi}
 \ge e^{-\Lambda t}\pi(x_k)^{1/p-1/2}\longrightarrow\infty.$
\end{proof}

Consider now the one-mode periodic charging model on $\ell^2(\Z)$,
with charge basis $(e_m)_{m\in\Z}$.  Define the charge operator
$\mathsf n e_m=me_m$ and the bilateral shift $Se_m=e_{m+1}$.  Let
$H=h(\mathsf n)$, where $h$ is a nonconstant even polynomial
($h(-m)=h(m)$) with positive leading coefficient, and use bare jumps
$S,S^\dagger,\mathsf n$.  The charge operator $\mathsf n$ is distinct
from the particle count $n$ of the preceding application.
The charge row alone generates the conservative dephasing semigroup
\[
 (\Phi_t^{\rm charge}(\rho))_{mk}
 =e^{-t|\fhat(0)|^2(m-k)^2/2}\rho_{mk}.
\]
The two filtered shifts are bounded, and their loss operators are
diagonal, so the correction Hamiltonian is zero. Their bounded GKLS
generator perturbs this dephasing generator. The bounded perturbation
product formula gives a conservative completely positive semigroup;
the KMS pairing of the shifts makes the Gibbs state invariant.
For a diagonal observable $g(\mathsf n)$, with $g:\Z\to\C$, the charge
row vanishes.  The invariant diagonal sector is the birth--death chain
with generator $Q$ given by
\[
 (Qg)(m)=r_m^+\bigl(g(m+1)-g(m)\bigr)
       +r_m^-\bigl(g(m-1)-g(m)\bigr),\qquad
 r_m^\pm=|\fhat(h(m\pm1)-h(m))|^2.
\]
The filter \eqref{eq:filter} gives $0\le r_m^\pm\le1$.
Its invariant probability is
\[
 Z_{\rm ch}=\sum_{m\in\Z}e^{-\beta h(m)},\qquad
 \pi_m=Z_{\rm ch}^{-1}e^{-\beta h(m)},\qquad
 \sigma_{\rm ch}=\sum_{m\in\Z}\pi_m|e_m\rangle\langle e_m|.
\]
Since diagonal observables satisfy
$\|g(\mathsf n)\|_{p,\sigma_{\rm ch}}=\|g\|_{p,\pi}$, the classical
norm bound is a lower bound for the quantum norm.  One has $\pi_m\to0$ as
$|m|\to\infty$, while the holding rates satisfy
$r_m^++r_m^-\le2$.  Thus \cref{prop:holding-obstruction} shows that
this invariant diagonal sector, and hence the full quantum semigroup, has
no finite-time $L^2$-to-$L^p$ endpoint for any $p>2$.  This gives a
periodic-charge counterpart of the $m>d_*$ obstruction above.  The sampler nevertheless has a positive gap. We prove this
separately on the diagonal and off-diagonal sectors.
After the unitary map $g\mapsto(\sqrt{\pi_m}g(m))_m$, the negative
diagonal generator is the Jacobi operator with diagonal
$r_m^++r_m^-$ and off-diagonal entries
$-\sqrt{r_m^+r_{m+1}^-}$.
Since $h$ is an even confining polynomial,
\[
 r_m^++r_m^-\longrightarrow1,\qquad
 \sqrt{r_m^+r_{m+1}^-}\longrightarrow0
 \qquad(|m|\to\infty).
\]
Thus the transformed operator is $I+K$, where $K$ is compact:
its three diagonals vanish at infinity and are norm limits of their
finite truncations.  In particular its essential spectrum is $\{1\}$.
The reversible classical form is
\[
 \langle g,-Qg\rangle_{L^2(\pi)}
 =\sum_{m\in\Z}\pi_mr_m^+|g(m+1)-g(m)|^2.
\]
Every coefficient is positive.  Thus zero energy forces
$g(m+1)=g(m)$ for every $m$, so $g$ is constant.  After the unitary
change of variables the kernel is $\C(\sqrt{\pi_m})_m$.
Since zero lies outside the essential spectrum $\{1\}$, it is an
isolated simple eigenvalue, and
\[
 \lambda_{\rm diag}=\inf(\operatorname{spec}(-Q)\setminus\{0\})>0.
\]

The diagonal and off-diagonal Hilbert--Schmidt subspaces reduce the
standard-form generator: each shift gain moves the two charge indices
by the same amount, and the loss operators are diagonal.
The charge row contributes the nonnegative form
\[
 \cE_{\rm charge}(Y)=\frac{|\fhat(0)|^2}{2}
 \sum_{m,k\in\Z}(m-k)^2|Y_{mk}|^2
 \ge\frac{e^{-1/2}}2\|Y_{\rm off}\|_2^2.
\]
The other rows have a nonnegative total form.  Consequently the full
sampler is primitive and has gap at least
$\min\{\lambda_{\rm diag},e^{-1/2}/2\}>0$.
Thus this model has a positive gap but no bounded $L^2$-to-$L^p$
map at any finite time, for any $p>2$.

\bibliographystyle{unsrtnat}
\bibliography{refHC}

\end{document}